%% file: cs_abc_paper.tex
\documentclass[11pt,oneside]{article}

\RequirePackage{amsthm,amsmath}
\RequirePackage[colorlinks,citecolor=blue,urlcolor=blue]{hyperref}

\usepackage[natbib,authordate, backend=biber]{biblatex-chicago} 
\usepackage{amsfonts}
\usepackage{amssymb}
\usepackage{geometry}
\usepackage{graphicx}
\usepackage{color}
\usepackage{multirow}
\usepackage[verbose]{wrapfig}
\usepackage{esint}
\usepackage{eucal} 
\usepackage{bbm} 
\usepackage{fixltx2e} 
\usepackage{mathrsfs} 
\usepackage{subfig}
\usepackage{multicol}
\usepackage{tabularx}
\usepackage{booktabs}
\usepackage{caption} 
\usepackage{algorithmicx, algpseudocode, algorithm}
\usepackage{tikz} 
\usepackage{enumerate} 
\usepackage{authblk} 
\usepackage{bm}

\usepackage[page,header]{appendix} 
\usepackage{minitoc} 

\input{Symbols_pream.tex}

\theoremstyle{plain}
\newtheorem{thm}{Theorem}[section]

\newtheorem{lem}{Lemma}
\newtheorem{remark}{Remark}

\newtheorem{cor}{Corollary}[section]
\theoremstyle{definition}

\makeatletter
\let\orgdescriptionlabel\descriptionlabel
\renewcommand*{\descriptionlabel}[1]{%
  \let\orglabel\label
  \let\label\@gobble
  \phantomsection
  \edef\@currentlabel{#1}%
  \let\label\orglabel
  \orgdescriptionlabel{#1}%
}
\makeatother

\makeatother


\begin{document}

\doparttoc 
\faketableofcontents 

\title{Optimal combination of composite likelihoods using approximate Bayesian computation with application to state-space models}
\author[1]{Wentao Li}
\author[2]{Rosabeth White}
\author[3]{Dennis Prangle}
\affil[1]{Department of Mathematics, University of Manchester}
\affil[2]{School of Mathematics, Statistics and Physics, Newcastle University}
\affil[3]{School of Mathematics, University of Bristol}
\date{}

\maketitle

\begin{abstract}
Composite likelihood methods can provide approximate inference when the full likelihood is intractable but marginal likelihood functions for small subsets of the data can be evaluated relatively easily. Composite likelihoods have been successfully applied in complex models such as max-stable models and state-space models. However, wider application is limited by two issues. First, the choice of composite likelihood weights can have a significant impact on the information efficiency and is currently an open question. Second, calibrated Bayesian inference with a composite likelihood requires curvature adjustment which is difficult for data without repeated measurements. This work shows that approximate Bayesian computation (ABC) can properly address these two issues by using multiple composite score functions as the summary statistic. First, it is shown that the summary-based posterior distribution gives optimal Godambe information among a wide class of estimators defined by estimating functions, including any linear combination of the component composite scores. Second, to make ABC computationally feasible for models where marginal likelihoods have no closed form, a novel approach is proposed to estimate marginal scores of all pseudo datasets using a Monte Carlo sample with size $N$. Sufficient conditions are given for the additional noise to be negligible with $N$ fixed as $n\rightarrow\infty$ where $n$ is the data size. So the computational cost is $O(n)$, much lower than the typical cost $O(n^{2})$ of pseudo-marginal methods needed to control the variance of the likelihood estimator in the Markov chain Monte Carlo (MCMC). Third, the asymptotic behaviour of ABC with summary statistics with heterogeneous convergence rates is studied, and an adaptive scheme to choose the component composite scores is proposed. This approach provides an off-line static parameter estimator for state-space models by optimally combining marginal information using model simulations. Numerical studies show that the new method significantly outperforms existing Bayesian composite likelihood methods and achieves similar statistical efficiency to particle MCMC using the full likelihood with adaptively combined composite scores.
\end{abstract}

Keywords: Approximate Bayesian computation; Composite scores; Unbiased estimating equation; Partial Information; Adaptive summary statistic.

\section{Introduction}\label{sec:intro}

In many applications the likelihood function is difficult to evaluate due to complicated data correlation or stochastic latent structure. If likelihoods for small subsets of the data can be evaluated easily, the composite likelihood method can be used to provide a useful class of approximate likelihoods by combining low-dimensional marginal likelihoods. Let $Y_{obs}$ be a dataset and $\left\{A_j, j=1, \dots, n_{c l}\right\}$ be a set of marginal or conditional events, a composite likelihood is a weighted product of valid likelihood terms
\begin{align}
c L(Y_{obs};\theta):=\prod_{j=1}^{n_{c l}} f\left(A_j ; \theta\right)^{w_j} \label{eq:composite_likelihood},
\end{align}
where $\theta$ is the unknown parameter, $f(A_j ; \theta)$ are the marginal likelihoods, also referred to as the sub-likelihoods, and $\{w_j\}_{j=1}^{n_{c l}}$ is a set of non-negative weights \citep{lindsay1988composite}. For example, the popular \textit{pairwise likelihood} is $pL(Y_{obs};\theta):=\prod_{i<j } f(y_i, y_j ; \theta)^{w_{ij}}$. Composite likelihoods have seen successful applications to max-stable models, state-space models, and generalised linear mixed models, amongst others   \citep{varin2008pairwise, varin2011overview, pauli2011bayesian, hui2018sparse}. 

The weights in a composite likelihood can have a significant impact. For clustered data, \citet{kuk2000pairwise} show that, for the pairwise likelihood, using equal weights when cluster sizes are unequal can lead to poor parameter estimates. They suggest weighting schemes where the weights depend on the cluster sizes. \citet{joe2009weighting} further suggest that weights should also depend on the degree of dependence within clusters. In time series analysis, \citet{davis2011comments} describe the $k$-th order consecutive pairwise likelihood. This is a binary weighted pairwise likelihood such that the sub-likelihoods for pairs of observations with a time lag greater than $k$ are given weight $0$. This approach is shown to improve estimation when most of the information about the parameter is contained in the close pairs because including all pairs can skew this information. Similar weighting schemes have been applied in the spatial setting to exclude sub-likelihoods for pairs with a spatial distance above some threshold \citep{bevilacqua2012estimating}. Instead of choosing a weight for each individual sub-likelihood, weights are often chosen for several candidate composite likelihoods and these components are combined to form the composite likelihood on which the parameter inference is based. For example, \citet{kenne2014combined} propose a weighted sum of the onewise likelihood, which considers only the univariate marginal densities, and the pairwise likelihood. The choice of weighting scheme is often a question for two reasons. First, whether to include onewise, pairwise, or higher-order terms depends on the correlation structure of the data, which is unknown in advance. Second, introducing more terms into the composite likelihood does not necessarily improve it and may even lead to worse performance \citep{xu2016note,davis2011comments}. This counter-intuitive behaviour makes it difficult to apply composite likelihood methods more widely. 

From the theory of unbiased estimating equations, optimal weights of composite components can be obtained by projecting the score function onto a linear manifold with dimension $p$ in the linear space of component sub-scores \citep{heyde1997quasi}. Specifically, assume $Y_{obs}$ follows a model with the score function $s(Y_{obs};\theta)$ where $\theta\in\mathcal{P}\subset\mathbb{R}^p$, and the true parameter is $\theta_0$. For a function $h(Y,\theta)$ depending on $\theta$, denote $h(Y,\theta_{0})$ by $h(Y)$ for convenience. Suppose a $d$-dimension estimating function $h(Y;\theta)$ is given, and $h(Y;\theta)$ is differential in $\theta$ and unbiased, i.e.~$E_{\theta}\{h(Y;\theta)\}=0$ $\forall\theta\in\mathcal{P}$, and $d>p$. Consider the following set of estimating functions,
\[
\mathcal{H}_h:=\{W(\theta)h(Y_{obs};\theta):W\in\mathbb{R}^{p\times d}\}.
\]
The optimal $W$ is 
\[
W^{*}_h(\theta)=E_{\theta}\{\nabla_{\theta}h(Y;\theta)\}^{T}Var\{h(Y;\theta)\}^{-1},
\]
in the sense of minimising the square loss $E_{\theta}\left[\{s(Y;\theta)-W(\theta)h(Y;\theta)\}\{s(Y;\theta)-W(\theta)h(Y;\theta)\}^T\right]$ and that the solution of $W^{*}_h(\theta)h(Y_{obs};\theta)$ achieves the maximal Godambe information matrix among all those from $\mathcal{H}_h$ \citep[Chapter 2]{heyde1997quasi}. Denote the sub-scores in $cL(Y_{obs};\theta)$ by $s(A_j;\theta):=\nabla_{\theta}\log f(A_j;\theta)$, $j=1,\dots,n_{cl}$, and the vector stacking up all sub-scores by $s(A_{1:n_{cl}};\theta)$. \citet{lindsay2011issues} note that since each sub-likelihood term is a valid likelihood, the sub-scores are unbiased estimating functions, and the combination $W_s^*(\theta)s(A_{1:n_{cl}};\theta)$ achieves higher efficiency than that of $cL(Y_{obs};\theta)$ with its optimal weights. However it is numerically challenging to calculate $W_s^*(\theta)$, since the cost of inverting the variance matrix within is $O\{(pn_{cl})^3\}$ where $n_{cl}$ is usually very large. 
	
It is more feasible to optimally combine multiple composite likelihoods than all sub-likelihoods. Consider $K$ composite log-likelihood components, denoted by $cl_{k}(Y_{obs};\theta)$, $k=1,\dots,K$, and the composite scores $\bigtriangledown_{\theta}cl_{k}(Y_{obs};\theta)$, denoted by $cs_{k}(Y_{obs};\theta)$. Let $cs_{1:K}(Y_{obs};\theta)$ be the function stacking all $cs_{k}(Y_{obs};\theta)$. Since $cs_{k}(Y_{obs};\theta)$ are unbiased estimating functions, $W^{*}_{cs}(\theta)cs_{1:K}(Y_{obs},\theta)$ achieves higher efficiency than any individual composite score. When moments in $W^{*}_{cs}(\theta)$ have no closed forms, they can be estimated using  simulated realisations from the model \citep{deng2014combined}. With this formula, \citet{kenne2014combined} propose an optimal combination of the onewise and pairwise composite likelihoods. However, there are two difficulties using the above optimal combination. First, it is computationally and numerically difficult to invert the covariance matrix in $W^{*}_{cs}(\theta)$ when $K$ is large or $cs_{1:K}(Y_{obs};\theta)$ is highly correlated among its dimensions \citep{li2018hybrid}. Second, there may not be a composite likelihood with the optimal combination as its score function, which prohibits some likelihood-based inference, such as a likelihood-ratio test or a quasi-posterior distribution where the composite likelihood substitutes the full likelihood. See \citet{lindsay2011issues} for more discussions.


A recent method to perform Bayesian inference without the composite likelihood function is to approximate the summary-based posterior distribution $\pi\{\theta\mid cs(Y_{obs};\theta_{0})\}$ using ABC methods, where $cs(Y_{obs};\theta)$ is the composite score \citep{ruli2016approximate,li2018convergence}. The unknown $\theta_{0}$ in the summary statistic can be estimated by the maximum composite likelihood estimator, denoted by $\widehat{\theta}_{cl}$. This can be justified in two ways. First, if $cs(Y_{obs};\theta_{0})$ satisfies a central limit theorem and some regularity conditions, $\pi\{\theta\mid cs(Y_{obs};\theta_{0})\}$ satisfies the Bernstein-von Mises theorem and is asymptotically calibrated, in the sense that the variance of the posterior mean is equal to the posterior variance in the first order \citep{li2018convergence}. In contrast, the quasi-posterior distribution, which replaces the full likelihood with the composite likelihood, underestimates the posterior uncertainty \citep{pauli2011bayesian}. Calibration methods that adjust the curvature of the composite score function to correct the asymptotic properties require estimation of the covariance matrix of the composite score function \citep{pauli2011bayesian}. This can be difficult for dependent data and lead to overestimated posterior uncertainty due to the uncertainty in the variance estimation \citep{ruli2016approximate, ribatet2012bayesian}. Second, $\pi\{\theta\mid cs(Y_{obs};\theta_{0})\}$ achieves the Godambe information matrix, same efficiency as the maximum composite likelihood estimator  \citep{li2018convergence}. This method is successfully applied in \citet{ruli2016approximate} to the max-stable model, and numerical studies show that it provides a better uncertainty estimator than the calibrated quasi-posterior distribution.

Contributions of this work are threefold. First, in Section \ref{sec:abc_scores} we propose to use the posterior distribution conditional on component composite scores to optimally combine multiple composite likelihoods. It shows that the optimal information matrix in $\mathcal{H}_h$ is achieved. This method avoids estimating moments of $s_{K}(Y_{obs})$, hence reducing the bias and inflation of uncertainty compared to the two-step procedure, and is more feasible when the component scores are strongly correlated. The optimal combination implies that the information efficiency always increases when more composite components are included. Second, in Section \ref{sec:adaptiveABC} a novel adaptive ABC scheme is proposed to adaptively select the composite components to be included. It balances the trade-off between the statistical information and the computational efficiency of ABC algorithms when more composite components are included. The scheme is supported by the asymptotic properties of the ABC posterior distribution when the summary statistics converge in heterogeneous rates. Third, when the composite scores are analytically unavailable due to some latent variable $X$, in Section \ref{sec:abc_est_scores} we propose to use $N$ simulations of $X$ to estimate all simulated composite scores in the ABC algorithm, which gives a noisy version of the summary statistics. Asymptotic results show the estimation error regarding $N$ is negligible compared to the uncertainty of the summary statistics regarding the data size $n$, and the rate of $N$ does not depend on $n$. This implies that the computational complexity is $O(n)$, in contrast to $O(n^{2})$ of particle MCMC, when applied to the inference of the static parameters in state-space models. 

\section{ABC combining multiple composite scores}\label{sub:abc_est_scores}

\subsection{Notations}
For a function $h(Y,\theta)$ depending on $\theta$, denote $h(Y,\theta_{0})$ by $h(Y)$ for convenience. For example, $cs_{1:K}(Y;\theta_0)$ is denoted by $cs_{1:K}(Y)$, $f(Y,X;\theta_{0})$ by $f(Y,X)$ and $f(Y\mid X,\theta_{0})$ by $f(Y\mid X)$ for some joint and conditional densities. The subscripts of expectation and variance operators, i.e.~$E(\cdot)$ and $Var(\cdot)$, indicate their probability measures to account for all randomness and are omitted for obvious cases. To avoid confusion, $\theta$ denotes the parameter with which $Y$ is generated, and $\beta$ denotes the parameter in estimating functions. For two square matrices $A$ and $B$, $A>B$ denotes that $A-B$ is positive definite. 
 
 Suppose the observed dataset $Y_{obs}$ contains $n$ observations $\{Y_{obs,1},\dots,Y_{obs,n}\}$ and follows the parametric density $f(\cdot;\theta_0)$, where $\theta_0\in\mathbb{R}^{p}$. The unknown data generating parameter value $\theta_{0}$ is of interest. Denote the prior distribution of the parameter by $\pi(\theta)$ and its support in $\mathbb{R}^{p}$ by $\mathcal{P}$. Let $\mathcal{Y}$ be the union of supports of $f(\cdot;\theta)$ over $\theta\in\mathcal{P}$. The Godambe information of an estimating equation $h(Y_{obs};\beta)=0$ is defined as 
\[
I_{G}(h):=c_{n}^2E\{\nabla_{\beta}h(Y_{obs};\beta)\}Var\{h(Y_{obs};\beta)\}^{-1}E\{\nabla_{\beta}h(Y_{obs};\beta)\}^{T}\mid_{\beta=\theta_{0}},
\]
for some scaling rate $c_{n}$. The inverse of $I_{G}(h)$ is the asymptotic covariance matrix of the solution of $h(Y_{obs};\beta)=0$ under broad conditions \citep{godambe1987quasi,lindsay1988composite}. On the other hand, for $Y\sim f(\cdot;\theta)$, if the summary statistic $h(Y)$ satisfies a central limit theorem with rate $c_{n}$, asymptotic mean $\mu_h(\theta)$ and covariance matrix $c_n^{-2}A_h(\theta)$, the asymptotic variance of the summary-based posterior density $\pi\{\theta\mid h(Y_{obs})\}$ is the inverse of 
\[
I(h):=c_n^2\nabla\mu_h(\theta_0)A_h(\theta_0)^{-1}\nabla\mu_h(\theta_0)^T,
\]
under conditions given by \citet{li2018convergence}.  

\subsection{Batching sub-likelihoods}
 
Given the composite likelihood \eqref{eq:composite_likelihood}, the vector of stacked sub-scores $s(A_{1:n_{cl}};\theta)$ has the dimension $pn_{cl}$ and is usually not suitable as the summary statistic in ABC due to its high dimensionality. In many cases the dimension will be higher than the dimension of the trivial sufficient statistic $s=Y$. Therefore, we propose the following batching method to reduce the number of component composite scores. 
	
For $Y\sim f(\cdot;\theta)$, suppose that the sub-likelihoods $f(A_j;\theta)$ in $cL(Y;\theta)$, $j=1,\dots,n_{cl}$, can be partitioned into $K\ll n_{cl}$ sets such that in the $k$th set, the marginal events $A_j$ are identically distributed with the marginal density $f_k(\cdot;\theta)$. Denote the subsets of $Y$ in the $k$th set by $Y_{kj}$, $j=1,\dots,n_k$, and let $\mathcal{Y}_k$ be the union of supports of $f_k(\cdot;\theta)$ over $\theta\in\mathcal{P}$. The batched composite log-likelihood formed by events in the $k$th set is 
$$cl_{k}(Y;\theta):=\frac{1}{n_k}\sum_{j=1}^{n_k}w_{kj}\log f_{k}(Y_{kj};\theta),$$
with scaling $n_k^{-1}$, where $w_{kj}$ are the corresponding weights in $cL(Y;\theta)$, $j=1,\dots,n_k$. Let $cs_{k}(Y;\theta)$ be the scaled composite score function $\nabla_{\theta}cl_{k}(Y;\theta)$ and $cs_{1:K}(Y;\theta)$ be the function stacking all $cs_{k}(Y;\theta)$, $k=1,\dots,K$.
	
The batching method is often used to formulate candidate composite likelihoods. For example, consider a times series $\{y_1,\dots,y_n\}$ following a stationary process. Let $\{Y_{kj}\}_{j=1}^{n-k}=\{(y_j,y_{j+k})\}_{j=1}^{n-k}$ be the set of pairs with lag $k$, $k=1,\dots,n-1$. Batching the pairwise sub-scores in the pairwise likelihood gives the scaled \textit{lag-$k$ pairwise log-likelihood},
\begin{align}
pl_k(Y;\theta):=\frac{1}{n_k}\sum_{j=1}^{n_k}\log f(y_j,y_{j+k};\theta), \label{eq:lag_k_batched} 
\end{align}
where $n_k=n-k$. The \textit{$m$th order summed pairwise log-likelihood} is defined as 
$$
pl_{{\rm equal};m}(Y;\theta):=\sum_{k=1}^m n_kpl_k(Y;\theta),
$$
which is the equally weighted sum of the lag-$1$ to lag-$m$ pairwise likelihoods. When $m=n-1$, it gives the pairwise log-likelihood $pl(Y;\theta)$. These composite likelihoods are discussed and compared for different $m$ in \citet{davis2011comments}. As another example, consider a stationary zero-mean spatial Gaussian random field $\{Z(s)\in\mathbb{R}:s\in\mathbb{R}^2\}$ where $s$ denotes the location coordinates. Let $d_1 < d_2 <\dots <d_{K}$ be the unique distances between pairs of observation sites, and let $\{Y_{kj}\}_{j=1}^{n_k}=\{(Z(s_{i_1}),Z(s_{i_2})): \|s_{i_1}-s_{i_2}\| = d_k)\}$. Then, the distance-$k$ pairwise log-likelihood $n_kcl_{k}(Y;\theta)$ is the composite log-likelihood for pairs of observations with distance $k$. The equally weighted sums for pairs below some threshold distance are considered in \citet{apanasovich2008aberrant} and \citet{bevilacqua2012estimating}.
	
In addition to reducing the dimension, another advantage of the batching method is that the convergence of the posterior distribution conditional on $cs_{1:m}(Y_{obs};\theta)$ can be based on the weak convergence of $cs_{1:m}(Y_{obs};\theta)$, shown in the next subsection, while there is no such theoretical support when using the stacked sub-scores.

\subsection{Posterior distribution conditional on multiple estimating functions} \label{sec:abc_scores}

Consider using the summary-based posterior distribution to combine $cs_{k}(Y_{obs};\theta)$, $k=1,\dots,K$, where $K$ is a fixed number and $n_k\rightarrow\infty$ as $n\rightarrow\infty$. 
More generally, consider an unbiased estimating function $h(Y,\beta)$ with dimension $d>p$. Assume the following conditions parallel to those in \citet{li2018convergence}.
\begin{description}
\item [{(C1)\label{cond:prior}}] There exists some $\delta_{0}>0$ such that $\mathcal{P}_{0}=\{\theta:|\theta-\theta_{0}|<\delta_{0}\}\subset\mathcal{P}$, $\pi(\theta)\in C^{2}(\mathcal{P}_{0})$ and $\pi(\theta_{0})>0$. 
\item [{(C2)\label{cond:sum_conv}}] There exists a $d$-dimensional vector $\mu_h(\theta,\beta)$, a $d\times d$ matrix $A_h(\theta,\beta)$ and a sequence $a_{n}$, satisfying $a_{n}\rightarrow\infty$ as $n\rightarrow\infty$,
such that for any $\theta\in\mathcal{P}_{0}$ and $\beta$ in a neighbourhood of $\theta_{0}$, 
\[
a_{n}\{h(Y;\beta)-\mu_h(\theta,\beta)\}\xrightarrow{}N\{0,A_h(\theta,\beta)\},\mbox{ as }n\rightarrow\infty,
\]
in distribution. Also assume that $h(Y_{obs})\rightarrow \mu(\theta_0)$ in probability. Furthermore, (i) $\mu_h(\theta)$ and $A_h(\theta)\in C^{1}(\mathcal{P}_{0})$, and $A_h(\theta)$ is positive definite for any $\theta\in\mathcal{P}_{0}$; (ii) $E_{\theta_0}\{\nabla_{\beta} h(Y)\}$ has full rank when $n$ is large enough.  
\item [{(C3)\label{cond:limit_moments}}] It holds that $A(\theta)=\lim_{n\rightarrow\infty}a_{n}^{2}Var_{\theta}\{h(Y)\}$ and $\mu(\theta)=\lim_{n\rightarrow\infty}E_{\theta}\{h(Y)\}$ for $\theta\in\mathcal{P}_{0}$.
\end{description}

The following theorem shows that $\pi\{\theta\mid h(Y_{obs})\}$ optimally combines information in all dimensions of $h(Y_{obs})$ in the sense of achieving the maximal Godambe information of the estimating functions in $\mathcal{H}_h$. 
\begin{thm} \label{thm:optimal_comb}
If \ref{cond:prior}-\ref{cond:limit_moments} and \citet[Conditions 4 and 5]{li2018convergence} are satisfied, the asymptotic variance of $\pi\{\theta\mid h(Y_{obs})\}$ is $a_n^{-2}\lim_{n\rightarrow\infty}a_n^2I_{G}(W^{*}_hh)^{-1}$, where $c_n=a_n^{-1}$. 
\end{thm}
The proof relies mainly on the unbiasedness property of the estimating function to show that, $\nabla_{\theta}E_{\theta}\{h(Y)\}=E_{\theta}\{\nabla_{\beta}h(Y)\}$ by differentiation under the expectation, the value of $\mu(\theta)$ at $\theta_{0}$ is unique and $I(s_K)$ has full rank. The composite score functions $cs_k(Y,\beta)$ are unbiased, hence Theorem \ref{thm:optimal_comb} holds if the conditions are satisfied for $cs_{1:K}(Y;\beta)$. Theorem \ref{thm:optimal_comb} says that the multiple composite scores (MCS)-based posterior distribution $\pi\{\theta| cs_{1:K}(Y_{obs})\}$ optimally combines the $K$ composite score components, and the Godambe information always increases as more composite scores are included. Some discussions of the conditions are given in the Supplement.

Figure \ref{fig:equicorr} illustrates the optimality of the MCS-based posterior distribution using a simple example. Suppose $Y_1,\dots,Y_n\in\mathbb{R}^q$ independently follow the equi-correlated Gaussian distribution $N_q(0, \sigma^2\Sigma)$ where $\Sigma$ has the off-diagonal elements $0.5$ and the diagonal elements $1$, and $\sigma^2$ is unknown. Both the onewise log-likelihood $ul(Y;\sigma^2)$ which combines likelihoods of individual observations and $pl(Y;\sigma^2)$ contain information about $\sigma^2$. \citet{kenne2014combined} proposes to combine the two by $cs_w(Y;\sigma^2)=2ps(Y;\sigma^2)-w(q-1)us(Y;\sigma^2)$, where $w$ is a constant to be chosen. Figure \ref{fig:equicorr} compares the posterior distributions based on the MCS that stacks $us(Y)$ and $ps(Y)$ and on the individual composite score (ICS) using $cs_w(Y)$. It shows that the approximation error of the ICS-based posterior distribution heavily depends on the composite weights and $w$ needs to be carefully chosen. In contrast, the MCS-based posterior distribution well approximates the full posterior distribution.

\begin{figure}[!ht]
\centering
\includegraphics[width=0.85\textwidth]{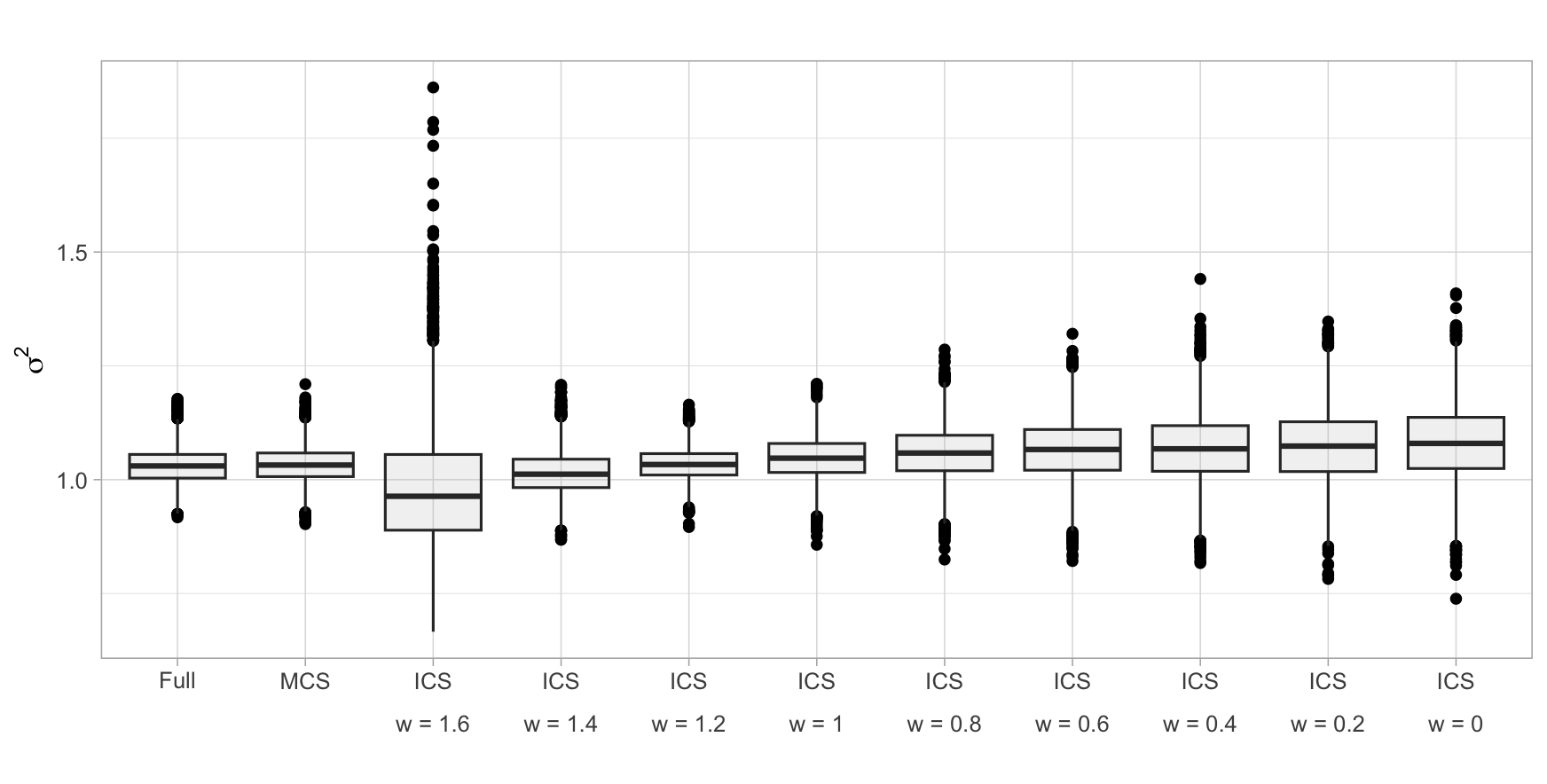}
\caption{Full posterior, MCS-based posterior based on $(us(Y),ps(Y))^T$ and ICS-based posteriors based on $cs_w(Y)$ for different $w$. The score-based posteriors are approximated using ABC. The observed dataset is from the equi-correlated normal model with $n=30$, $q=50$ and true parameter $\sigma_0=1$.}
\label{fig:equicorr}
\end{figure}

\begin{remark}
When $W^*_{cs}$ itself is of interest, e.g.~if its elements have a physical meaning or one would like to obtain the optimal composite score function,  it is possible to estimate the matrix coefficient $W^*_{cs}$ using simulations from $\pi\{\theta| cs_{1:K}(Y_{obs})\}$ by the  expansion $E\{\theta| cs_{1:K}(Y_{obs})\}\approx\theta_0+I(cs_{1:K})^{-1}W^*_{cs}cs_{1:K}(Y_{obs})$ \citep[Lemma A3]{li2018convergence}. 
One can estimate $W^*_{cs}$ by generating a joint sample of $(\theta,Y)$ from $\pi\{\theta| cs_{1:K}(Y_{obs})\}f(Y;\theta)$ and estimating the regression coefficients above. This avoids to invert the possibly ill-behaved covariance matrix $\Var\{cs_{1:K}(Y_{obs})\}$ in $W^*_{cs}$.
\end{remark}

\begin{remark}
For a $d$-dimensional summary statistic $r(Y)$ satisfying a central limit theorem and other regularity conditions in \citet[Conditions 3--5]{li2018convergence}, since $r(Y)-E\{r(Y)\}$ is an unbiased estimating function, by Theorem \ref{thm:optimal_comb} immediately we have the follows. First, the asymptotic information matrix $I(r)$ always increases as more summary statistics satisfying the conditions are included. Second, approximating the posterior distribution with $\pi\{\theta| r(Y_{obs})\}$ is first-order equivalent to approximating the score function $s(Y_{obs};\theta)$ at $\theta_0$ with the optimal linear transformation of $r(Y_{obs})-\mu_r(\theta_0)$, which is $a_n^2\nabla\mu_r(\theta_0)^TA_r(\theta_0)^{-1}\{r(Y_{obs})-\mu_r(\theta_0)\}$, in the sense of minimising the L2 distance with $s(Y_{obs})$.
\end{remark}

\subsection{ABC with estimated composite-scores}\label{sec:abc_est_scores}

It is popular to use Approximate Bayesian computation to simulate from summary-based posterior distributions; this only requires the ability to simulate from the model and evaluate the summary statistics \citep{beaumont2009adaptive}. Two standard ABC algorithms are provided in the Supplements. Evaluation of the stacked composite scores, $cs_{1:K}(Y)$, requires evaluation of the unknown parameter $\theta_{0}$ and all sub-score functions in the candidate composite scores $cs_{k}(Y;\theta)$, $k=1,\dots,K$. For models that depend on some latent vectors $X$, the sub-scores often do not have closed forms. When the number of evaluations is limited and the composite likelihood is formed from low-dimensional marginal events, numerical methods such as quadratures, Laplace approximation and finite differences can be used \citep{varin2008pairwise,ng2011composite}. However, these methods are computationally expensive and not feasible for our purposes, as ABC often requires a large number of summary statistic evaluations. Since the sub-scores in the composite scores are expectations over low and fixed dimensional marginal densities of latent states, we propose an importance sampling estimator for the composite scores, designed specifically for the ABC context.

Suppose there exists a latent random vector $X$ such that the joint density of $Y$ and $X$, $f(y, x;\theta)$, can be evaluated easily given $\theta$. Denote the marginal density of $X$ by $f_X(x;\theta)$. Since the marginal events $Y_{kj}$, $j=1,\dots,n_k$, are identically distributed, it is natural to assume that their associated marginal latent vectors are also identically distributed as follows. For each of $k=1,\dots,K$, there exists $n_k$ subsets of $X$, denoted by $\{X_{kj}\}_{j=1}^{n_k}$, such that $(X_{kj},Y_{kj})$ are identically distributed over $j$ and the density of $Y_{kj}$ conditional on $X$ is equal to that conditional on $X_{kj}$. Denote the density of $(X_{k1},Y_{k1})$ by $f_k(y,x;\theta)$ and its marginal density on $X_{k1}$ by $f_{k,X}(x;\theta)$. Let $\mathcal{X}\subset\mathbb{R}^l$ be the union of supports of $f_X(\cdot;\theta)$ over $\theta\in\mathcal{P}$ and $\mathcal{X}_k\subset\mathbb{R}^{l_k}$ be that for $f_{k,X}(\cdot;\theta)$. Then by Fisher's identity, $cs_{k}(Y)$ can be expressed as  
\begin{align}
cs_{k}(Y) & =\frac{1}{n}\sum_{j=1}^{n}E_{\theta_0}\{\bigtriangledown_{\beta}\log f_{k}(Y_{kj},X_{kj})\mid Y_{kj}\} \nonumber\\
& = \frac{1}{n}\sum_{j=1}^{n}\int\nabla_{\beta}\log f_{k}(Y_{kj},x)\frac{f_{k}(Y_{kj}\mid x)}{f_{k}(Y_{kj})}f_{k,X}(x)\,dx. \nonumber
\end{align}
The integrals above are expectations with respect to the density $f_{k,X}(x)$ which does not depend on $Y$ and $j$, which suggests using a sample of $X$ to estimate all the integrals. Therefore, for all pseudo datasets simulated in the ABC algorithm, we propose to estimate the deviation $cs_{1:K}(Y)-cs_{1:K}(Y_{obs})$ using a Monte Carlo sample simulated from some proposal distribution $q(x)$ in the space of $X$, and to estimate the unknown $\theta_0$ using an asymptotically unbiased estimator $\widehat{\theta}$. Specifically, let $\{x^{(1)},\dots,x^{(N)}\}$ be a sample simulated from $q(\cdot)$, and for each $i=1,\dots,N$ and $k=1,\dots,K$, let $\{x_{kj}^{(i)}\}_{j=1}^{n_k}$ be the identically distributed subsets of $x^{(i)}$ corresponding to $\{X_{kj}\}_{j=1}^{n_k}$ with marginal density $q_k(\cdot)$. Then $cs_{k}(Y)$ can be estimated by 
\begin{align}
& \widehat{cs}_{k}(Y;\widehat{\theta}) =\frac{1}{n}\sum_{j=1}^{n}\frac{1}{N}\sum_{i=1}^{N}\bigtriangledown_{\beta}\log f(Y_{kj},x_{kj}^{(i)};\widehat{\theta})\frac{f_{k}(Y_{kj},x_{kj}^{(i)};\widehat{\theta})}{\widehat{f}_{k}(Y_{kj};\widehat{\theta})q_k(x_{kj}^{(i)})},\label{eq:est_fisher_identity} \\
& \mbox{where }\widehat{f}_{k}(Y_{kj};\widehat{\theta})=\frac{1}{N}\sum_{i=1}^N\frac{f_{k}(Y_{kj},x_{kj}^{(i)};\widehat{\theta})}{q_k(x_{kj}^{(i)})}, \nonumber
\end{align}
and in the ABC algorithm, the deviation $cs_{1:K}(Y)-cs_{1:K}(Y_{obs})$ is replaced by $\widehat{cs}_{1:K}(Y)-\widehat{cs}_{1:K}(Y_{obs})$.

We argue that it is feasible to use $\widehat{cs}_{k}(Y;\widehat{\theta})$ in the ABC algorithm by showing that the computational cost of obtaining $\widehat{cs}_{k}(Y;\widehat{\theta})$ is linear in $n$ and thus comparable to that of simulating a pseudo-dataset, which is often linear in $n$. The cost of evaluating \eqref{eq:est_fisher_identity} is linear in $nN$, and the simulation size $N$ determines the trade-off between the approximation error of $\widehat{cs}_{1:K}(Y;\widehat{\theta})$ and the computational cost. Both $\widehat{\theta}$ and $x_1,\dots,x_N$ are obtained independently from the ABC algorithm, and their costs are considered negligible. Then the question is when $n$ increases, whether $N$ needs to increase for the approximation error to be negligible. Below it is shown that the approximation error can be negligible with $N$ fixed at a moderate value.

Since the ABC posterior probability mass concentrates in a neighbourhood of $\theta_{0}$ as $n\rightarrow\infty$ \citet[Lemma A3]{li2018convergence}, we only consider $Y$ simulated from $\theta\in\mathcal{P}_{0}$. The ABC posterior distribution depends on the data only through the deviation $\widehat{cs}_{1:K}(Y;\widehat{\theta})-\widehat{cs}_{1:K}(Y_{obs};\widehat{\theta})$. It means that the error from estimating $cs_{1:K}(Y)$ is negligible if the estimation error in $\widehat{cs}_{1:K}(Y;\widehat{\theta})-\widehat{cs}_{1:K}(Y_{obs};\widehat{\theta})$ is dominated by $cs_{1:K}(Y)-cs_{1:K}(Y_{obs})$. The estimation error has two sources, from the Monte Carlo error of using $X^{(1:N)}$ and from using $\widehat{\theta}$. The former is studied in the theorem below. Technical conditions of the theorem are listed in the Supplement.

\begin{thm} \label{thm:estimated_score}
Assume \ref{cond:prior}--\ref{cond:limit_moments} hold for $h(Y;\beta)=cs_{1:K}(Y;\beta)$, \ref{cond:differentiability}--\ref{cond:average_joint_information} in the Supplement hold and $N$ is fixed. For any $\theta\in\mathcal{P}_{0}$ and $Y\sim f(\cdot;\theta)$, with a random sample $X^{(1:N)}=\{X^{(1)},\dots,X^{(N)}\}$ from $q(\cdot)$, as $n\rightarrow\infty$, if $\widehat{\theta}\rightarrow\theta_c\in\mathcal{P}_0$ in probability, we have the following decomposition,
$$
\begin{aligned}
 & \widehat{cs}_{1:K}(Y;\widehat{\theta})-\widehat{cs}_{1:K}(Y_{obs};\widehat{\theta})-\{cs_{1:K}(Y;\widehat{\theta})-cs_{1:K}(Y_{obs};\widehat{\theta})\}\\
 & =\frac{1}{\sqrt{N}}\left\{c_{1}(X^{(1:N)},Y,Y_{obs};\widehat{\theta})+c_{2}(X^{(1:N)};\widehat{\theta})(\theta-\theta_{0})\right\},
\end{aligned}
$$
where remainders $c_{1}$ and $c_{2}$ are specified in the Supplement. It holds that $c_{1}(X^{(1:N)},Y,Y_{obs};\widehat{\theta})=O_{p}(a_n^{-1})$, $c_{2}(X^{(1:N)};\widehat{\theta})=O_{p}(1)$.
\end{thm}
By \ref{cond:sum_conv}, $cs_{1:K}(Y;\widehat{\theta})-cs_{1:K}(Y_{obs};\widehat{\theta})$ can be decomposed into $\mu(\theta,\widehat{\theta})-\mu(\theta_0,\widehat{\theta})+O_p(a_n^{-1})$ which dominates the right hand side of the above decomposition. Therefore, the theorem above shows $N$ can take a moderate value and does not need to increase with $n$ for the Monte Carlo error in $\widehat{cs}_{1:K}(Y;\widehat{\theta})-\widehat{cs}_{1:K}(Y_{obs};\widehat{\theta})$ to be dominated by $cs_{1:K}(Y;\widehat{\theta})-cs_{1:K}(Y_{obs};\widehat{\theta})$. The proof is given in the Supplement. Roughly speaking, the two remainders come from decomposing $\widehat{cs}_{k}(Y;\beta)-cs_{k}(Y;\beta)$ into two empirical processes of $X^{(1:N)}$ as follows,
$$
\begin{aligned}
& \frac{1}{n_k} \sum_{j=1}^{n_k} \frac{1}{N} \sum_{i=1}^N u\left(Y_{k j}, x_{kj}^{(i)} ; \beta\right)-\frac{1}{{n_k}} \sum_{j=1}^{n_k} u\left(Y_{k j} ; \beta\right) \\
= & \frac{1}{N} \sum_{i=1}^N\left[\frac{1}{{n_k}} \sum_{j=1}^{n_k} \left\{u\left(Y_{k j}, x_{kj}^{(i)} ; \beta\right)-v\left(x_{kj}^{(i)} ; \theta, \beta\right)\right\}-\left\{\frac{1}{{n_k}} \sum_{j=1}^{n_k} u\left(Y_{k j} ; \beta\right)-E_{q_k}(v(X_{k1} ; \theta, \beta))\right\}\right] \\
+ & \frac{1}{N} \sum_{i=1}^N\left[\frac{1}{{n_k}} \sum_{j=1}^{n_k}v\left(x_{kj}^{(i)} ; \theta, \beta\right)-E_{q_k}\{v(X_{k1} ; \theta, \beta)\}\right],
\end{aligned}
$$
where $u(Y_{k j}, x_{kj}^{(i)} ; \beta)$ and $u(Y_{k j}; \beta)$ denote the individual terms in  $\widehat{cs}_{k}(Y;\beta)$ and $cs_{k}(Y;\beta)$, respectively, and $v(x ; \theta, \beta)$ is the expectation of $u(Y_{k 1}, x ; \beta)$ with respect to $Y_{k1}$. Their orders are obtained by showing the associated classes of functions are Donsker classes. Both $c_{1}$ and $c_{2}$ have zero mean under the distribution of $X_{1:N}$, and $c_{1}$ has zero mean under the distribution of $(Y,Y_{obs})$.

The error from using $\widehat{\theta}$ can also be negligible if $\widehat{\theta}$ converges to $\theta_0$ in the rate of $a_n^{-1}$ and $\nabla_{\beta}{cs}_{1:K}(Y;\beta)$ converges to a Lipschitz function in probability as $n\rightarrow\infty$. Specifically, assume $\nabla_{\beta}{cs}_{1:K}(Y;\beta)\rightarrow\mu_{\beta}(\beta,\theta)$ in probability, and consider $Y\sim f(\cdot;\theta)$. Intuitively the approximation error can be written as 
$$
\begin{aligned}
& \{cs_{1:K}(Y ; \widehat{\theta})-cs_{1:K}(Y_{obs} ; \widehat{\theta})\}-\{cs_{1:K}(Y ; \theta_0)-cs_{1:K}(Y_{obs} ; \theta_0)\} \\
= & \{\mu_{\beta}(\theta, \theta_0)-\mu_{\beta}(\theta_0, \theta_0)+o_p(1)\}(\widehat{\theta}-\theta_0),
\end{aligned}
$$
where $\dot{\theta}$ is consistent to $\theta_0$. As the tolerance of ABC goes to $0$, $cs_{1:K}(Y ; \theta_0)-cs_{1:K}(Y_{obs} ; \theta_0)$ is in the order of $O(a_n^{-1})$. If the ABC posterior distribution is consistent for $\theta_0$, as the tolerance goes to $0$, $\theta-\theta_0$, and hence $\mu_{\beta}(\theta, \theta_0)-\mu_{\beta}(\theta_0, \theta_0)$, also goes to $0$, and the approximation error is $o_p(a_n^{-1})$ and negligible. The MCLE using the equally weighted sum of the component composite log-likelihood can be used as $\widehat{\theta}$, since standard asymptotic results are available for its convergence to $\theta_0$ \citep{varin2011overview}. Therefore, the overall computational complexity is linear in $n$. The impact of $N$ illustrated in the numerical example in Section \ref{sec:ssmodel} conforms to the theorem above. 


\begin{remark}
In contrast, particle Markov chain Monte Carlo (pMCMC) evaluates the likelihood of the full data using the particle filter and has the computational complexity $O(n^2)$ \citep{kantas2015particle}.  If the number of particles $N_p$ is fixed, the computational complexity of pMCMC is $O(n)$. However, it is recommended that $N_p$ should be chosen such that the variance of the log-likelihood estimator is roughly some fixed value \citep[]{pitt2012some, sherlock2015efficiency, doucet2015efficient}. Generally the variance of the estimate of the log-likelihood increases linearly with $n$, hence $N_p$ increases linearly with $n$, and the overall computational complexity of pMCMC is $O(n^2)$. Furthermore, numerical studies in Section \ref{sec:num_studies} show that, with the optimal combination property, the statistical efficiency of the new algorithm is fairly close to pMCMC in several cases.	
\end{remark}


\subsection{Application to the state-space model}\label{sec:ssmodel}

Here the approach in Section \ref{sec:abc_est_scores} is applied to a general state space model. Assume the observed process $\left\{Y_t\right\}$ depends on a latent Markov process $\left\{X_t\right\}$. The state-space model can be described by the two probability densities: for $t \geq 1$,

$$ 
Y_t\sim f(\cdot|X_t; \theta),\quad X_t\sim g(\cdot| X_{t-1}; \theta),
$$
where $f$ is the measurement density, $g$ is the transition density and $\theta$ is the unknown $p$-dimensional static parameter to be estimated.  Assume the latent process $\{X_t\}$ is stationary. For the simplicity of illustration, assume that both $Y_t$ and $X_t$ are scalar, though the method here can be extended naturally to multivariate cases. Denote the joint density for time steps $i:j$ by $f_{i:j}$ and $g_{i:j}$. Given observations $y_{1:n}$, the full likelihood is
\begin{align*}
{L}(\theta; y_{1:n}) =  \int \dots \int g_{1:n}(x_{1:n}; \theta) f_{1:n}(y_{1:n} | x_{1:n} ; \theta) dx_1\dots dx_n.
\end{align*}
Except in some special cases, full likelihood-based inference is difficult due to the $n$-dimensional integral. Although composite likelihoods will generally not be available in analytical form, estimation of the lower dimensional sub-likelihood integrals is easier.

As an example, consider composite likelihoods formed from likelihoods of pairs at some fixed lag $k$. Here, the sub-likelihoods are 2-dimensional. Assume $g_{1:n}(\cdot; \theta)$ can be easily simulated from, and $f_{i,i+k}(y_i, y_{i+k}\mid x_i, x_{i+k};\theta)$ is easy to evaluate. In the ABC algorithm, the pairwise scores for pseudo-data simulated for any parameter value can be estimated using the same sample from the latent process. More specifically, by the stationarity of $X_t$, the sub-likelihood evaluated at $\theta_0$ of the pair $(y_{i}, y_{i+k})$ simulated for any parameter value $\theta$ can be written as
\begin{align*}
f_{i,i+k}(y_i, y_{i+k}) & = \int \int f_{i,i+k}(y_i, y_{i+k}, x_i, x_{i+k}) dx_idx_{i+k} \\
& = \int \int f_{i,i+k}(y_i, y_{i+k}\mid X_i=x_1, X_{i+k}=x_{1+k})g_{1,1+k}(x_1,x_{1+k})dx_1dx_{1+k},
\end{align*}
and estimated by
\[
\widehat{f}_{i,i+k}(y_{i},y_{i+k};\widehat{\theta})=N^{-1}\sum_{i=1}^{N}f_{i,i+k}(y_{i},y_{i+k}\mid x_{1}^{(i)},x_{1+k}^{(i)};\widehat{\theta})
\]
where $x_{1,1+k}^{(i)}\overset{i.i.d}{\sim}g_{1,1+k}(\cdot;\widehat{\theta})$, $i=1,\dots,N$. Similarly, the corresponding sub-score $\nabla \log f_{i, i+k}(y_i, y_{i+k})$ can be estimated as in \eqref{eq:est_fisher_identity} using $g_{1,1+k}(\cdot;\widehat{\theta})$ as $q(\cdot)$ and samples $\{x_{1,1+k}^{(i)}\}_{i=1,\dots,N}$ as above. The estimator uses the fact that in composite score based ABC, for all pseudo and observed datasets, all sub-scores are evaluated at the same $\beta=\widehat{\theta}$, e.g. the MCLE. Therefore a Monte Carlo sample can be drawn from $g_{1,1+k}(\cdot;\widehat{\theta})$ and reused to estimate all sub-scores.  

If the latent process is non-stationary, a similar approach could be implemented. However, the same samples could not be reused for pairs with the same time lag. Instead, we would require samples $x_{1: n}^{(j)}, j=1, \dots, N$, drawn from the latent process $f\left(x_{1: n} ; \theta\right)$, and the sub-likelihood for pair $\left(y_i, y_{i+k}\right)$ would be estimated using samples $\left\{x_i^{(j)}, x_{i+k}^{(j)}\right\}, j=1, \dots, N$. However, samples only need to be generated once, not at each ABC iteration, and the computational cost is still linear in $n$. 

We now demonstrate the impact of $N$ with a numerical example. Consider the stationary AR(1) with additive noise state-space model \citep[see, for example,][for further details]{varin2008pairwise}. The pairwise sub-likelihoods and sub-scores are analytically available for this model. Therefore, we can directly compare the ABC parameter estimates obtained using analytical and estimated summaries. Our experiments show both methods produce very similar posterior mean estimates, hence these results are omitted. Figure \ref{fig:ARscoreest} shows the relative efficiency (see Section \ref{sec:num_studies} for definition) of ABC with estimated summaries with respect to ABC with analytical summaries for a range of $n$ and $N$. For all $n$, the relative efficiency increases towards 1 with increasing $N$, but can be very low when $N$ is small. However, with a sufficiently large $N$ a fairly accurate variance estimate can be obtained. Additionally, we observe that, for a fixed $N$, the relative efficiencies are similar for all $n$, suggesting that the $N$ required to achieve a desired level of posterior variance approximation does not depend on the data size. 

\begin{figure}[!ht]
\centering
\includegraphics[width=\textwidth]{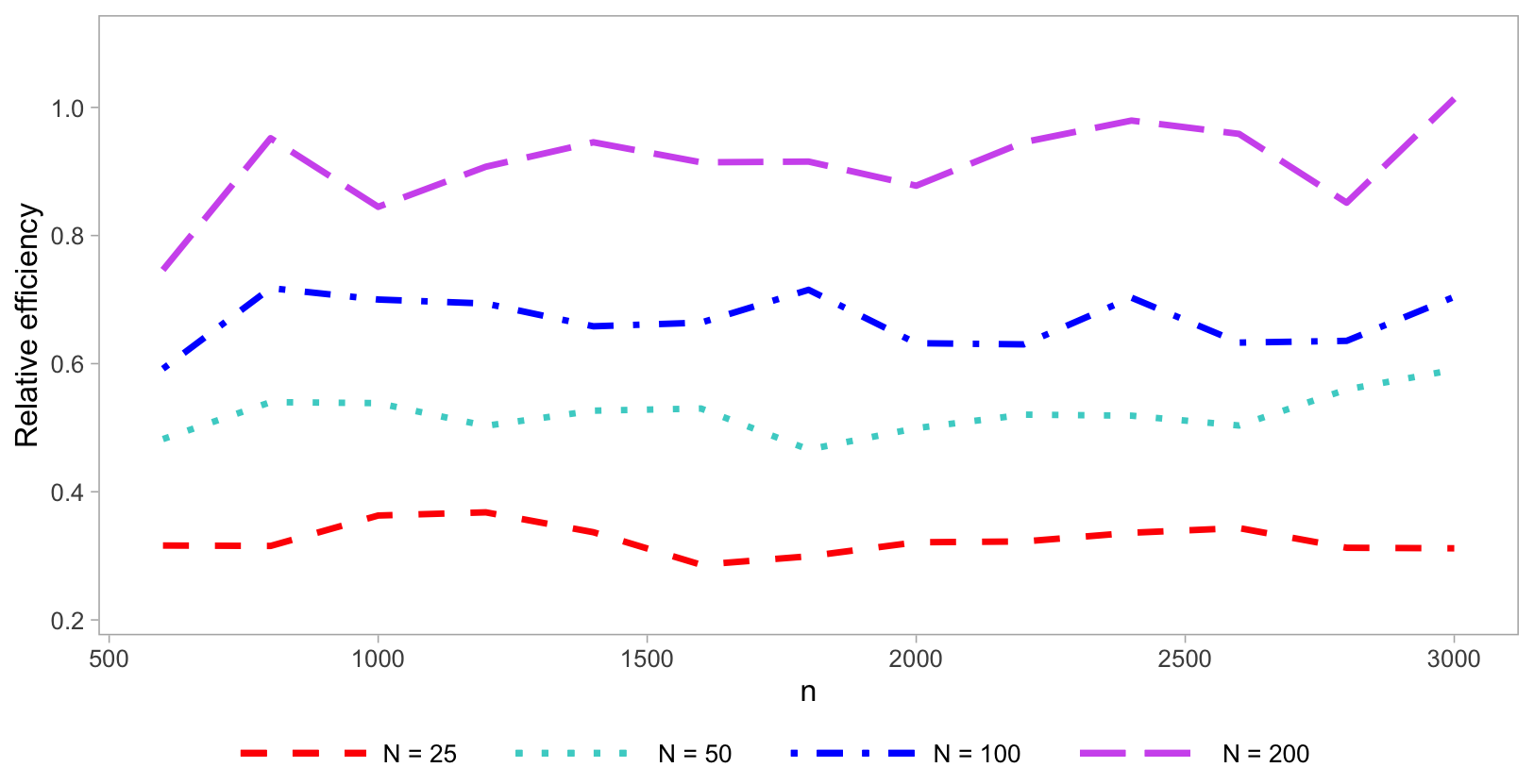}
\caption{Average relative efficiency of ABC with estimated summaries with respect to ABC with analytical summaries, for different time series lengths $n$ and numbers of Monte Carlo samples $N$. The summary statistic is the stack of the lag-$k$ pairwise scores for $k=1,2$, evaluated at the MCLE using the lag-2 summed pairwise likelihood. The results are calculated from 75 replicated datasets generated from the AR(1) with additive noise model with true parameters $\tau^2=0.3$, $\gamma=0.95$, $\sigma^2$=1 \citep[see][for model definition]{varin2008pairwise} and for each of $n = 600, 800, 1000, \dots, 3000$. }
\label{fig:ARscoreest}
\end{figure}

\section{Adaptive selection of composite scores}\label{sec:adaptiveABC}

\subsection{Motivation}

Suppose there are $\tau$ candidate component composite scores $cs_{k}(Y;\theta)$, $k=1,\dots,\tau$. If they converge in the same rate, all candidates should be included in the summary statistic by Theorem \ref{thm:optimal_comb}. The increased computational complexity from the increased dimensionality can be mitigated by efficient ABC algorithms including dimension reduction \citep{blum2013comparative} and efficient importance proposal to ensure a non-degenerate acceptance probability \citep{li2018convergence}. However, since the size of events $n_k$ in each set varies, convergence rates of $cs_k$ can be different, and Theorem \ref{thm:optimal_comb} can not be applied. For example, consider a time-series dataset with length $n$ and the lag-$k$ pairwise scores with $k<m$. The choice of $m$ needs to be determined. If $m$ is too small, we exclude informative summaries and lose efficiency. However, the lag-$m$ pairwise score has $n-m$ summation terms, and as $m$ increases, its convergence rate changes, e.g.~if $m=n/2$. As another example, for spatial datasets , the number of observations involved in the pairwise score batched by distance can be very small for large distances, and varies greatly for different distances if the observations are irregularly spaced.

\subsection{ABC posterior under summary statistic with heterogeneous rates}

Here we study the asymptotic behaviour of the ABC posterior distribution with summary statistics that converge in heterogeneous rates as $n\rightarrow\infty$. Suppose there are multiple summary statistics $\{s_{k}(Y)\}_{k=1}^{\tau}$, each of them satisfies a central limit theorem with rate $a_{kn}$, $k=1,\dots,\tau$, where $a_{kn}\rightarrow\infty$ as $n\rightarrow\infty$ and $a_{kn}$ may be different. For simplicity, assume $a_{kn}$ are in a decreasing order, i.e.~the fastest convergence rate is $a_{1n}$ and the slowest rate is $a_{\tau n}$. The asymptotic properties of the ABC posterior distribution with $a_{kn}$ being identical are studied in \citet{li2018convergence} and \citet{frazier2018asymptotic}. It is known that under certain regularity
conditions, the posterior distribution conditional on $S_{k}$ satisfies
a Bernstein-von Mises theorem with the rate $a_{kn}$. The ABC posterior
distribution can retain the convergence rate if $\varepsilon_{n}$
degenerates to $0$ fast enough, at the rate $a_{kn}^{-1}$, otherwise
its rate would be limited to the slower
rate$\varepsilon_{n}^{-1}$. The convergence properties here extend results in \citet[Proposition 1]{li2018convergence} to the scenario of heterogeneous rates.

Denote the dimensions of the summaries by $d_{s,k}$. Let $d=\sum_{k=1}^{\tau}d_{s,k}$. Let $K(x)$ be a positive kernel, where $\max_{x}K(x)=1$, and $\varepsilon>0$ is the bandwidth. Let $s_{1:\tau}(Y)$ be the stack of $s_{1}(Y),\dots,s_{\tau}(Y)$ and assume it follows a parametric density $f_n(s\mid \theta)$. Denote the random variable with density $f_n(s\mid \theta)$ by $S_{1:\tau}$. Assume a joint density, $\pi_{\varepsilon}(\theta,s\mid s_{{\rm obs}})$, for $(\theta,s)$ as
\begin{align} 
 & \frac{\pi(\theta)f_{n}(s\mid\theta)K\{\varepsilon^{-1}(s-s_{{\rm obs}})\}}{\int_{\mathbb{R}^p\times\mathbb{R}^{d}}\pi(\theta)
f_{n}(s\mid\theta)K\{\varepsilon^{-1}(s-s_{{\rm obs}})\}\,d\theta d s},\label{eq:parameter_summary_jointdensity}
\end{align}
where $s_{{\rm obs}}=s_{1:\tau}(Y_{{\rm obs}})$. The approximate Bayesian computation posterior density is the marginal density
\begin{equation} \label{eq:ABC_def}
\pi_{\varepsilon}(\theta\mid s_{{\rm obs}})=\int \pi_{\varepsilon}(\theta,s\mid s_{{\rm obs}})\,d s.
\end{equation}

Extension of relevant conditions of \citet[Proposition 1]{li2018convergence} for heterogeneous rates are listed below. Other regularity
conditions from \citet[Proposition 1]{li2018convergence} are given in the Appendix.
\begin{description}
\item[{(C2*)\label{sum_conv}}] There exists a sequence of vectors $a_{1n},\dots,a_{\tau n}$,
with scalars $a_{kn}\rightarrow\infty$ as $n\rightarrow\infty$ for
$k=1,\dots,\tau$, a $d$-dimensional vector $s(\theta)$ and a $d\times d$
matrix $A(\theta)$, such that for all $\theta\in\mathcal{P}_{0}$,
\[
a_{n,diag}\{S_{1:\tau}-s(\theta)\}\xrightarrow{\fL}N\{0,A(\theta)\},\mbox{\ensuremath{n\rightarrow\infty}},
\]
in distribution, where $a_{n,diag}=diag(a_{1n}I_{d_{s,1}},\dots,a_{\tau n}I_{d_{s,\tau}})$.
We also assume that $s_{\rm{obs}}\rightarrow s(\theta_{0})$ in
probability. Furthermore, that 

\begin{itemize}
\item[(i)] $s(\theta)$ and $A(\theta)\in C^{1}(\mathcal{P}_{0})$, and $A(\theta)$
is positive definite for any $\theta$; 
\item[(ii)] for any $\delta>0$ there exists $\delta'>0$ such that $\|s(\theta)-s(\theta_{0})\|>\delta'$
for all $\theta$ satisfying $\|\theta-\theta_{0}\|>\delta$; and 
\item[(iii)] for each $k=1,\dots,\tau$, $Ds_{k}(\theta)^{T}A_{kk}^{-1}(\theta)Ds_{k}(\theta)$
has full rank at $\theta=\theta_{0}$.
\end{itemize}
\end{description}
\begin{description}
\item[{(C3*)\label{sum_approx}}] There exists $\alpha_{n}$ satisfying
$\alpha_{n}/a_{\tau n}^{2/5}\rightarrow\infty$ and a density $r_{max}(w)$
satisfying \ref{kernel_prop} (ii)-(iv), such that $\sup_{\theta\in\mathcal{P}_{0}}\alpha_{n}|f_{W_{n}}(w|\theta)-\widetilde{f}_{W_{n}}(w|\theta)|\leq c_{3}r_{max}(w)$
for some positive constant $c_{3}$.
\end{description}
Without loss of generality, assume that $S_{k}$ are ordered with
their convergence rates strictly decreasing, i.e.~for $i>j$, $a_{in}=o(a_{jn})$.
Let 
\begin{align*}
a_{n,\varepsilon} & =\begin{cases}
a_{1n}, & \text{when }\varepsilon_{n}=O(a_{1n}^{-1}),\\
a_{kn}, & \text{when }\varepsilon_{n}=O(a_{kn}^{-1})\text{ and }\varepsilon_{n}=\omega(a_{k-1,n}^{-1})\text{ for }1<k\leq\tau,\\
\varepsilon_{n}^{-1}, & \text{ when }\varepsilon_{n}=\omega(a_{\tau n}^{-1}),
\end{cases}
\end{align*}
indicating the rate that is closest to $\varepsilon_{n}^{-1}$ among
rates slower than or equal to $\varepsilon_{n}^{-1}$ or $\varepsilon_{n}^{-1}$
itself. The result below considers the posterior distribution of approximate Bayesian computation, denoted by $\Pi_{\varepsilon}(\theta\in A\mid s_{{\rm obs}})$ for $A\in\mathscr{B}^p$,
as $n\rightarrow\infty$, and shows that it converges with the rate $a_{n,\varepsilon}$. 

\begin{thm}
\label{thm:ABC_heteo_rates}Assume Condition \ref{cond:prior},\ref{sum_conv},\ref{sum_approx} and \citet[Conditions 4 and 5]{li2018convergence}. Let $\theta_{\varepsilon}$
be the ABC posterior mean. As $n\rightarrow\infty$, if $\varepsilon_{n}=o(a_{\tau n}^{-1/2})$,
the following results hold:
\begin{description}
\item [{(a)}] $\sup_{A}\left|\Pi_{\varepsilon}\{a_{n,\varepsilon}(\theta-\theta_{0})\in A\mid s_{\text{obs}}\}-\int_{A}\psi_{\varepsilon}(t)\,dt\right|\rightarrow 0$ in probability, where $\psi_{\varepsilon}(t)$ is a probability density in $\mathcal{P}$
depending on $\varepsilon_{n}$. 
\item [{(b)}] $\lim_{n\rightarrow\infty}\int_{\|t\|>M}\psi_{\varepsilon}(t)\,dt\rightarrow0$
in probability as $M\rightarrow\infty$.  
\item [{(c)}] $\theta_{\varepsilon}-\theta_{0}=O_{p}(a_{1n}^{-1})+O_p(\alpha_n^{-1}a_{n,\varepsilon}^{-1})$.
\end{description}
 
\end{thm}

The proof and the explicit form of $\psi_{\varepsilon}(t)$ are given in Appendix B.2. The expression of the leading term of $\psi_{\varepsilon}(t)$ is tedious as it depends on the projections of the summary statistics onto multiple linear sub-spaces grouped by different rates. The proof extends the proof of \citet[Proposition 1]{li2018convergence} to accommodate the heterogeneous rates by dividing the summary statistics into three groups according to their rates (faster, equal or slower than $\varepsilon_n^{-1}$) and handling each group separately. The main technical challenge is to find different upper bounds and linear projections onto the $p$-dimension space of $\theta$ for each group. 

With the ABC approximation, some summary statistics with faster rates are dominated by the bandwidth and some are not. The theorem shows that the convergence rate of the ABC posterior distribution depends on the tolerance level $\varepsilon_n$ in such a way that it is equal to the fastest rate among those of the summary statistics that are not dominated by the bandwidth. In particular, the summary-based posterior distribution, which is a special case of the ABC posterior distribution when $\varepsilon_n=0$, converges with the fastest rate $a_{1n}$. So the ABC posterior distribution loses the efficiency when $\varepsilon_n$ is sufficiently big that the approximation dominates the more efficient part of the summary statistics. This result suggests to only choose the composite scores that have the highest convergence rates so that the bandwidth does not dominate their efficiency, and include all such composite scores. 

On the other hand, the ABC posterior mean is consistent to $\theta_0$ and converges faster than the ABC posterior distribution. Similar to the result for homogeneous rates in \citet{li2018convergence}, it is much easier for the ABC posterior mean to converge than the ABC posterior variance.  Let $\widetilde{\pi}_{\varepsilon}(\theta\mid s_{obs})$ be the normal counterpart of $\pi_{\varepsilon}(\theta\mid s_{obs})$ where the likelihood of $s$ is replaced by its normal limit density, and $\widetilde{\theta_{\varepsilon}}$ be its mean. In the expansion of $\theta_{\varepsilon}-\theta_{0}$, the term $O_p(a_{1n}^{-1})$ corresponds to $\widetilde{\theta_{\varepsilon}}-\theta_{0}$. It is surprising that under the normal density, its rate does not depend on the bandwidth and always has the fastest rate. This is due to the symmetry of normal density in all directions leading to cancellation of terms when being averaged. The second term comes from the difference of $\widetilde{\pi}_{\varepsilon}(\theta\mid s_{obs})$ and $\pi_{\varepsilon}(\theta\mid s_{obs})$ and is much faster than the $a_{n,\varepsilon}^{-1}$ rate of the ABC posterior variance. Compared to the results obtained under conditions in \citet{frazier2018asymptotic} \citep{lawless2023asymptotics}, by assuming asymptotic normality of the composite scores which is typically satisfied in scenarios where composite likelihood is used, Theorem \ref{thm:ABC_heteo_rates} gives sharper results for the convergence rate and the limit distribution of ABC posterior distribution, and reveals the faster convergence of the ABC posterior mean. 

The theorem above also implies that when some summary statistics have slower rates than others, using the bandwidth obtained by fixing the acceptance probability at a non-degenerate value will cause loss of efficiency such that the ABC posterior distribution converges in the slowest rate $a_{\tau n}$ because such bandwidth is in the order $O(a_{\tau n}^{-1})$. Specifically, following arguments in \citet{li2018convergence}, for the summary statistic $s$ sampled with parameter value $\theta$, the acceptance probability depends on 
\begin{align*}
	\frac{s_{}-s_{{\rm obs}}}{\varepsilon_{n}} & =\frac{1}{\varepsilon_{n}}[\{s_{}-s(\theta)\}+\{s(\theta)-s(\theta_{0})\}+\{s(\theta_{0})-s_{{\rm obs}}\}].
\end{align*}
 For the acceptance probability to not go to $0$, $\varepsilon_{n}$ needs to decrease with a rate not faster than that of the distance between $s$ and $s_{\rm obs}$. Since the distance is determined by the largest distance among the dimensions, $\varepsilon_{n}$ needs to be $O(a_{\tau n}^{-1})$. Therefore, among the summary statistics, except those with the slowest convergence rate, all others are dominated by the bandwidth. The bandwidth might be relaxed to be $O(a_{\tau n}^{-3/5})$ if the regression adjustment is used \citep{li2018convergence}, but still dominates summaries that converge with rates faster than $O(a_{\tau n}^{-3/5})$.


\subsection{Adaptively chosen component composite scores}

For a set of summary statistics, Theorem \ref{thm:ABC_heteo_rates} suggests to include only the subset that has the fastest convergence rate. Theorem \ref{thm:optimal_comb} suggests to include as many summaries in the subset until the gain in efficiency of including more is too small and can not offset the additional ABC approximation error due to the increased dimension. Therefore, if such a subset can be identified, or the summary statistics can be sorted in the order of decreasing rates, the summary statistics can be added sequentially until the approximate posterior distribution does not improve. This is not straightforward when selecting summary statistics in general. However, the batched composite scores often have a natural order that is related to the order of convergence rates. For example, given a time series dataset $y_{1:n}$, the lag-$k$ pairwise scores $pl_k$ with increasing $k$ are ordered by the corresponding time lag between pairs. In \citet{davis2011comments} and \citet{ng2011composite}, when $k$ is a fixed number, it is shown that the lag-$k$ pairwise score with $k$ fixed has the same convergence rate $\sqrt{n}$. Furthermore, it is plausible to assume that if the inclusion of pairs with lag $k$ results in negligible improvement in the quality of the posterior inference, then the contribution of pairs with lag greater than $k$ will also be negligible. Therefore the lag-batched pairwise scores with increasing lags are ordered with decreasing rates. 


Suppose there is a sequence of summary statistics $z_1(Y),z_2(Y),\dots,z_K(Y)$ where $z_{k}$ might be a subset of $z_{k+1}$ and the summaries can have heterogeneous convergence rates. Based on the above discussion, we propose to choose $z_m$ by monitoring the turning point $m$ of the sequence of corresponding ABC posterior distributions at which the ABC posterior mean converges and the variance of the ABC posterior does not decrease and instead increases. One naive way is to conduct multiple runs of standard ABC algorithm separately at a grid of $\{1,\dots,K\}$ and to locate the turning points. However, this may not be computationally practical as each ABC run requires manual tuning and parallel runs are memory intensive. We propose to use a single run of the ABC-population Monte Carlo (PMC) scheme with the value of $k$ sequentially increased to find the turning point. Let $q(\theta;\mu,\nu)$ be the probability density of a location-scale family with location $\mu$ and scale $\nu$. A generic summary statistic adaption scheme is given below.

\begin{center}
\begin{algorithm}
\caption{\textbf{Summary statistics adaptation scheme}}
\small{
{At the $k$th step, suppose a sample from the ABC posterior $\pi_{\varepsilon_k}\{\theta\mid z_k(Y_{obs})\}$ is available. Denote the empirical posterior mean and variance matrix by $\widehat{E}_k$ and $\widehat{V}_k$.} 
{\begin{enumerate}
	\item Use an ABC-PMC algorithm starting from the proposal $q(\theta;\widehat{E}_k, \widehat{V}_k)$ to sample from $\pi_{\varepsilon_{k+1}}\{\theta\mid z_{k+1}(Y_{obs})\}$.
	\item Calculate the empirical mean and variance matrix of the sample from step 1 denoted by $\widehat{E}_{k+1}$ and $\widehat{V}_{k+1}$.
	\item If $\widehat{E}_{k+1}\neq \widehat{E}_{k}$, or $\widehat{E}_{k+1}=\widehat{E}_{k}$ and $\widehat{V}_{k+1}\leq\widehat{V}_{k}$, let $k\leftarrow k+1$. Otherwise terminate.   
\end{enumerate}
}}
\label{algo:adaptation}
\end{algorithm}
\end{center}

Let $a_{n,{\rm min}}^{(k)}$ be the slowest convergence rate of those of the dimensions of $z_k$ and $q$ be the smallest index of $\{k:a_{n,{\rm min}}^{(k)}\neq a_{n,{\rm min}}^{(k+1)}\}$. By listing $z_k$ in the order where $Var\{\theta| z_k(Y_{obs})\}$ is non-increasing in $k$ for $k\leq q$, with a fixed acceptance rate, the scheme will terminate at $k=q$ or $k<q$ when the reduction of the summary-based posterior variance can not offset the increase of Monte Carlo variation. For example, with $pl_k$ as the composite components, let $z_k=ps_{1:k}$; with the k-consecutive log-likelihood $tl_k$ which combines the sub-likelihoods of tuples with size $k$, let $z_k=ts_k$. In the scheme $\widehat{E}_k=\widehat{E}_{k+1}$ is checked before comparing the two ABC posterior variance matrices. This is because the ABC posterior mean converges with the bandwidth larger than that required for the ABC posterior variance to converge, as shown in Theorem \ref{thm:ABC_heteo_rates}. Hence if the ABC posterior mean changes significantly when $z_k$ is changed to $z_{k+1}$, it indicates that $z_{k+1}$ is much more informative or $\varepsilon_n$ is not small enough, and more adaption is needed.

Another benefit of Algorithm \ref{algo:adaptation} is that if $z_k$ is computationally more costly as $k$ increases, it sequentially increases the cost of computing the summary throughout the algorithm by the adaptive importance sampling. Even if a suitable summary $z_K(Y)$ was known prior to running ABC, designing a sequence of increasingly more informative and more costly summary statistics with $z_K(Y)$ as the last one, Algorithm \ref{algo:adaptation} can reduce the overall computing time by starting the algorithm with the cheapest summary. This is illustrated by the numerical study in Section \ref{sec:SV}. 

The specifics of the implementation of the scheme in our numerical studies are given below. First, the ABC-PMC algorithm in \citet{lenormand2013adaptive} is used, the description of which is given in the Supplements. The stopping rule in step 3 of the scheme is used at the end of each PMC cycle. It compares estimates between the current cycle and the previous cycle to determine which summary is used in the next cycle. In particular, if the mean changes or variance decreases, the same summary used in the current cycle is used in the next cycle; if the mean does not change and the variance does not decrease, the summary is adapted for the next cycle; if the summary is adapted in two consecutive cycles, the algorithm terminates. Second, statistical tests are used to assess the equivalence for the empirical means and variances of the ABC samples in step 3. In our numerical studies we use bootstrap hypothesis tests (see, for example, \citet{davidson1997bootstrap}, \citet{boos1989bootstrap}). The choice of these criteria is outside the scope of this work. More discussions on tuning the significance levels of the tests can be seen in \citet{rosabeththesis}. Third, the regression-adjusted mean and variance are used within the adaptation rule due to its significant improvement on the variance estimators \citep{li2018convergence}.

\section{Numerical studies}\label{sec:num_studies}

Given a composite likelihood $cL(Y_{obs};\theta)$, suppose a set of candidate composite log-likelihoods $\{n_kcl_k(Y_{obs};\theta)\}_{k=1}^{K_n}$ is obtained using the batching method. Let $$cl_{{\rm equal};d}(Y;\theta)=\sum_{k=1}^d n_kcl_k(Y_{obs};\theta)$$ be the equally weighted sum of the composite log-likelihood components up to $k=d$, and $cs_{{\rm equal};d}(Y;\theta)$ be its score function. Let $cs_{1:d}(Y;\theta)$ be the stack of the scores of the composite log-likelihood components up to $k=d$.  Let $cL_{{\rm adj};d}(Y_{obs};\theta)$ be $cL_{{\rm equal};d}(Y;\theta)$ adjusted using the method in \citet{pauli2011bayesian}, given in Supplement. Let $\widehat{\theta}$ be the MCLE of $cs_{{\rm equal};d}(Y;\theta)$ with some value of $d$. In this section the following approximate posterior distributions will be compared with the full posterior distribution: 
\begin{description}\addtolength{\itemsep}{-0.5\baselineskip}
\item[Composite posterior:] The quasi-posterior density proportional to $\pi(\theta)cL_{{\rm equal};d}(Y_{obs};\theta)$.
\item[Calibrated composite posterior:] The quasi-posterior density proportional to $\pi(\theta)cL_{{\rm adj};d}(Y_{obs};\theta)$.
\item[ICS-ABC:] Individual composite score ABC, the ABC posterior distribution conditional on $cs_{{\rm equal};d}(Y_{obs};\widehat{\theta})$.
\item[MCS-ABC:] Multiple composite scores ABC, the ABC posterior distribution conditional on $cs_{1:d}(Y_{obs};\widehat{\theta})$.
\item[Adaptive MCS-ABC:] The value of $d$ of MCS-ABC is adaptively selected using Algorithm \ref{algo:adaptation}.
\end{description}

The full posterior distributions are obtained using either MCMC or pMCMC \citep{andrieu2010particle}. The number of particles in pMCMC is chosen such that the standard deviation of the estimate of the log-likelihood evaluated at an estimate of the posterior mean is approximately $1.5$ \citep{pitt2012some,doucet2015efficient,sherlock2015efficiency}. The ABC posterior distributions are obtained using the adaptive ABC-PMC algorithm of \citet{lenormand2013adaptive} followed by the linear regression adjustment \citep{beaumont2002approximate}. Note that both $cs_{\rm equal;d}(Y;\widehat{\theta})$ and $cs_{1:d}(Y;\widehat{\theta})$ are constructed from the same set of sub-scores, and the computational costs  of evaluating both are the same. Further note that the two methods are equivalent with $d=1$.

When comparing these methods, the following metrics are used to measure the quality of the posterior distribution approximation: 
the root mean square error (RMSE) of the approximate posterior point estimators $\widetilde{\theta}$ over $M$ replications, defined as 
$$
{\rm RMSE}(\widetilde{\theta})=\sqrt{\frac{1}{M}\sum_{i=1}^M (\widetilde{\theta}_i-\theta_{\rm post})},
$$ 
where $\theta_{\rm post}$ is the full posterior mean; the log generalised variance ($\log GV$), defined as $\log\det(\widehat{\Sigma})$ where $\widehat{\Sigma}$ is the approximate posterior variance matrix; and the average relative efficiency (RE) of the approximate posterior distribution with respect to the full posterior distribution over $M$ replications, defined as $$
{\rm RE}=\frac{1}{M}\sum_{i=1}^M\frac{\det{\Sigma_{\rm post}}}{\det{\widehat{\Sigma}}},
$$
where $\Sigma_{\rm post}$ is the full posterior variance matrix. The relative efficiency reports the gained proportion of the full efficiency when using the approximation, and is better when it is closer to $1$. The ratio between the REs of two approximations reports the efficiency improvement of one over the other. 

	

\subsection{Gaussian random field for spatial data}

 Let $\{Z({s}), {s} \in \mathbb{R}^2\}$ be a stationary zero-mean spatial Gaussian random field where ${s}$ denotes location coordinates. Given $n$ observations, $\left(Z\left({s}_1\right), Z\left(s_2\right), \dots, Z\left(s_n\right)\right) \sim$ $N(0, \Sigma)$ where $\Sigma_{i j}=\operatorname{Cov}\left(Z\left(s_i\right), Z\left(s_j\right)\right)$. We assume an exponential covariance function $\mathrm{Cov}(Z(s_i), Z(s_j))= \sigma^2 \exp( - {\| s_i- s_j\|}/{\phi})$ with known variance $\sigma^2=5$ and unknown scalar range $\phi>0$. Let $v_{ij}=Z(s_i)-Z(s_j)$, $i,j\in\{1,\dots,n\}$, be the pairwise differences. Batching the pairwise log-likelihood of all differences by the Euclidean distance between locations gives the candidate composite log-likelihoods  
 \begin{align}\label{eq:plik}
cl_k(Y; \phi) &=  \frac{1}{n_k}\sum_{||s_i- s_j||=d_k} \log f(v_{ij};\phi), \qquad  k = 1, \dots, K_n, 
  \end{align}
where $K_n$ is the number of unique distances between all pairs of observations, $\{d_k\}$ are the unique distances, and $n_k$ is the number of pairs with distance $d_k$. Let $cL_{{\rm equal};d}(Y;\theta)$ denote the sum of the composite components corresponding to pairs with distances $||s_i- s_j||< d$ for some distance $d$. Equivalently, this corresponds to the binary weighted composite likelihood in which a weight of zero is assigned to pairs of observations if the distance between their corresponding observation sites is greater than or equal to $d$. We use $d=\infty$ to indicate that all sites are included. \citet{bevilacqua2012estimating} assess the efficiency of these weighted composite likelihoods in the MCLE setting.

The observed spatial process is generated on the $n=64$ points on the regular $8\times 8$ grid $\{0,1,\dots, 7\}^2$ with the true parameters $\phi_0=10$. We assign a Uniform $U(0,100)$ prior distribution to $\phi$.  Figure \ref{fig:spatial} shows boxplots of the sampled posteriors. 

\begin{figure}[!ht]
\centering
\includegraphics[width=\textwidth]{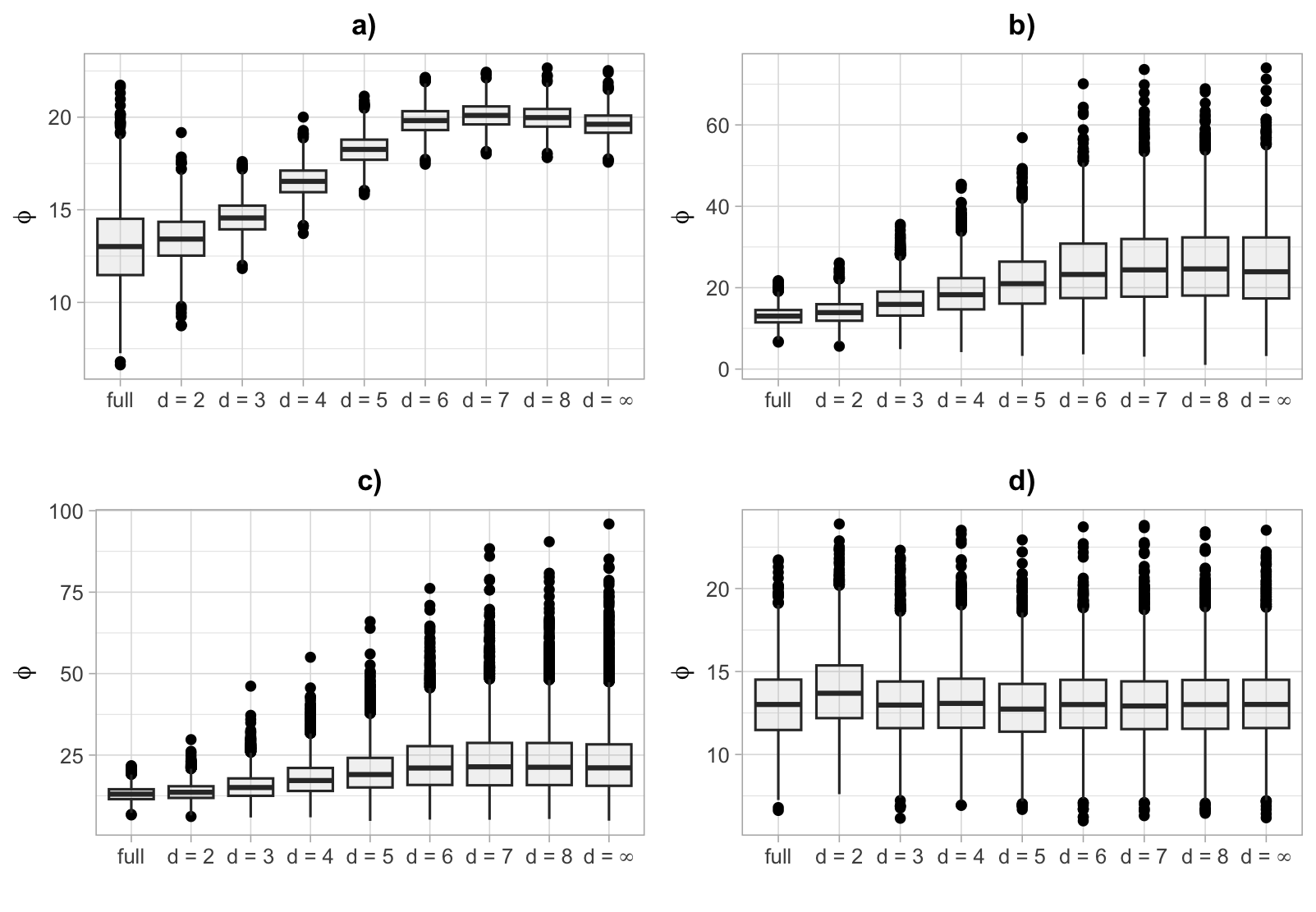}
\caption{Full posterior and approximate posteriors using composite likelihood methods with different $d$. a) Composite posteriors; b) Calibrated composite posteriors; c) ICS-ABC; d) MCS-ABC. }
\label{fig:spatial}
\end{figure}

For all choices of $d$, the (non-calibrated) composite posteriors are over-concentrated (with respect to the full posterior), whereas the calibrated composite posteriors are overly dispersed. Similarly, the ICS-ABC posteriors produce over-inflated estimates of posterior variance. In all three cases, the level of over or under concentration decreases and the posterior median estimates generally improve as $d$ is decreased.  In contrast, MCS-ABC produces a good approximation to the full posterior with any choice of $d$, with a slight improvement when $d$ is increased from $d=2$ to $d=3$, and very similar results for all $d\geq 3$. 

Figure \ref{fig:spatial_rep} and Table \ref{tab:spatial_rep} report the performance over 50 replications for each method at their best performing weighting scheme, which is $d=3$ for MCS-ABC and $d=2$ for all other methods. MCS-ABC obtains a significantly lower RMSE estimating the posterior mean compared with the other methods. Calibrated composite posterior and ICS-ABC both overestimate posterior variance (relative efficiency below $1$), and composite posterior significantly underestimates posterior variance (relative efficiency above $1$). In contrast, MCS-ABC obtains efficiency very close to the full efficiency.

\begin{figure}[!ht]
\centering
\includegraphics[width=\textwidth]{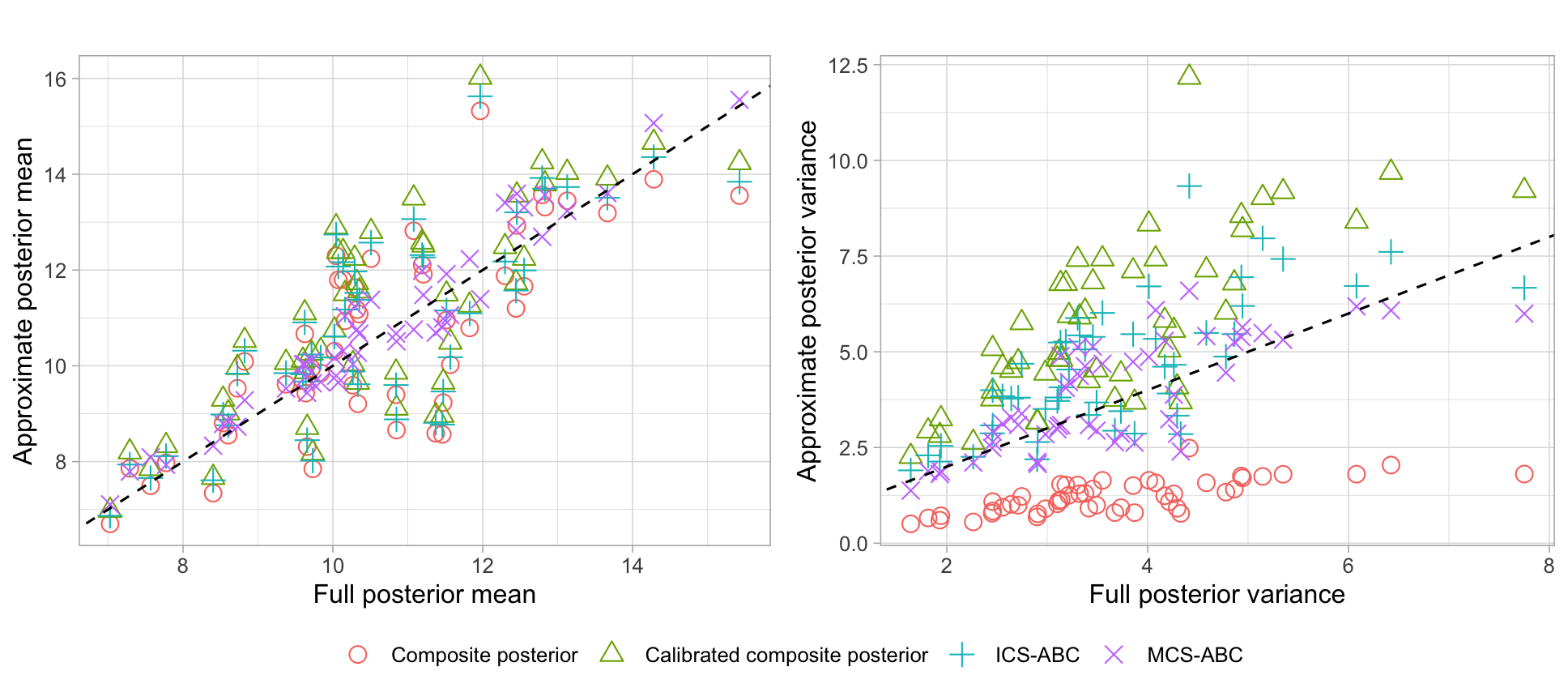}
\caption{Posterior means (left) and posterior variances (right) obtained from 50 replications. The dashed line shows the line $y=x$. }
\label{fig:spatial_rep}
\end{figure}

\begin{table}[!ht]
\[
\begin{array}{lcc}
\hline & \text { RMSE posterior mean } & \text { Average relative efficiency } \\
\hline \text { Composite posterior } & 1.302 & 3.129 \\
\text { Calibrated composite posterior } & 1.427 & 0.681 \\
\text { ICS-ABC } & 1.347 & 0.847 \\
\text { MCS-ABC } & \textbf{0.490} & \textbf{0.988} \\
\hline
\end{array}
\]
\caption{RMSEs of the approximate posterior mean estimates with respect to the full posterior mean, and the average relative efficiencies with respect to the full posterior variance. Results are obtained from 50 replications. The smallest RMSE and relative efficiency closest to 1 are highlighted in bold.  }
\label{tab:spatial_rep}
\end{table}

Our results show two benefits of the MCS-ABC method. First, including distant pairs does not negatively impact the quality of the inference. For $cL_{{\rm equal};d}(Y;\theta)$, as noted by \citet{bevilacqua2012estimating} and observed in our results for the three methods using it, the close pairs of observations contain more information about $\phi$. As $d$ is increased, more pairs of observations with greater distance are inextricably linked with and skew the information contained in the nearby pairs, resulting in poorer estimation. This issue is not observed with MCS-ABC. Second, MCS-ABC is able to make use of additional information to achieve an improved quality of inference. \citet{bevilacqua2012estimating} suggest choosing the $d$ that maximises the trace of the Godambe information matrix. For this example, this is $d=2$, providing the best estimation with the composite posterior methods and ICS-ABC. However, results for MCS-ABC suggest that, under the optimal combination, the pairs with distance $2 \leq ||s_i- s_j||< 3$ provide some additional information. Only the MCS-ABC method is able to make use of this information.

\subsection{ARFIMA(0,r,0) model}\label{sec:ARFIMA}

Suppose $Y_{obs}$ follows the ARFIMA($0,r,0$) model, defined by the equation
 \begin{align*}
(1 - B)^r Y_t = \epsilon_t, \qquad \epsilon_t \overset{i.i.d}{\sim} N(0, 1),
\end{align*}
where $(1-B)^r$ is the fractional differencing operator and $r\in(0,0.5)$ is the fractionally integrated component to be estimated. The process is stationary and exhibits long-range positive dependence with larger $r$ indicating longer-range dependence. 

Consider the set of lag-$k$ pairwise log-likelihoods $pl_k(Y_{obs};\theta)$ defined in \eqref{eq:lag_k_batched}. With the $d$th-order summed pairwise likelihood $pl_{{\rm equal};d}(Y;\theta)$ for some fixed $d$, \citet{davis2011comments} show that for $r\in (0,0.2)$, the MCLE is fairly efficient with the asymptotic variance similar to that of the MLE. However, there are two issues when $r\in[0.2,0.5)$. First, the convergence rate of the MCLE is slower than $\sqrt{n}$ when $r\in[0.25,0.5)$, hence the asymptotic efficiency of the MCLE relative to the MLE is 0 for $r\geq0.25$. Second, including the lag-$k$ pairwise likelihood with $k>1$ does not improve efficiency when $r\in[0.2,0.25)$, although, intuitively, pairs with a larger lag should be beneficial when $r$ is larger to account for the longer range of dependency. Despite these issues, numerical studies in \citet{davis2011comments} show that the $1$st-order summed pairwise likelihood can perform reasonably well for finite samples, even for $r\geq 0.25$.

Here we compare the performance of ICS-ABC using $ps_{\rm equal;d}(Y;\widehat{\theta})$ and MCS-ABC using $ps_{1:d}(Y;\widehat{\theta})$ as the summary statistics with different values of $d$, where $\widehat{\theta}$ is the MCLE of $r$ using $ps_{\rm equal;1}(Y_{obs};\theta)$. A Uniform $U(0,0.5)$ prior distribution is assigned to $r$. The results are reported in Figure \ref{fig:ARFIMA}. All methods produce fairly similar posterior mean estimates, hence these results are omitted. However, as shown in Figure \ref{fig:ARFIMA}, the choice of ABC summary has considerable effect on the posterior variance estimates. 

\begin{figure}[!ht]
\centering
\includegraphics[width=\textwidth]{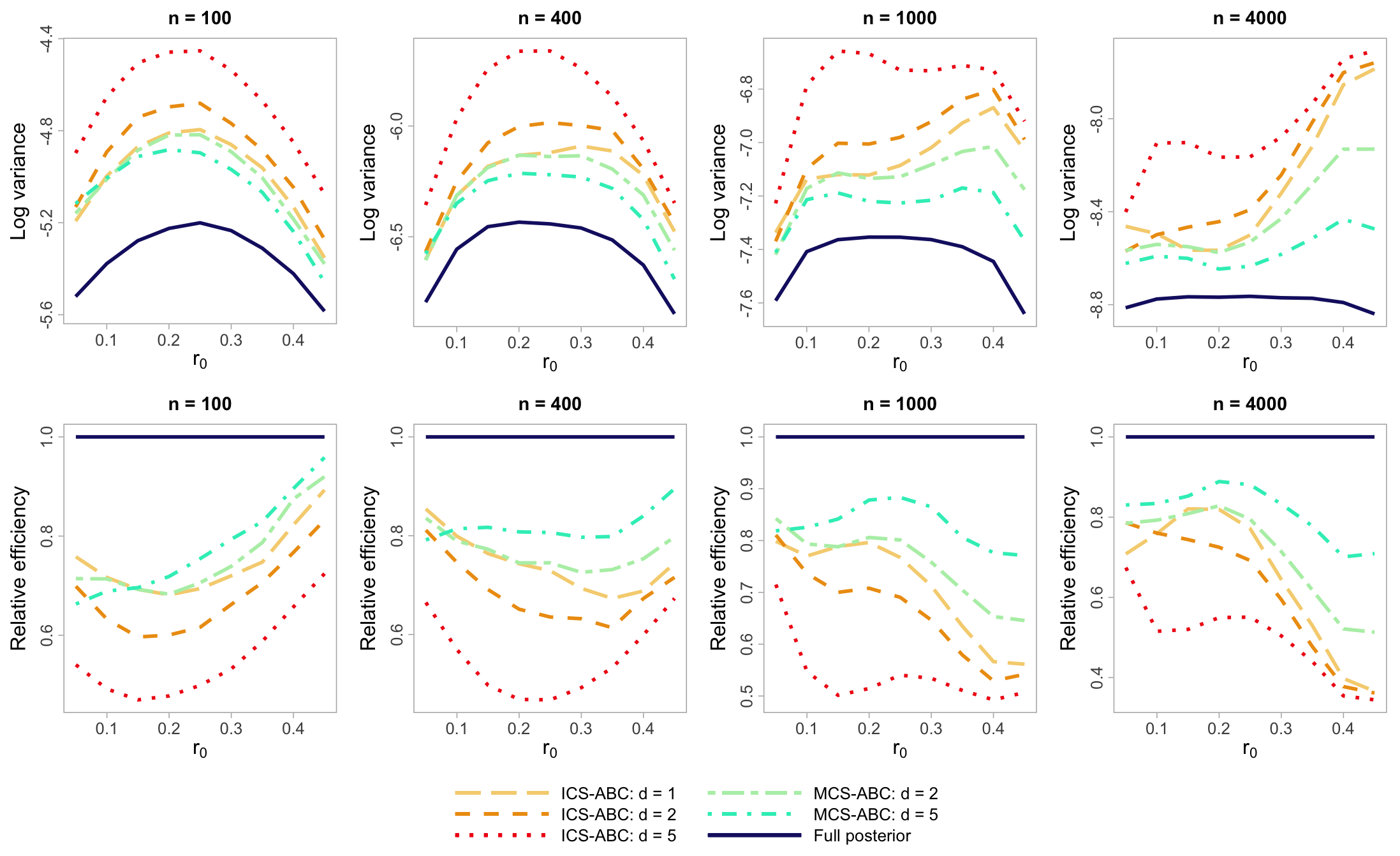}
\caption{Posterior log variances (top) and relative efficiencies with respect to the full posterior distribution (bottom). The results are averaged over 50 replicated datasets generated from the ARFIMA($0,r,0$) model for each of $r_0= 0.05, 0.10, 0.15, \dots, 0.45$ and $n = 100, 400, 1000, 4000$. }
\label{fig:ARFIMA}
\end{figure}

Our results show that ICS-ABC performs best with $d=1$ in almost all cases; including higher order lags decreases efficiency. In contrast, with MCS-ABC, increasing $d$ generally improves efficiency, up to $0.709/0.366=1.94$ times when $r_0=0.45$ and $n=4000$. This is expected from the optimal combination property. It also suggests that the counter-intuitive behaviour of ICS-ABC and the MCLE is due to the non-optimal combination using equal weights. Therefore, MCS-ABC solves the second issue mentioned above.   

Additionally, we observe that MCS-ABC almost always outperforms ICS-ABC with the same $d$, and the difference in performance becomes more significant as $r_0$ increases, up to more than doubling the efficiency with $d=5$ when $r_0=0.45$ and $n=4000$. Due to the first issue mentioned above, the relative efficiency of ICS-ABC deteriorates quickly when $r\geq0.25$ for larger data sizes. Although that of MCS-ABC also decreases when $r\geq0.25$ as $n$ increases, with higher-order composite components that capture the longer-range dependency, the loss of efficiency is much slower than that of ICS-ABC and relative efficiency does not fall below 0.7 using $d=5$. Therefore, MCS-ABC mitigates the first issue in finite samples. 



\subsection{Stochastic volatility model}\label{sec:SV}

Consider the stochastic volatility (SV) model:

$$
Y_t = \exp\left(\frac{X_t}{2}\right) S_t,\quad X_t =\gamma X_{t-1}+W_t,
$$
where $S_t\sim N(0,\sigma^2)$, $W_t\sim N(0, \tau^2)$, $X_0\sim N(0, \frac{\tau^2}{1-\gamma^2})$ and $|\gamma| < 1$. The parameters to be estimated are $\theta = (\tau^2, \gamma, \sigma^2)^T$. We assign inverse-gamma $InvGa(1,1)$ priors to the variance parameters $\tau^2$ and $\sigma^2$, and a uniform $U(-1,1)$ prior to the AR(1) coefficient $\gamma$. The observed dataset is generated with $\theta_0 = (0.25, 0.95, 1)^T$ and has length $n=1000$. 

Consider the set of lag-$k$ pairwise scores $ps_k(Y_{obs};\theta)$. All terms in the scores are two-dimensional integrals and have no closed form. Therefore, we estimate the terms using the method in Section \ref{sec:ssmodel}. The Monte Carlo sample size $N$ is chosen by using a pilot study to find the value of $N$ with which the estimation of a baseline composite score, e.g.~$ps_{{\rm equal};2}$, stabilises. Details can be seen in \citet{rosabeththesis}. $N=100$ is chosen.



The posterior distribution of $\theta$ is estimated using MCS-ABC with the summary statistic $\widehat{ps}_{1:m}(Y;\widehat{\theta})$ where $\widehat{\theta}$ is the MCLE of $ps_{{\rm equal};2}(Y_{obs};\theta)$.  
The ABC posterior distributions obtained using different $d$ are shown in Figure \ref{fig:SVnestmarg}. Our results show clear improvements in the ABC posterior as $d$ increases for $d\in[1, 4]$, especially for $\tau^2$ and $\gamma$, and using $d=4$ provides a good approximation to the full posterior. The results suggest that pairs with lag $\geq5$ do not provide significant additional information. However, similar to Figure \ref{fig:spatial}, including high order lags does not have a negative impact on the quality of the inference. These results are due to the optimal combination property of MCS-ABC. In contrast, results in \citet{varin2008pairwise} show that estimation quality can decrease when $d$ is increased above a threshold value in application of the (equally weighted) MCLE to state-space models.

\begin{figure}[!ht]
\centering
\includegraphics[width=\textwidth]{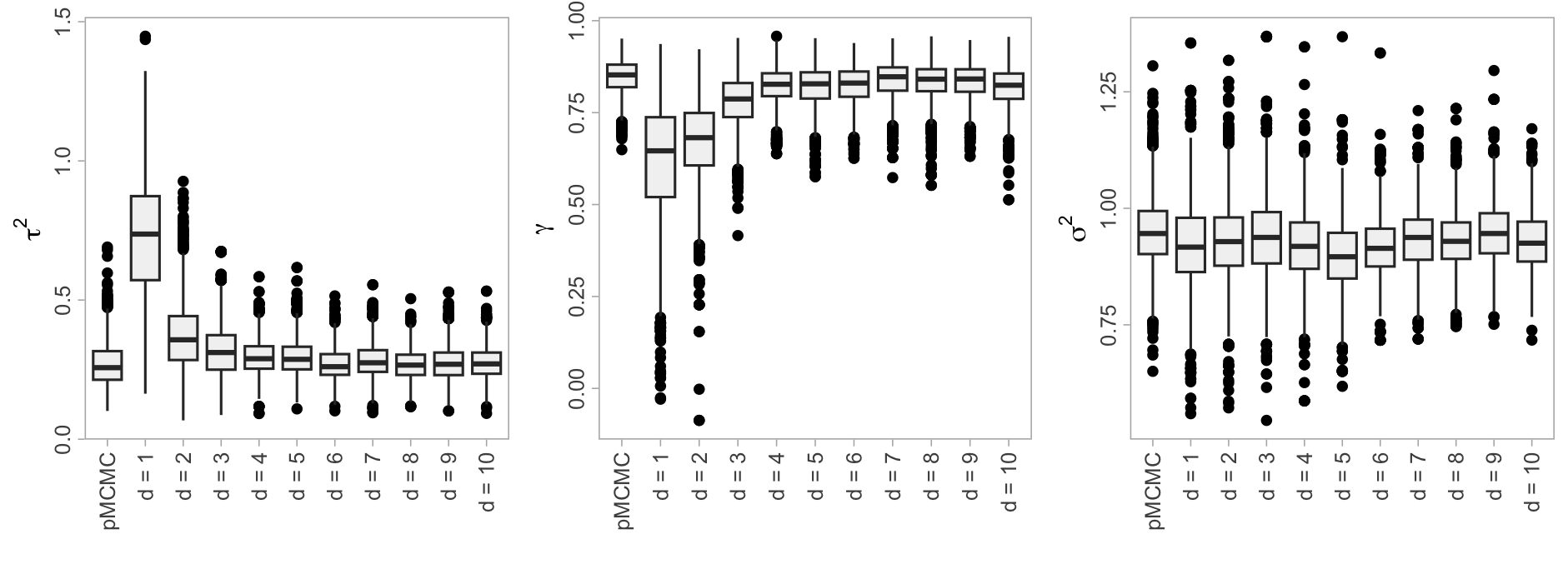}
\caption{ABC posterior marginals for parameters of SV model using $\widehat{ps}_{1:m}(Y;\widehat{\theta})$ as the summary, and the full posterior obtained by pMCMC.}
\label{fig:SVnestmarg}
\end{figure}

In order to obtain the best possible approximation with MCS-ABC, we need to choose a suitable $d$. Typically, this choice of $d$ will not be known beforehand and will depend on the unknown true parameter value. Algorithm \ref{algo:adaptation} is applied to choose $d$. Its performance is compared at the same computing cost with the ABC-PMC algorithm using a pre-specified, fixed $d$. The comparison is reported in Figure \ref{fig:SVnestcomp}. 

\begin{figure}[!ht]
\centering
\includegraphics[width=\textwidth]{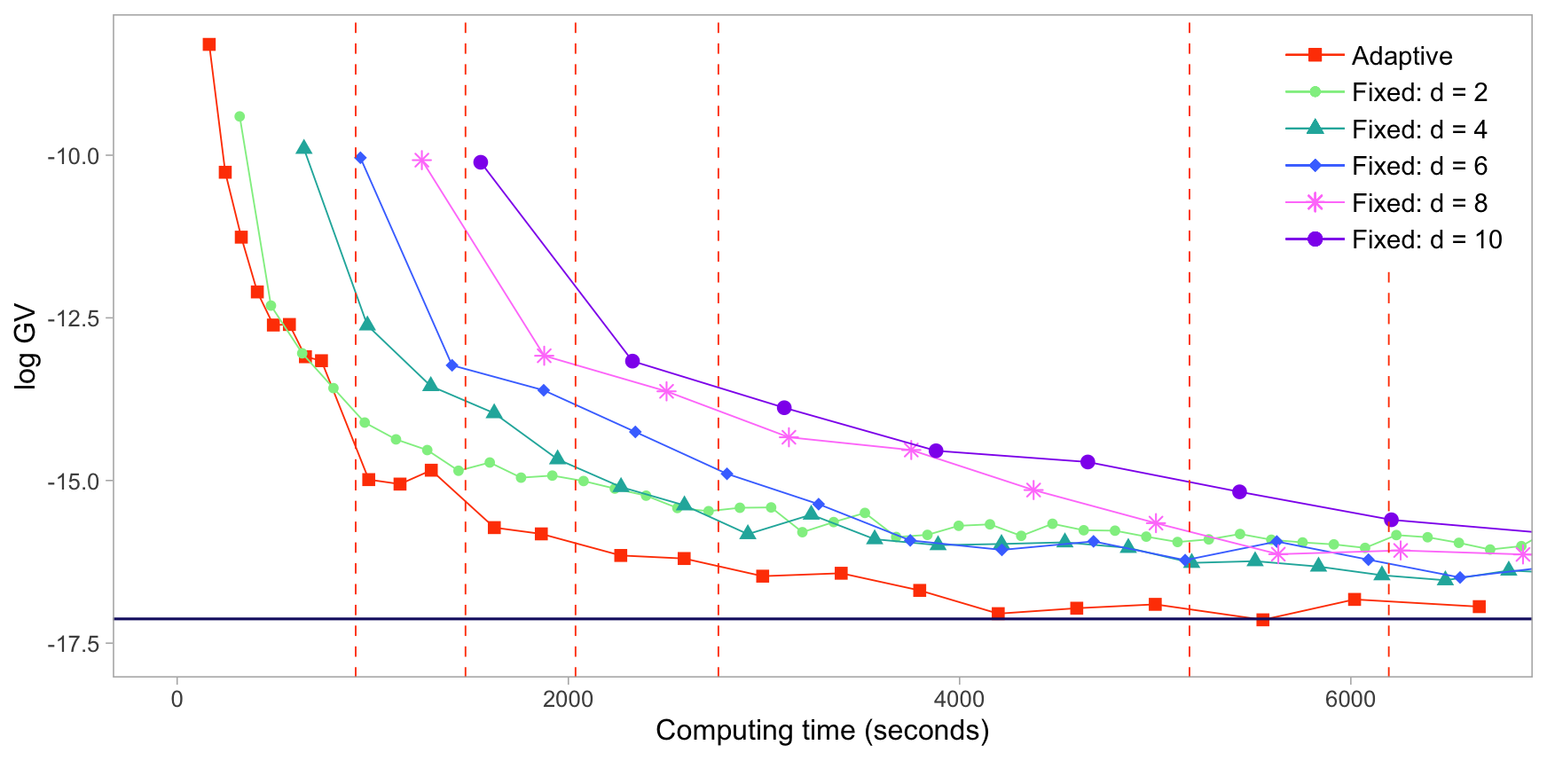}
\caption{ABC posterior $\log GV$ obtained with Adaptive MCS-ABC and fixed $d$ MCS-ABC compared by computing time. Each point shows the variance estimate at the end of a PMC cycle. The horizontal line indicates the posterior $\log GV$ obtained with pMCMC. The vertical lines indicate a change of summary in Algorithm \ref{algo:adaptation}. }\label{fig:SVnestcomp}
\end{figure}

Algorithm \ref{algo:adaptation} terminates after including the lag $d=7$ pairwise component. The final posterior variance estimate is similar to that obtained with pMCMC, and, based on the results in Figure \ref{fig:SVnestmarg}, we do not expect the inclusion of additional composite components to improve estimation. Therefore, the algorithm has successfully chosen the summary statistic close to a sufficient statistic. Moreover, for a fixed computational budget, the posterior variance estimate obtained by Algorithm \ref{algo:adaptation} is more accurate than the fixed summary algorithm with any $d$ considered. For example, with the most expensive summary considered ($d=10$), the fixed summary algorithm has only completed 1 PMC cycle within 2000 seconds, whereas Algorithm \ref{algo:adaptation}, which uses the cheapest summaries in the early PMC cycles, has completed 13 cycles and achieved a substantially more accurate variance estimate. Therefore even if the choice of a suitable $d$ is known beforehand, Algorithm \ref{algo:adaptation} can still be applied to significantly reduce the computational cost.

\section{Discussion} 


This paper proposes to combine multiple composite scores in the Bayesian framework to achieve the maximal Godambe information of the optimal estimating function in $\mathcal{H}_{cs_{1:K}}$. It proposes a novel algorithm to estimate and select the component composite scores to conditional on and provides justification based on asymptotic theory. This approach avoids evaluating the asymptotic moments of component composite scores which is often challenging. 

For time series models, composite likelihood is a computationally convenient alternative to the full likelihood giving fairly accurate point estimators, with the price of loss of efficiency due to ignoring correlations between sub-likelihoods. Our numerical studies show that the optimally combination of lag-batched composite scores significantly reduces the loss of efficiency and avoids it in several cases with enough lags by accounting for the correlations between component composite scores in the summary-based posterior distribution. The computational complexity is $O(n)$, making the proposed ABC method an computationally attractive alternative to the popular particle MCMC for long time series. 

When the chosen $\widehat{\theta}$ is a poor estimator of $\theta_0$ or for the question whether $\widehat{\theta}$ is needed at all, the following class of estimating functions is of interest, $\mathcal{H}_{cs_{1:K}}':=\{W(\theta)[cs_{1:K}(Y_{\rm obs};\theta_1)-E_{\theta}\{cs_{1:K}(Y;\theta_1)\}]:W(\theta)\in\mathbb{R}^{p\times d}\}$ for some constant $\theta_1$. Its maximal Godambe information is
\[
I_{G'}(h):=c_{n}^2\nabla_{\theta}E_{\theta}\{cs_{1:K}(Y;\theta_1)\}Var\{cs_{1:K}(Y_{obs};\theta_1)\}^{-1}\nabla_{\theta}E_{\theta}\{cs_{1:K}(Y;\theta_1)\}^{T}\mid_{\theta=\theta_{0}},
\]
and achieved by $\pi\{\theta\mid cs_{1:K}(Y_{\rm obs};\theta_1)\}$ in the first order. The impact of the value of $\theta_1$ requires further research. When $\theta_1=\theta_0$, obviously $I_{G'}(h)=I_{G}(h)$. When $\theta_1\neq\theta_0$, since each sub-likelihood is the likelihood of the sub-event, $\nabla_{\theta}E_{\theta}\{cs_{1:K}(Y;\theta_1)\}\mid_{\theta=\theta_{0}}\leq E_{\theta_0}\{\nabla_{\beta}cs_{1:K}(Y;\theta_0)\}$ by the Cram\'er-Rao inequality. However, due to that the sub-scores are correlated, the comparison of $Var\{cs_{1:K}(Y_{obs};\theta_1)\}$ and $Var\{cs_{1:K}(Y_{obs};\theta_0)\}$, hence $I_{G'}(h)$ and $I_{G}(h)$, is case dependent.   


\printbibliography

\newpage

\newrefsection 

\appendix
\part{Supplement} 
\parttoc 

\section{Related composite likelihood methods}\label{sub:cl_methods}

Given a composite likelihood $cL(Y_{obs};\theta)$, the maximum composite likelihood estimate (MCLE) is 
$$\widehat{\theta}_c:=\argmax_{\theta}cL(Y_{obs};\theta).$$ 
When the dataset contains $n$ independent and identically distributed observations, standard asymptotic results show that $\widehat{\theta}_c$ is consistent and asymptotic normal with variance $n^{-1}G(\theta_0)^{-1}$, where $G(\theta):=H(\theta)^TJ(\theta)^{-1}H(\theta)$ is the Godambe information for a single observation, with $H(\theta):=E\{-\nabla cs(Y;\theta)\}$ and $J(\theta):=Var\{cs(Y;\theta)\}$ \citep{varin2011overview}. Generally, composite likelihoods do not satisfy the second Bartlett equality, and this can result in a loss of efficiency compared to the maximum likelihood estimation. 

Asymptotic properties of the MCLE are also available for dependent data in some cases. In time series settings, \citet{ng2011composite} and \citet{davis2011comments} provide conditions for consistency and asymptotic normality of composite likelihood estimators for stationary time series with both short and long-range dependence. For spatial data, \citet{wang2013partial} discuss asymptotic properties of the MCLE as the size of the geographic sampling area increases. In this case, the MCLE has desirable asymptotic properties under assumptions that the spatial dependence decays sufficiently quickly as the geographic distance increases. In this work we assume the marginal events in sub-likelihoods are identically distributed and not necessarily independent, which encompasses both asymptotic cases. Therefore, the asymptotic results for both independent and dependent data can be applied in our framework.  

In Bayesian framework, the composite posterior density $\pi_c(\theta ; Y_{obs}) \propto \pi(\theta) c L(Y_{obs};\theta)$ can be estimated using, for example, MCMC. \citet{smith2009extended} and \citet{chan2017performance} apply this approach with the pairwise likelihood to max-stable processes for modelling spatial extremes. However, the direct replacement of the full likelihood with a composite likelihood can lead to underestimation of the posterior uncertainty \citep{ribatet2012bayesian}. To address this, \citet{pauli2011bayesian} and \citet{ribatet2012bayesian} propose calibration methods to adjust the curvature of the composite likelihood to correct its asymptotic properties. \citet{pauli2011bayesian} propose to use the weighted composite log-likelihood, 
$$cl_{adj}(Y_{obs};\theta)=\sum_{j=1}^{N_{c l}} \frac{\log f\left(A_j ; \theta\right)}{\widetilde{\lambda}}\mbox{, where }\widetilde{\lambda}=\operatorname{tr}\{H(\widehat{\theta}_c)^{-1} J(\widehat{\theta}_c)\}/p.$$
Calibration methods generally improve estimates of posterior variability. However, the calibrated composite posterior distribution can be over-dispersed, resulting in an overestimation of posterior variability \citep{ruli2016approximate}. \citet{ruli2016approximate} propose to use composite likelihoods with ABC by using $J(\widehat{\theta}_c)^{-1/2}cs(Y;\widehat{\theta}_c)$ as the summary statistic for $Y$ simulated with the parameter $\theta$. The covariance matrix of the observed composite score is used to weight the ABC distance to account for the potentially differing variances of terms in the score \citep{gleim2013approximate}. \citet{ruli2016approximate} show in numerical studies that their ABC scheme outperforms the calibrated composite posterior method of \citet{pauli2011bayesian}. 

For all methods above, a suitable composite likelihood, including the associated weights, must first be chosen. For many models there is no clear choice for the weights, or there are several candidate composite likelihoods to choose from.

\section{Standard ABC algorithms} \label{ABC_algos}

\begin{center}
\begin{algorithm}[H]
\caption{Importance sampling ABC}
\vspace{0.2cm}
\small{
{Suppose a proposal density $q(x)$, a bandwidth $\varepsilon$ and a Monte Carlo size $N$ are given.}
{\begin{enumerate}\setlength{\itemsep}{0.1em}
\item Draw $\theta^{(1)}, \dots, \theta^{(N)} \sim q(\cdot)$.

\item For $i=1, \dots, N$, simulate $\boldsymbol{Y}^{(i)} \sim f\left(\cdot ; \theta^{(i)}\right)$.

\item For $i=1, \dots, N$, calculate $\boldsymbol{s}^{(i)}$ and accept $\theta^{(i)}$ if $d\left(\boldsymbol{s}^{(i)}, \boldsymbol{s}_{\mathrm{obs}}\right)<\varepsilon$. Define the associated weights as $\pi(\theta^{(i)})/q(\theta^{(i)})$.

\end{enumerate}}}
\end{algorithm}
\end{center}

\begin{center}
\begin{algorithm}[H]
\caption{ABC-population Monte Carlo \citep{lenormand2013adaptive}}
\vspace{0.2cm}
\small{
{Choose $N>0$ and $\alpha \in (0,1)$, where $N_{\alpha}:= N\alpha$ is the number of samples to keep at each PMC step.}
{\begin{enumerate}\setlength{\itemsep}{0.1em}
\item Set $t=1$. For $i = 1, \dots, N$\\
Simulate $\theta^{(i,0)} \sim \pi(\theta)$ and $\bm{y}^{(i,0)}\sim f(\cdot; \theta^{(i,0)})$. Compute $d^{(i,0)} = d(S(\bm{y}^{(i,0)}), S(\bm{y}_{\text{obs}}))$. Set $w^{(i,0)} = 1$.
\item Let $\epsilon_1$ be the first $\alpha$-quantile of $\{d^{(i,0)}\}_{i = 1, \dots,N}$.
\vspace{0.05cm}
Let $\left\{(\theta^{(i,1)}, w^{(i,1)}, d^{(i,1)}, \bm{y}^{(i,1)}, \bm{s}^{(i,1)})\right\}_{i = 1, \dots, N_\alpha} = $\\
\vspace{0.15cm}
$\qquad \left\{(\theta^{(i,0)}, w^{(i,0)}, d^{(i,0)}, \bm{y}^{(i,0)}, S(\bm{y}^{(i,0)}))| d^{(i,0)} \leq \epsilon_1\right\}_{i = 1,\dots,N} $\\
\vspace{0.15cm}
Set $\Sigma^{(1)}$ to be twice the weighted covariance of $\{(\theta^{(i,1)}, w^{(i,1)})\}_{i = 1, \dots, N_\alpha}$\\
Set $t = t +1$.
\item For $i = N_{\alpha} +1, \dots, N$\\
Sample $\theta^{*}$ from $\{\theta^{(j,{t-1})}\}_{j=1,\dots, N_\alpha}$ with probability $\{w^{(j_{t-1})} / \sum_{k=1}^{N_{\alpha}} w^{(k_{t-1})}\}_{j = 1, \dots, N_{\alpha}}$\\
Generate $\theta^{(i,{t-1})} \sim N(\theta^{*}, \Sigma^{(t-1)})$ and $\bm{y}^{(i,{t-1})} \sim f(\cdot ; \theta^{(i,{t-1})})$\\
Compute $d^{(i,{t-1})} = d(S(\bm{y}^{(i,{t-1})}), S(\bm{y}_{\text{obs}}))$
\begin{flalign*}
\text{Set } w^{(i,{t-1})} = \frac{\pi(\theta^{(i,{t-1})})}{\sum_{j=1}^{N_{\alpha}}(w^{(j,{t-1})}/\sum_{k=1}^{N_\alpha}w^{(k,{t-1})})\zeta(\theta^{(i,{t-1})}-\theta^{(j,{t-1})}, \Sigma^{(t-1)})}&&
\end{flalign*}
where $\zeta(x, A) = (2\pi)^{-p/2}\det(A)^{-1/2}e^{-\frac{1}{2}(x^TA^{-1}x)}.$
\item Let $\epsilon_t$ be the first $\alpha$-quantile of $\{d^{(i,{t-1})}\}_{i = 1, \dots, N}$\\
\vspace{0.05cm}
Let $\left\{(\theta^{(i,t)}, w^{(i,t)}, d^{(i,t)}, \bm{y}^{(i,{t-1})}, \bm{s}^{(i,t)})\right\}_{i=1, \dots, N_\alpha} = $\\
\vspace{0.15cm}
$\qquad \left\{(\theta^{(i,{t-1})}, w^{(i,{t-1})}, d^{(i,{t-1})}, \bm{y}^{(i,{t-1})}, S(\bm{y}^{(i,{t-1})}))| d^{(i_{t-1})} \leq \epsilon_t \right\}_{i=1, \dots, N}$\\
\vspace{0.15cm}
Set $\Sigma^{(t)}$ to be twice the weighted covariance of $\{(\theta^{(i,t)}, w^{(i,t)})\}_{i=1, \dots, N_\alpha}$\\
Set $t = t +1$.
\end{enumerate}}}
\label{algo:nested}
\end{algorithm}
\end{center}
In \citet{lenormand2013adaptive} the authors use a stopping rule based on the proportion of newly accepted particles that aims to terminate the algorithm when further simulations would result in very little change to the estimated posterior.

\section{Technical conditions of results in Section \ref{sub:abc_est_scores} and discussions}

\subsection{Theorem \ref{thm:optimal_comb}}

For a dataset containing i.i.d observations, \ref{cond:sum_conv} and \ref{cond:limit_moments} hold alongside the corresponding central limit theorems and moment convergence, which can be verified by standard arguments in, e.g.~\citet{serfling2009approximation}. For a dataset with dependent observations, in the literature, \ref{cond:sum_conv} and \ref{cond:limit_moments} are often developed to obtain the asymptotic normality of the maximum composite likelihood estimator \citep{bevilacqua2012estimating,heagerty1998composite}. For example, \citet{davis2011comments} consider the estimation of a class of linear time series models, including ARMA, ARIMA processes, etc., using the $m$th-order summed pairwise likelihood. They discuss sufficient conditions for \ref{cond:sum_conv} to hold and give the analytical forms of the moment in \ref{cond:limit_moments}. They show that for both long and short-memory models, the pairwise likelihood is inferior to the $m$th-order summed pairwise likelihood with fixed $m$, since most of the dependence occurs in pairs that are close to each other and the pairwise likelihood has too many redundant pairs of observations when $n$ is large. However, the asymptotic variance of the maximum composite likelihood estimator is not monotone in $m$, and it is not clear how to choose $m$, i.e.~weights of which component composite likelihood should be $0$. As another example, \citet{ng2011composite} consider a class of state-space models with latent Gaussian autoregressive processes. They discuss the asymptotics of the $m$th-order summed pairwise log-likelihood and the \textit{$k$-consecutive log-likelihood}, defined as
\begin{align}
tl_{k}(Y;\beta):=(n-k)^{-1}\sum_{t=1}^{n-k}\log f(Y_{t},\dots,Y_{t+k};\beta). \label{eq:consecutive_batched}	
\end{align}
 For $pl_{1:k}(Y;\beta)$ and $tl_{k}(Y;\beta)$, they provide conditions for \ref{cond:sum_conv} to hold and prove the existence of $A(\theta)$ and $\mu(\theta)$ in \ref{cond:limit_moments}. In particular, \ref{cond:sum_conv}(ii) is satisfied if for some $k\in\{1,\dots,K\}$, the Fisher information $E_{\theta_0}\{\Delta_{\beta}\log f_k(Y_{k1})\}$ is of full rank.

\subsection{Theorem \ref{thm:estimated_score}}
Let $\mathcal{Y}_k$ be the union of supports of $f_k(\cdot;\theta)$ over $\theta\in\mathcal{P}$. Let $\mathcal{X}\subset\mathbb{R}^l$ be the union of supports of $f_X(\cdot;\theta)$ over $\theta\in\mathcal{P}$ and $\mathcal{X}_k\subset\mathbb{R}^{l_k}$ be that for $f_{k,X}(\cdot;\theta)$. Let $\eta_{y}(x;\beta)=f_k(y| x;\beta)/f_k(y;\beta)$. For all $k=1,\dots,K$ and a small neighbourhood $B_{\theta_0}\subset\mathcal{P}_0$ of $\theta_0$, we assume the following conditions on several probability densities of $X$ and $Y$. 
\begin{description}
	\item[(C4)\label{cond:differentiability}] $q(x) \in C^{1}(\mathcal{X}), f_k(y, x ; \beta) \in C^{1}(\mathcal{Y}_k \times \mathcal{X}_k \times \mathcal{P}_0)$, and $f_k(y;\beta) \in C^{1}(\mathcal{Y}_k)$, for all $\beta \in \mathcal{P}_0$.
	\item[(C5)\label{cond:fk_upper_bound}] $\exists$ $c_1$, $c_2$ and $c_3$ such that $\sup _{x, y, \beta} f_k(x| y ; \beta)<c_1$, $\sup _{x, y, \beta}\left\|\nabla_x f_k(x| y ; \beta)\right\|<c_2$ and $\sup _{x, y, \beta}\left\|\nabla_{\beta} f_k(x| y ; \beta)\right\|<c_3$.
	\item[(C6)\label{cond:fk_information}]  $\sup_{\theta,\beta\in\mathcal{P}_0}E_\theta\left\{\left\|\nabla_\beta \log f_k(Y_{k1} ; \beta)\right\|^2\right\}<\infty$.
	\item[(C7)\label{cond:Q_Donsker}] $\exists$ a partition $\mathbb{R}^{l_k}=\bigcup_j I_j$, where $I_j$ are unit cubes, such that
$$
\sum_{j=1}^{\infty}\left[\left\{\sup _{x \in I_j}\left\|\nabla_x \log q_k(x)\right\|+1\right\} \sup_{x \in I_j}q_k(x)^{-1} Q(I_j)\right]^{\frac{l_k}{l_k+1}}<\infty,
$$
where $Q(I_j)=\int_{I_j}q_k(x)\,dx$. 
	\item[(C8)\label{cond:IS_q}] $\sup_{x\in\mathcal{X}_k,\beta\in B_{\theta_0}}f_k(x;\beta)/q_k(x)<M$ where $M$ is a positive constant. 	
	\item[(C9)\label{cond:envelop1}]  $E_{q(\cdot)f_k(\cdot;\theta)}\left\{ \sup_{\beta\in B_{\theta_0}}\eta_{Y_{k1}}(X_{k1};\beta)^{4} \right\}<\infty$ for all $\theta\in\mathcal{P}_0$,
	\item[(C10)\label{cond:envelop2}]  $E_{q}\left[\sup_{\theta\in\mathcal{P}_{0}}E_{\theta}\left\{ \nabla_{\beta}\log f_{k}(Y_{k1};\theta)\sup_{\beta\in B_{\theta_0}}\nabla_{\beta}\log f_{k}(Y_{k1},X_{k1};\beta)^{T}\eta_{Y_{k1}}(X_{k1};\beta)\right\} ^{2}\right]<\infty$.
	\item[(C11)\label{cond:average_joint_information}] $\lim _{n \rightarrow \infty} a_n^2\Var_{q(\cdot)f_k(\cdot; \theta)}\left\{\frac{1}{n_k} \sum_{j=1}^{n_k} \nabla_\beta \log f\left(Y_{k j}, X_{k1} ; \beta\right) \eta_{Y_{k1}}(X_{k1};\beta)\right\}<\infty$ for all $\theta\in\mathcal{P}_0$ and $\beta\in B_{\theta_0}$.
	
\end{description}

Conditions \ref{cond:differentiability} and \ref{cond:fk_upper_bound}
are standard regularity conditions for probability densities in fixed-dimensional spaces. Our results will depend on the uniform convergence of some empirical processes of $X_{1:N}$ over some function classes which requires \ref{cond:fk_information} and \ref{cond:Q_Donsker}. \ref{cond:fk_information} assumes finite second moments on the score function of the marginal density of $Y_{k1}$. \ref{cond:Q_Donsker} assumes that the user-chosen $q(x)$ satisfies a standard bracketing number condition for Donsker classes. For $q_k(x)$ being a continuous probability density, since $Q(I_j)\approx q_k(x_j^*)$ for some $x_j^*\in I_j$, \ref{cond:Q_Donsker} requires that $q_k(x)$ decreases sufficiently fast in the tails. It is easy to check that the Gaussian family satisfies \ref{cond:Q_Donsker} and the Student's t family does not. \ref{cond:IS_q}--\ref{cond:envelop2} are needed to control the variance of the importance sampling estimators in \eqref{eq:est_fisher_identity}. The importance weights in \eqref{eq:est_fisher_identity} are in the form of $f_{k}(y,x;\beta)/\{f_{k}(y;\beta)q_k(x)\}$, up to a remainder, and can be written as $\eta_y(x;\beta)\cdot f_{k}(x;\beta)/q_k(x)$. \ref{cond:IS_q} requires that the proposal density $q_k(x)$ has equal or heavier tails than $f_k(x;\beta)$. Hence it is natural to use $f_X(x;\widehat{\theta})$ as the proposal density $q(x)$  with which the importance weights in \eqref{eq:est_fisher_identity} have the simpler form $f_{k}(Y_{kj}| x_{kj}^{(i)};\widehat{\theta})/\widehat{f}_{k}(Y_{kj};\widehat{\theta})$. \ref{cond:envelop1} and \ref{cond:envelop2} assume that on average the importance sampling estimators using $\eta_y(x;\beta)$ as the importance weights over $x$ have bounded variances, and these can be expected to hold, since in $\eta_y(x;\beta)$ the proposal density $f_k(y;\beta)$ is a continuous mixture density of the target density $f_k(y| x;\beta)$ over $x$. \ref{cond:average_joint_information} depends on the weak convergence property of the component composite scores and can be verified as a byproduct when obtaining the asymptotic covariance matrix of the maximum composite likelihood estimator, e.g.~in \citet{ng2011composite}.

\section{Proof of Theorem \ref{thm:optimal_comb}}
Since the summary-based posterior $\pi\{\theta| h(Y_0)\}$ is invariant to the scaling of $h(Y)$ by a constant diagonal matrix, without loss of generality hereafter assume the weights $\{w_{kj}\}_{j=1}^{n_k}$, $k=1,\dots,K$, of the sub-scores satisfy $\sum_{j=1}^{n_k}w_{kj}/n_k=1$.

\begin{proof}[Proof of Theorem \ref{thm:optimal_comb}] 
Since $h(Y;\theta)$ is an unbiased estimating function, $\mu(\theta_0)=\lim_{n\rightarrow\infty}E_{\theta_0}\{h(Y)\}=0.$ By \citet[Proposition 1]{li2018convergence}, the asymptotic variance
of $\pi\{\theta| h(Y_{obs})\}$ is the inverse of $a_{n}^{2}I(h)$
if \ref{cond:prior}, \ref{cond:sum_conv}, \citet[Condition 4 and 5]{li2018convergence}  and the following conditions hold: 
\begin{description}
\item [{(i)}] For any $\delta<\delta_{0}$ there exists $\delta'>0$ such
that $\max_{\|\theta-\theta_{0}\|>\delta}\|\mu(\theta)\|>\delta'$.
\item [{(ii)}] $a_n^2I(h)$ has full rank.
\end{description}
For (i) above, if it doesn't hold, there exists a series $\theta_{t}\rightarrow\theta_{0}$,
as $t\rightarrow\infty$, such that $\{\theta_{t}\}\backslash\theta_{0}$
is not empty and $\mu(\theta_{t})=0$ for all $t$. By the Taylor expansion and $\mu(\theta_{0})=0$, $\mu(\theta_{t})=\nabla\mu(\dot{\theta}_{t})^{T}(\theta_{t}-\theta_{0})$, where $\dot{\theta}_{t}$ is between $\theta_{t}$
and $\theta_{0}$. It implies that for each $t$, the rank of $\nabla\mu(\dot{\theta}_{t})$, hence $\nabla\mu(\theta_{0})$, is smaller than $p$. On the other hand, an important property of unbiased estimating function is $\nabla E_{\theta}\{h(Y;\theta)\}=E_{\theta}\{\nabla h(Y;\theta)\}$, which can be obtained by differentiating under the integral and the interchange of the order of integration and differentiation is permitted using Leibniz integral rule \citep[Section 2.3]{sen2010finite}. Then we have  
$$
\nabla\mu(\theta_0)=\lim_{n\rightarrow\infty}\nabla E_{\theta}\{h(Y;\theta)\}\mid_{\theta=\theta_0}=\lim_{n\rightarrow\infty} E_{\theta}\{\nabla h(Y;\theta)\}\mid_{\theta=\theta_0}.
$$
and by \ref{cond:sum_conv}(ii) there is contradiction. Therefore (i) holds, and by \ref{cond:sum_conv}(ii) again, (ii) holds. 

Hence by \ref{cond:limit_moments},
\begin{align*}
a_n^2I(h) & =\lim_{n\rightarrow\infty}E_{\theta_{0}}\{\nabla_{\beta}h(Y)\}^{T}Var_{\theta_0}\{h(Y)\}^{-1}E_{\theta_{0}}\{\nabla_{\beta}h(Y)\}\\
 & =\lim_{n\rightarrow\infty}a_n^2I_{G}(W^{*}h).
\end{align*}

\end{proof}

\section{Proof of Theorem \ref{thm:estimated_score}}
\begin{lem}
\label{lem:Vaart}\citep{van2000asymptotic} Consider a partition
$\mathbb{R}^{d}=\bigcup_{j=1}^{\infty}I_{j}$, where $I_{j}$ are
unit cubes. Let $\mathcal{F}$ be the class of all functions $f:\mathbb{R}^{d}\rightarrow\mathbb{R}$
satisfying that $f(x)\in C^{1}(\mathbb{R}^{d})$ and $\|\sup_{x\in I_{j}}\nabla f(x)\|\leq M_{j}$
for some constants $M_{j}<\infty$, $\forall j$. If 
\begin{equation}
\sum_{j=1}^{\infty}\{M_{j}Q(I_{j})\}^{\frac{d}{d+1}}<\infty,\label{eq:bracketing_bound}
\end{equation}
then $\mathcal{F}$ is a Q-Donsker class.
\end{lem}
\begin{proof}
This holds by letting $r=1$, $V=d$ and $\alpha=1$ in \citet[Ex19.9]{van2000asymptotic}.
\end{proof}
For each $k=1,\dots,K$, let $H:=\{h\in L_{2}(f_{k,\theta}):\theta\in\mathcal{P}_{0}\}$
where $L_{2}(f_{k,\theta})$ contains all random variables following
$f_{k}(\cdot;\theta)$ with finite $L_{2}$ norm. For $h\in H$ and
some function $g(\cdot)$ on H, let $\|g(h)\|_{H}:=\sup_{h\in H}\|g(h)\|_{\theta}$
where $\|\cdot\|_{\theta}$ is the the $L_{2}$ norm under $f_{k}(\cdot;\theta)$.
Let $\psi_{y}(x;\beta)=f_{k}(x| y;\beta)/q_k(x)$.

\begin{cor}
\label{cor:IS_weight_Donsker}Assume \ref{cond:differentiability}--\ref{cond:Q_Donsker}. Then $\{\psi_{\tilde{Y}}(x;\beta):\tilde{Y}\in H,\beta\in\mathcal{P}_{0}\}$
is a Q-Donsker class.
\end{cor}
\begin{proof}
Let $M_{j}=\sup_{x\in I_{j},\beta}\left\{ \left\Vert \nabla_{x}\psi_{\tilde{Y}}(x;\beta)\right\Vert _{H}\right\} $.
By \ref{cond:fk_upper_bound}, for $x\in I_{j}$, we have 
\begin{align*}
\left\Vert \nabla_{x}\psi_{\tilde{Y}}(x;\beta)\right\Vert _{H} & =\sup_{\tilde{Y}\in H}\left\Vert \nabla_{x}\frac{f_{k}(x\mid\tilde{Y};\beta)}{q_{k}(x)}\right\Vert _{\theta}\\
 & =\sup_{\tilde{Y}\in H}\left\Vert \frac{f_{k}(x\mid\tilde{Y};\beta)}{q_{k}(x)}\left\{ \nabla_{x}\log f_{k}(x\mid\tilde{Y};\beta)-\nabla_{x}\log q_{k}(x)\right\} \right\Vert _{\theta}\\
 & \leq c_{1}\{c_{2}+\sup_{x\in I_{j}}\|\nabla_{x}\log q_{k}(x)\|\}\sup_{x\in I_{j}}q_{k}(x)^{-1}.
\end{align*}
Then by \ref{cond:Q_Donsker}, we have $\sum_{j=1}^{\infty}\{M_{j}Q(I_{j})\}^{\frac{l_{k}}{l_{k}+1}}<\infty$
hence $M_{j}$ exists, and by Lemma \ref{lem:Vaart}, $\mathcal{F}$
in Lemma \ref{lem:Vaart} is a Q-Donsker class. By \ref{cond:differentiability}, the target
set is a subset of $\mathcal{F}$ and therefore the corollary holds. 
\end{proof}

Define $u(y,x;\beta):=\nabla_{\beta}\log f_{k}(y,x;\beta)f_{k}(x| y;\beta)/q_{k}(x)$
and $v(x;\theta,\beta):=E_{\theta}\{u(Y_{k1},x;\beta)\}$. 
\begin{lem}\label{lem:v_class_Donsker} Assume \ref{cond:differentiability}--\ref{cond:fk_information}.
Then $\{v(x;\theta,\beta)-v(x;\theta_{0},\beta)\in\mathbb{R}^{p}:\theta\in\mathcal{P}_{0},\beta\in\mathcal{P}_{0}\}$
is a Q-Donsker class.
\end{lem}
\begin{proof}
Similar to Corollary \ref{cor:IS_weight_Donsker}, this can be proved
by finding upper bounds of the derivatives of functions in the class
in $I_{j}$ and applying Lemma \ref{lem:Vaart}. By \ref{cond:differentiability}, $v(x;\theta,\beta)-v(x;\theta_{0},\beta)$
is differentiable. By the Taylor expansion, 
\[
v(x;\theta,\beta)-v(x;\theta_{0},\beta)=E_{\dot{\theta}}\left\{ \nabla_{\beta}\log f_{k}(Y;\dot{\theta})u(Y,x;\beta)^{T}\right\} (\theta-\theta_{0}),
\]
where $\dot{\theta}$ is between $\theta$ and $\theta_{0}$ and on
the line connecting them. For $x\in I_{j}$, $\beta\in\mathcal{P}_{0}$
and all $j=1,\dots,l_{k}$, 
\begin{align}
 & \left\Vert \frac{\partial}{\partial x_{j}}\{v(x;\theta,\beta)-v(x;\theta_{0},\beta)\}\right\Vert \nonumber \\
\leq & E_{\dot{\theta}}\left[\left\Vert \nabla_{\beta}\log f_{k}(Y;\dot{\theta})\left\{ \nabla_{\beta}\log f_{k}(Y,x;\beta)^{T}\frac{\partial}{\partial x_{j}}\frac{f_{k}(x| Y;\beta)}{q_{k}(x)}+\frac{\partial}{\partial x_{j}}\nabla_{\beta}\log f_{k}(Y,x;\beta)^{T}\frac{f_{k}(x| Y;\beta)}{q_{k}(x)}\right\} (\theta-\theta_{0})\right\Vert \right].\label{eq:lem2_1}
\end{align}
In the right hand side of the above, $\log f_{k}(Y,x;\beta)=\log f_{k}(Y;\beta)+\log f_{k}(x| Y;\beta)$
and
\begin{align*}
\left\{ \frac{\partial}{\partial x_{j}}\nabla_{\beta}\log f_{k}(Y,x;\beta)\right\} f_{k}(x| Y;\beta) & =\left\{ \frac{\partial}{\partial x_{j}}\nabla_{\beta}\log f_{k}(x| Y;\beta)\right\} f_{k}(x| Y;\beta)\\
 & =\frac{\partial}{\partial x_{j}}\nabla_{\beta}f_{k}(x| Y;\beta)-\nabla_{\beta}\log f_{k}(x| Y;\beta)\frac{\partial}{\partial x_{j}}f_{k}(x\mid Y;\beta).
\end{align*}
Then the right hand side of (\ref{eq:lem2_1}) can be rewritten as
\begin{align*}
 & E_{\dot{\theta}}\left[\left\Vert \nabla_{\beta}\log f_{k}(Y;\dot{\theta})\right.\right.\\
 & \cdot\left\{ \left(\nabla_{\beta}\log f_{k}(Y;\beta)+\nabla_{\beta}\log f_{k}(x\mid Y;\beta)\right)^{T}\left(\frac{\frac{\partial}{\partial x_{j}}f_{k}(x\mid Y;\beta)}{q_{k}(x)}-\frac{f_{k}(x\mid Y;\beta)}{q_{k}(x)}\frac{\partial}{\partial x_{j}}\log q_{k}(x)\right)\right.\\
 & \left.\left.\left.+\frac{\nabla_{\beta}\frac{\partial}{\partial x_{j}}f_{k}(x\mid Y;\beta)}{q_{k}(x)}-\nabla_{\beta}\log f_{k}(x\mid Y;\beta)\frac{\frac{\partial}{\partial x_{j}}f_{k}(x\mid Y;\beta)}{q_{k}(x)}\right\} ^{T}(\theta-\theta_{0})\right\Vert \right]\\
= & E_{\dot{\theta}}\left[\left\Vert \nabla_{\beta}\log f_{k}(Y;\dot{\theta})\left\{ \nabla_{\beta}\log f_{k}(Y;\beta)\left(\frac{\partial}{\partial x_{j}}f_{k}(x\mid Y;\beta)-f_{k}(x\mid Y;\beta)\frac{\partial}{\partial x_{j}}\log q_{k}(x)\right)\right.\right.\right.\\
 & \left.\left.\left.-\nabla_{\beta}f_{k}(x\mid Y;\beta)\frac{\partial}{\partial x_{j}}\log q_{k}(x)+\nabla_{\beta}\frac{\partial}{\partial x_{j}}f_{k}(x\mid Y;\beta)\right\} ^{T}q_{k}(x)^{-1}(\theta-\theta_{0})\right\Vert \right]\\
\leq & E_{\dot{\theta}}\left\{ \left\Vert \nabla_{\beta}\log f_{k}(Y;\dot{\theta})\right\Vert ^{2}\right\} ^{1/2}\left[E_{\dot{\theta}}\left\{ \left\Vert \nabla_{\beta}\log f_{k}(Y;\beta)\right\Vert ^{2}\right\} ^{1/2}\left\{ c_{2}+c_{1}\sup_{x\in I_{j}}\left\Vert \frac{\partial}{\partial x_{j}}\log q_{k}(x)\right\Vert \right\} \right.\\
 & \left.+c_{3}\sup_{x\in I_{j}}\left\Vert \frac{\partial}{\partial x_{j}}\log q_{k}(x)\right\Vert +\sup_{x\in I_{j},y,\beta}\left\Vert \nabla_{\beta}\frac{\partial}{\partial x_{j}}f_{k}(x\mid y;\beta)\right\Vert \right]\sup_{x\in I_{j}}q_{k}(x)^{-1}\sup_{\theta\in\mathcal{P}_{0}}\|\theta-\theta_{0}\|,
\end{align*}
where the inequality above holds by Cauchy-Schwarz inequality and
\ref{cond:fk_upper_bound}. By \ref{cond:differentiability} and \ref{cond:fk_upper_bound}, $\sup_{x\in I_{j},y,\beta}\left\Vert \nabla_{\beta}\frac{\partial}{\partial x_{j}}f_{k}(x\mid y;\beta)\right\Vert <\infty$.
Then by \ref{cond:fk_upper_bound} , \ref{cond:fk_information} and that $\mathcal{P}_{0}$ is compact, for $x\in I_{j}$,
$\beta\in\mathcal{P}_{0}$ and all $j=1,\dots,l_{k}$,
\[
\left\Vert \frac{\partial}{\partial x_{j}}\{v(x;\theta,\beta)-v(x;\theta_{0},\beta)\}\right\Vert \leq\left(c^{(1)}+c^{(2)}\sup_{x\in I_{j}}\left\Vert \frac{\partial}{\partial x_{j}}\log q_{k}(x)\right\Vert \right)\sup_{x\in I_{j}}q_{k}(x)^{-1}\sup_{\theta\in\mathcal{P}_{0}}\|\theta-\theta_{0}\|,
\]
for some constants $c^{(1)}$ and $c^{(2)}$. Similarly to the proof
of Corollary \ref{cor:IS_weight_Donsker}, by \ref{cond:Q_Donsker} the lemma holds.
\end{proof}

\begin{proof}[Proof of Theorem \ref{thm:estimated_score}] Since $Y$ and $X_{1:N}$ are independent, for
simplicity the condition notions will be omitted hereafter, and lower
case alphabets are used to indicate the conditioning when needed.
Since $\widehat{cs}_{1:K}(Y;\widehat{\theta})-cs_{1:K}(Y;\widehat{\theta})$
is the stack of $\widehat{cs}_{k}(Y;\widehat{\theta})-cs_{k}(Y;\widehat{\theta})$,
it is sufficient to decompose the latter. Consider a fixed $N$ and
$k\in\{1,\dots,K\}$. Let 
\[
\widetilde{cs}_{k}(Y;\beta)=\frac{1}{n_{k}}\sum_{j=1}^{n_{k}}\frac{1}{N}\sum_{i=1}^{N}\bigtriangledown_{\beta}\log f_{k}(Y_{kj},x_{kj}^{(i)};\beta)\frac{f_{k}(Y_{kj},x_{kj}^{(i)};\beta)}{f_{k}(Y_{kj};\beta)q_{k}(x_{kj}^{(i)})}.
\]
First we note that the Monte Carlo error of $\widehat{f}_{k}(y;\widehat{\theta})$
in $\widehat{cs}_{k}(Y;\beta)$ is negligible, which implies the difference
between $\widetilde{cs}_{k}(Y;\beta)$ and $\widehat{cs}_{k}(Y;\beta)$
is negligible. This can be seen by verifying conditions of and applying
the proof of \citet[Theorem 19.26]{van2000asymptotic}. By \ref{cond:IS_q} and
\ref{cond:envelop1}, $\sup_{\beta\in B_{\theta_{0}}}\|\psi_{\widetilde{Y}}(x_{k};\beta)\|_{H_{k}}=\sup_{\theta\in\mathcal{P}_{0},\beta\in B_{\theta_{0}}}E_{\theta}\left\{ \psi_{Y_{k1}}(x_{k};\beta)^{2}\right\} $
is integrable with respect to $x_{k}$ under the probability density
$q_{k}(\cdot)$, and with $x_{k}$ and $\beta$ fixed, $\psi_{\widetilde{Y}}(x_{k};\beta)$
has bounded norm for $\widetilde{Y}\in H_{k}$. Then $\sup_{\beta\in B_{\theta_{0}}}\|\psi_{\widetilde{Y}}(x_{k};\beta)\|_{H_{k}}$
is the envelope function of $\{\psi_{\widetilde{Y}}(x_{k};\beta):\beta\in B_{\theta_{0}},\widetilde{Y}\in H_{k}\}$.
By \ref{cond:IS_q} and \ref{cond:envelop1} again and that $\mathcal{P}_{0}$ is compact, $\sup_{\theta\in\mathcal{P}_{0},\beta\in B_{\theta_{0}}}E_{\theta}\left\{ \psi_{Y_{k1}}(x_{k};\beta)-\psi_{Y_{k1}}(x_{k};\theta_{0})\right\} ^{4}$
is integrable with respect to $x_{k}$ under the probability density
$q_{k}(\cdot)$. Then by the dominated convergence theorem, as $\beta\rightarrow\theta_{0}$
we have 
\begin{align*}
\left\Vert E_{q}\left\{ \psi_{\widetilde{Y}}(x_{k};\beta)-\psi_{\widetilde{Y}}(x_{k};\theta_{0})\right\} ^{2}\right\Vert _{H}^{2} & =\sup_{\theta\in\mathcal{P}_{0}}E_{\theta}\left[E_{q}\left\{ \psi_{Y_{k1}}(x_{k};\beta)-\psi_{Y_{k1}}(x_{k};\theta_{0})\right\} ^{2}\right]^{2}\\
 & \leq\sup_{\theta\in\mathcal{P}_{0}}E_{q(\cdot)f_{k}(\cdot;\theta)}\left\{ \psi_{Y_{k1}}(x_{k};\beta)-\psi_{Y_{k1}}(x_{k};\theta_{0})\right\} ^{4}\rightarrow0.
\end{align*}
Then together with Corollary \ref{cor:IS_weight_Donsker}, conditions
are satisfied to apply the proof of \citet[Theorem 19.26]{van2000asymptotic}.
By applying results therein, as $\widehat{\theta}$ converges in probability
and $N\rightarrow\infty$, we have that the empirical process $\sqrt{N}\{N^{-1}\sum_{i=1}^{N}\psi_{\widetilde{Y}}(X_{k}^{(i)};\widehat{\theta})-1\}$
weakly converges uniformly over $\widetilde{Y}\in H_{k}$ where $X_{k}^{(1)},X_{k}^{(2)},\dots$
are i.i.d random variables with density $q_{k}(\cdot)$. Therefore
\[
\frac{\widehat{f}(Y_{kj};\widehat{\theta})}{f(Y_{kj};\widehat{\theta})}=\frac{1}{N}\sum_{i=1}^{N}\psi_{Y_{kj}}(x_{kj}^{(i)};\widehat{\theta})=1+O_{p}\left(\frac{1}{\sqrt{N}}\right),
\]
uniformly over $j=1,\dots,n_{k}$, and we only need to consider $\widetilde{cs}_{k}$. 

Recall that $u(Y_{kj},x_{kj}^{(i)};\beta)$ are terms in the summation
in $\widetilde{cs}_{k}(Y;\beta)$. Define $u(y;\beta):=\triangledown_{\beta}\log f_{k}(y;\beta)$,
then $u(Y_{kj};\beta)$ are terms in the summation in $cs_{k}(Y;\beta)$.
Fisher's identity gives $u(Y_{kj};\beta)=E_{q_{k}}\{u(Y_{kj},X;\beta)\mid Y_{kj}\}$.
Define several terms having zero mean as follows, 
\begin{align*}
 & \overline{e}(Y,x^{(i)};\theta,\beta):=\frac{1}{n_{k}}\sum_{j=1}^{n_{k}}\left\{ u(Y_{kj},x_{kj}^{(i)};\beta)-v(x_{kj}^{(i)};\theta,\beta)\right\} ,\\
 & V(Y,x^{(i)};\theta,\beta):=\overline{e}(Y,x^{(i)};\theta,\beta)-E_{q}\{\overline{e}(Y,X;\theta,\beta)\}\text{ and }\\
 & U(x^{(i)};\theta,\beta):=\frac{1}{n_{k}}\sum_{j=1}^{n_{k}}v(x_{kj}^{(i)};\theta,\beta)-E_{q_{k}}\{v(X_{k1};\theta,\beta)\}.
\end{align*}
Note that 
\begin{align*}
 & E_{q}\{\overline{e}(Y,X;\theta,\beta)\}=\frac{1}{n_{k}}\sum_{j=1}^{n_{k}}u(Y_{kj};\beta)-E_{\theta}\{u(Y_{k1};\beta)\}=cs_{k}(Y;\beta)-E_{\theta}\{u(Y_{k1};\beta)\}\mbox{ and }\\
 & E_{q_{k}}\{v(X_{k1};\theta,\beta)\}=E_{\theta}\{u(Y_{k1};\beta)\}.
\end{align*}
Then we have the following decomposition, separating deviations due
to $Y$ and $X_{1:N}$. 
\begin{align}
 & \{\widetilde{cs}_{k}(Y;\widehat{\theta})-\widetilde{cs}_{k}(Y_{obs};\widehat{\theta})\}-\{cs_{k}(Y;\widehat{\theta})-cs_{k}(Y_{obs};\widehat{\theta})\}\nonumber \\
= & \frac{1}{\sqrt{N}}c_{1}(X^{(1:N)},Y,Y_{obs};\widehat{\theta})+\frac{1}{N}\sum_{i=1}^{N}\left\{ U(X^{(i)};\theta,\widehat{\theta})-U(X^{(i)};\theta_{0},\widehat{\theta})\right\} ,\label{eq:Thm2_decomp}\\
\text{where } & c_{1}(X^{(1:N)},Y,Y_{obs};\beta)=\frac{1}{\sqrt{N}}\sum_{i=1}^{N}\left\{ V(Y,X^{(i)};\theta,\beta)-V(Y_{obs},X^{(i)};\theta_{0},\beta)\right\} .\nonumber 
\end{align}
For the first term in (\ref{eq:Thm2_decomp}), we first claim that
$c_{1}(X^{(1:N)},Y,Y_{obs};\beta)/\sqrt{N}$ is dominated by $cs_{k}(Y;\beta)-cs_{k}(Y_{obs};\beta)$
for large $N$ as $n\rightarrow\infty$ where $\beta\in B_{\theta_{0}}$
is a constant. The claim is proved below. Since $E_{q}\{V(y,X;\theta,\beta)\}=0$
and $E_{\theta}\{V(Y,x;\theta,\beta)\}=0$, by the law of total variance, 

\begin{align}
 & Var\left\{ c_{1}(X^{(1:N)},Y,Y_{obs};\beta)\right\} =E\left[Var_{q}\left\{ c_{1}(X^{(1:N)},Y,Y_{obs};\beta)\right\} \right]\nonumber \\
= & E\left[Var_{q}\left\{ V(Y,X^{(1)};\theta,\beta)-V(Y_{obs},X^{(1)};\theta_{0},\beta)\right\} \right]\nonumber \\
= & Var\left\{ V(Y,X^{(1)};\theta,\beta)-V(Y_{obs},X^{(1)};\theta_{0},\beta)\right\} \nonumber \\
\leq & Var\left\{ \overline{e}(Y,X^{(1)};\theta,\beta)-\overline{e}(Y_{obs},X^{(1)};\theta_{0},\beta)\right\} +Var\left\{ cs_{k}(Y;\beta)-cs_{k}(Y_{obs};\beta)\right\} .\label{eq:Thm2_1}
\end{align}
In the above, since $E_{\theta}\{\overline{e}(Y,X^{(1)};\theta,\beta)\}=0$,
by the law of total variance we have  
\begin{align}
 & Var\left\{ \overline{e}(Y,X^{(1)};\theta,\beta)-\overline{e}(Y_{obs},X^{(1)};\theta_{0},\beta)\right\} \nonumber \\
\leq & E_{q}\left[Var_{\theta}\left\{ \frac{1}{n_{k}}\sum_{j=1}^{n_{k}}u(Y_{kj},X_{kj}^{(1)};\beta)\right\} +Var_{\theta_{0}}\left\{ \frac{1}{n_{k}}\sum_{j=1}^{n_{k}}u(Y_{kj},X_{kj}^{(1)};\beta)\right\} \right].\label{eq:Thm2_2}
\end{align}

By \ref{cond:average_joint_information} and \ref{cond:IS_q}, $a_{n}^{2}Var_{q(\cdot)f_{k}(\cdot;\theta)}\left\{ n_{k}^{-1}\sum_{j=1}^{n_{k}}u(Y_{kj},X_{kj};\beta)\right\} <\infty$,
and by (C3), $a_{n}^{2}Var\{cs_{k}(Y;\beta)\}<\infty$ and away from
$0$, for all $\theta\in\mathcal{P}_{0}$, $\beta\in B_{\theta_{0}}$
and large enough $n$. Therefore $Var\{c_{1}(X^{(1:N)},Y,Y_{obs};\beta)\}=O(a_{n}^{-2})$,
for all $\beta\in B_{\theta_{0}}$. Since $Var\{c_{1}(X^{(1:N)},Y,Y_{obs};\beta)\}$
is continuous in $\beta$, by the continuous mapping theorem, $Var\{c_{1}(X^{(1:N)},Y,Y_{obs};\widehat{\theta})\}=O(a_{n}^{-2})$.
Then since $E\{V(Y,X^{(1)};\theta,\beta)\}=0$, by Markov inequality
we have $c_{1}(X^{(1:N)},Y,Y_{obs};\widehat{\theta})=O_{p}(a_{n}^{-1})$
as $n\rightarrow\infty$. Therefore the claim holds. 

For the second term in (\ref{eq:Thm2_decomp}), 
\begin{align*}
 & U(X^{(i)};\theta,\beta)-U(X^{(i)};\theta_{0},\beta)\\
 & =\frac{1}{n_{k}}\sum_{j=1}^{n_{k}}v(X_{kj}^{(i)};\theta,\beta)-\frac{1}{n_{k}}\sum_{j=1}^{n_{k}}v(X_{kj}^{(i)};\theta_{0},\beta)-E_{q_{k}}\{v(X_{k1};\theta,\beta)\}+E_{q_{k}}\{v(X_{k1};\theta_{0},\beta)\}\\
\Rightarrow & \frac{1}{N}\sum_{i=1}^{N}\{U(X^{(i)};\theta,\beta)-U(X^{(i)};\theta_{0},\beta)\}=\frac{1}{\sqrt{N}}c_{2}(X_{1:N},\theta,\beta)(\theta-\theta_{0}),\\
 & \text{where }c_{2}(X_{1:N},\theta,\beta)=\frac{1}{n_{k}}\sum_{j=1}^{n_{k}}\frac{1}{\sqrt{N}}\sum_{i=1}^{N}(E_{\dot{\theta}}-E_{q_{k},\dot{\theta}})\left\{ \nabla_{\beta}\log f_{k}(Y_{k1};\dot{\theta})u(Y_{k1},X_{kj}^{(i)};\beta)^{T}\right\} ,
\end{align*}
and $\dot{\theta}$ is between $\theta$ and $\theta_{0}$. In the
proof of Lemma \ref{lem:v_class_Donsker}, it has been shown that
\begin{align*}
\left\{ E_{\dot{\theta}}\left\{ \nabla_{\beta}\log f_{k}(Y_{k1};\dot{\theta})u(Y_{k1},x;\beta)^{T}\right\} :\theta\in\mathcal{P}_{0},\beta\in\mathcal{P}_{0}\right\} 
\end{align*}
 is a Q-Donsker class. Similarly to the first part of the proof, by
\ref{cond:envelop2} and \ref{cond:IS_q}, the class has an integrable envelope function, and
by the dominated convergenec theorem, conditions are satisfied to
apply the proof of \citet[Theorem 19.26]{van2000asymptotic}. Then
we have the empirical process $N^{-1/2}\sum_{i=1}^{N}(E_{\dot{\theta}}-E_{q_{k},\dot{\theta}})\left\{ \nabla_{\beta}\log f_{k}(Y_{k1};\dot{\theta})u(Y_{k1},X_{k}^{(i)};\widehat{\theta})^{T}\right\} $
weakly converges uniformly over $\theta\in\mathcal{P}_{0}$ with $X_{k}^{(1)},X_{k}^{(2)},\dots$
i.i.d with density $q_{k}(\cdot)$. Therefore we have $c_{2}(X_{1:N},\theta,\widehat{\theta})=O_{p}(1)$
for all $\theta\in\mathcal{P}_{0}$.
\end{proof}

\section{Proof of Theorem \ref{thm:ABC_heteo_rates}} \label{append:proof_sec4}

\subsection{Notations and set-up}

For completeness the required conditions from \citet{li2018convergence}
are stated below, and denoted by (C4{*}) and (C5{*}) for convenience.
\begin{description}
\item [{(C4{*})\label{kernel_prop}}] 

\begin{itemize}
\item[(i)] $\int\bv K(\bv)\,d\bv=0$. 
\item[(ii)] $\int\prod_{k=1}^{l}v_{i_{k}}K(\bv)\,d\bv<\infty$ for any coordinates
$(v_{i_{1}},\dots,v_{i_{l}})$ of $\bv$ and $l\leq p+6$. 
\item[(iii)] $K(\bv)=K(\|\bv\|_{\Lambda}^{2})$ where $\|\bv\|_{\Lambda}^{2}=\bv^{T}\Lambda\bv$
and $\Lambda$ is a diagonal matrix, and $K(\bv)$ is a decreasing
function of $\|\bv\|_{\Lambda}$. 
\item[(iv)] $K(\bv)=O(e^{-c_{1}\|\bv\|^{\alpha_{1}}})$ for some $\alpha_{1}>0$
and $c_{1}>0$ as $\|\bv\|\rightarrow\infty$. 
\end{itemize}
\item [{(C5{*})\label{sum_approx_tail}}] $\sup_{\theta\in\mathcal{P}_{0}^{c}}f_{W_{n}}(w|\theta)=O(e^{-c_{2}\|w\|^{\alpha_{2}}})$
as $\|w\|\rightarrow\infty$ for some positive constants $c_{2}$
and $\alpha_{2}$, and $A(\theta)$ is bounded in $\mathcal{P}$.
\end{description}
For a constant $d\times p$ matrix $A$, let the minimum and maximum eigenvalues of $A^{T}A$ be $\lambda_{{\rm min}}^{2}(A)$ and $\lambda_{{\rm max}}^{2}(A)$ where $\lambda_{{\rm min}}(A)$ and $\lambda_{{\rm max}}(A)$ are non-negative. Obviously, for any $p$-dimension vector $x$, $\lambda_{{\rm min}}(A)\|x\|\leq\|Ax\|\leq\lambda_{{\rm max}}(A)\|x\|$. For two matrices $A$ and $B$, we say $A$ is bounded by $B$ or $A\leq B$ if $\lambda_{{\rm max}}(A)\leq\lambda_{{\rm min}}(B)$. For a set of matrices $\{A_{i}:i\in I\}$ for some index set $I$, we say it is bounded if $\lambda_{{\rm max}}(A_{i})$ are uniformly bounded in $i$. For a non-negative function $h(x)$, integrable in $\mathbb{R}^{l}$, denote the normalised function $h(x)/\int_{\mathbb{R}^{l}}h(x)\,dx$ by $h(x)^{({\rm norm})}$. For a vector $x$, denote a general polynomial of elements of $x$ with degree up to $l$ by $P_{l}(x)$. 

For any fixed $\delta<\delta_{0}$, let $B_{\delta}$ be the neighborhood $\{\theta:\|\theta-\theta_{0}\|<\delta\}$. Let $\pi_{\delta}(\theta)$ be $\pi(\theta)$ truncated in $B_{\delta}$. Let $t=t(\theta)=a_{n,\varepsilon}(\theta-\theta_{0})$ and $v=v(s)=\varepsilon_{n}^{-1}(s-s_{{\rm obs}})$, the rescaled versions $\theta$ and $s$. For any $A\in\mathscr{B}^{p}$, let $t(A)$ be the set $\{\phi:\phi=t(\theta)\text{ for some }\theta\in A\}$. 

Recall that $\ftil_{n}(s\mid\theta)$ denotes the normal density with
mean $s(\theta)$ and covariance matrix $a_{n,diag}^{-1}A(\theta)a_{n,diag}^{-1}$.
Define $\tPi_{\varepsilon}(\theta\in A\mid s_{{\rm obs}})$ to be
the normal counterpart of $\Pi_{\varepsilon}(\theta\in A\mid s_{{\rm obs}})$
with truncated prior, obtained by replacing $\pi(\theta)$ and $f_{n}(s\mid\theta)$
in $\Pi_{\varepsilon}$ by $\pi_{\delta}(\theta)$ and $\widetilde{f}_{n}(s\mid\theta)$.
So let 
\[
\tpi_{\varepsilon}(\theta,s\mid s_{{\rm obs}})=\frac{\pi_{\delta}(\theta)\ftil_{n}(s\mid\theta)K\{\varepsilon_{n}^{-1}(s-s_{{\rm obs}})\}}{\int_{B_{\delta}}\int_{\mathbb{R}^{d}}\pi_{\delta}(\theta)\ftil_{n}(s\mid\theta)K\{\varepsilon_{n}^{-1}(s-s_{{\rm obs}})\}\,d\theta ds},
\]
$\tpi_{\varepsilon}(\theta\mid s_{{\rm obs}})=\int_{\mathbb{R}^{d}}\tpi_{\varepsilon}(\theta,s\mid s_{{\rm obs}})\,ds$
and $\tPi_{\varepsilon}(\theta\in A\mid s_{{\rm obs}})$ be the distribution
function with density $\tpi_{\varepsilon}(\theta\mid s_{{\rm obs}})$.
Denote the mean of $\tPi_{\varepsilon}$ by $\ttheta_{\varepsilon}$. 

For $A\subset\mathbb{R}^{p}$ and a scalar function $h(\theta,s)$,
let 
\[
\pi_{A}(h)=\int_{A}\int_{\mathbb{R}^{d}}h(\theta,s)\pi(\theta)f_{n}(s\mid\theta)K\{\varepsilon_{n}^{-1}(s-s_{{\rm obs}})\}\varepsilon_{n}^{-d}\,dsd\theta,
\]
and 
\[
\tpi_{A}(h)=\int_{A}\int_{\mathbb{R}^{d}}h(\theta,s)\pi_{\delta}(\theta)\ftil_{n}(s\mid\theta)K\{\varepsilon_{n}^{-1}(s-s_{{\rm obs}})\}\varepsilon_{n}^{-d}\,dsd\theta.
\]
Then $\Pi_{\varepsilon}(\theta\in A\mid s_{{\rm obs}})=\pi_{A}(1)/\pi_{\mathcal{P}}(1)$
and its normal counterpart $\tPi_{\varepsilon}(\theta\in A\mid s_{{\rm obs}})=\tpi_{A}(1)/\tpi_{\mathcal{P}}(1)$.

For two scalar sequences $x_{n}$ and $y_{n}$, we use the following
limit notations: $x_{n}=o(y_{n})$ if $x_{n}/y_{n}\rightarrow0$;
$x_{n}=O(y_{n})$ if $x_{n}/y_{n}\rightarrow c<\infty$; $x_{n}=\Theta(y_{n})$
if $x_{n}/y_{n}\rightarrow c\in(0,\infty)$; and $x_{n}=\omega(y_{n})$
if $x_{n}/y_{n}\rightarrow\infty$. Let $I_{m}$ be the identity matrix
with dimension $m\times m$. Given matrices $X_{1},\dots,X_{m}$,
let $diag(X_{1},\dots,X_{m})$ denote the block diagonal matrix with
diagonal blocks $X_{1},\dots,X_{m}$. Recall that the summary statistics
are $S_{1},\dots,S_{\tau}$, $S_{k}\in\mathbb{R}^{d_{s,k}}$ and
$d=\sum_{k=1}^{\tau}d_{s,k}$. Given an m-dimensional vector $x=(x_{1},\dots,x_{m})$,
let $x_{diag}$ denote $diag(x_{1},\dots,x_{m})$. For example, $a_{n,diag}$
is the diagonal matrix corresponding to $a_{n}:=(a_{1n}\boldsymbol{1}_{d_{s,1}}^{T},\dots,a_{\tau n}\boldsymbol{1}_{d_{s,\tau}}^{T})^{T}$
where $\boldsymbol{1}_{m}$ is the all-one vector with dimension $m$.
Denote $A(\theta_{0})^{-1/2}$ and $A(\theta_{0})^{1/2}$ by $A_{\theta_{0},-1/2}$
and $A_{\theta_{0},1/2}$ for short. Let $W_{{\rm obs}}=a_{n,diag}A_{\theta_{0},-1/2}\{s_{{\rm obs}}-s(\theta_{0})\}.$
By Condition \ref{sum_conv}, $W_{{\rm obs}}\rightarrow Z$ in distribution
as $n\rightarrow\infty$, where $Z\sim N(0,I_{d})$. Let $\widetilde{\pi}_{\varepsilon,tv}(t,v):=a_{n,\varepsilon}^{-d}\pi_{\delta}(\theta_{0}+a_{n,\varepsilon}^{-1}t)\widetilde{f}_{n}(s_{\text{obs}}+\varepsilon_{n}v\mid\theta_{0}+a_{n,\varepsilon}^{-1}t)K(v)$.
For any $A\in\mathscr{B}^{p}$, we rewrite $\tPi_{\varepsilon}$ as,
\begin{equation}
\tPi_{\varepsilon}(\theta\in A\mid s_{{\rm obs}})=\int_{\mathbb{R}^{d}}\int_{t(B_{\delta})}\tPi(\theta\in A\mid s_{{\rm obs}}+\varepsilon_{n}v)\tpi_{\varepsilon,tv}(t,v)^{({\rm norm})}\,dtdv.\label{eq:ABC_posterior_altform1}
\end{equation}

Let $k_{\varepsilon}=\sum_{k=1}^{\tau}k\mathbbm{1}_{\{a_{n,\varepsilon}=a_{kn}\}}+(\tau+1)\mathbbm{1}_{\{\varepsilon_{n}=\omega(a_{\tau n}^{-1})\}}$,
i.e.~the index of summary the convergence rate of which $a_{n,\varepsilon}$
is equal to when $\varepsilon_{n}^{-1}$ is between the convergence
rates of two summaries or the arbitrary $\tau+1$ when $\varepsilon_{n}^{-1}$
is slower than than all $a_{kn}$. We use the superscripts $(i)$,
where $i=1,2,3$, to indicate groups or divisions with elements corresponding
to summary statistics with convergence rates faster, equal to and
slower than $a_{n,\varepsilon}$, and the division is determined by
$k_{\varepsilon}$. For example, let $d^{(1)}=\sum_{k=1}^{k_{\varepsilon}-1}d_{s,k}$
be the dimension of summary statistics with convergence rate faster
than $a_{n,\varepsilon}$, i.e.~being $\omega(a_{n,\varepsilon})$;
$d^{(2)}=d_{s,k_{\varepsilon}}$ be that of those with rates equal;
and $d^{(3)}=d-d^{(1)}-d^{(2)}$ be that of those with rates slower.
Let $a_{n,diag}^{(1)}$, $a_{n,diag}^{(2)}$ and $a_{n,diag}^{(3)}$
be the diagonal submatrices of $a_{n,diag}$ with corresponding blocks,
i.e. $a_{n,diag}^{(1)}=diag(a_{1n}I_{d_{s,1}},\dots,a_{k_{\varepsilon}-1,n}I_{d_{s,k_{\varepsilon}-1}}),$
$a_{n,diag}^{(2)}=a_{k_{\varepsilon}n}$ and $a_{n,diag}^{(3)}=diag(a_{k_{\varepsilon}+1,n}I_{d_{s,k_{\varepsilon}+1}},\dots,a_{\tau,n}I_{d_{s,\tau}})$.
Let the partition of columns of $A_{\theta_{0},-1/2}$ be $(A_{\theta_{0},-1/2}^{(1)},A_{\theta_{0},-1/2}^{(2)},A_{\theta_{0},-1/2}^{(3)})$
the column indices of the submatrices are $1:d^{(1)}$, $(d^{(1)}+1):(d^{(1)}+d^{(2)})$
and $(d^{(1)}+d^{(2)}+1):d$ respectively. 

In order to discuss various cases for convergence results depending
on the rate of $\varepsilon_{n}\rightarrow0$, convergence rates below
will be used. Given the definition of $a_{n,\varepsilon}$ and $k_{\varepsilon}$,
as $n\rightarrow\infty$ we have 
\begin{align*}
 & a_{n,\varepsilon}^{-1}a_{k_{\varepsilon}n}\begin{cases}
=1, & \text{when }k_{\varepsilon}\leq\tau,\\
\rightarrow\infty, & \text{when }k_{\varepsilon}=\tau+1,
\end{cases}\ a_{n,\varepsilon}^{-1}a_{kn}\rightarrow\begin{cases}
\infty, & \text{when }k\leq k_{\varepsilon},\\
1 & \text{when }k=k_{\varepsilon},\\
0, & \text{when }k>k_{\varepsilon},
\end{cases}\\
 & a_{n,\varepsilon}\varepsilon_{n}\begin{cases}
\rightarrow c\in[0,\infty), & \text{when }k_{\varepsilon}\leq\tau,\\
=1, & \text{when }k_{\varepsilon}=\tau+1,
\end{cases}\ a_{n,\varepsilon}^{-1}a_{n,diag}\rightarrow diag(\infty I_{d^{(1)}},I_{d^{(2)}},0I_{d^{(3)}}),\\
 & \text{and }\varepsilon_{n}a_{n,diag}\rightarrow diag(\infty I_{d^{(1)}},cI_{d^{(2)}},0I_{d^{(3)}}).
\end{align*}

\subsection{Proof of Theorem \ref{thm:ABC_heteo_rates}} \label{proof_thm:ABC_heteo_rates}

The proof follows the structure of the proof of \citet[Proposition 1]{li2018convergence}.
All lemmas in \citet{li2018convergence} needed for its Proposition
1 are extended to the scenario of heteogenous convergence rates. Define
\[
g_{n}(t,v):=N\{Ds(\theta_{0})t;a_{n,\varepsilon}\varepsilon_{n}v+a_{n,diag}^{-1}a_{n,\varepsilon}A_{\theta_{0},1/2}W_{\text{obs}},a_{n,diag}^{-1}a_{n,\varepsilon}^{2}A(\theta_{0})a_{n,diag}^{-1}\}K(v).
\]
It is shown in Lemma \ref{lem:ABC_dens_expans} that $g_{n}(t,v)$
is the dominating term of $\tpi_{\varepsilon}(\theta,s\mid s_{{\rm obs}})$
under the transformation $t(\theta)$ and $v(s)$ up to a scale. Throughout
the proof the data are considered to be random. For two square matrices $A$ and $B$, $A>B$ denotes that $A-B$ is positive definite.

\begin{proof}[Proof of Theorem \ref{thm:ABC_heteo_rates}]

Lemma \ref{lem:ABC_CDF_normal} shows that $\Pi_{\varepsilon}\{a_{n,\varepsilon}(\theta-\theta_{0})\in A\mid s_{\text{obs}}\}$
and $\tPi_{\varepsilon}\{a_{n,\varepsilon}(\theta-\theta_{0})\in A\mid s_{\text{obs}}\}$
have the same limit. Therefore it is sufficient to only consider the
convergence of $\tPi_{\varepsilon}$ of the scaled and centered $\theta$.
For $\varepsilon_{n}=o(a_{\tau n}^{-1/2})$, according to (\ref{eq:ABC_posterior_altform1}),
we have 
\[
\tPi_{\varepsilon}\{a_{k_{\varepsilon}n}(\theta-\theta_{0})\in A\mid s_{\text{obs}}\}=\int_{\mathbb{R}^{d}}\int_{t(B_{\delta})}\tPi_{\varepsilon}\{a_{k_{\varepsilon}n}(\theta-\theta_{0})\in A\mid s_{\text{obs}}+\varepsilon_{n}v\}\tpi_{\varepsilon,tv}(t',v)^{({\rm norm})}\,dt'dv,
\]
where $t'\in\mathbb{R}^{p}$ is the integration variable to be distinguished
from $t$ when $t$ is used as the integration variable. The leading
term of the above can be obtained by Lemma \ref{lem:missp_CDF_expans}
and Lemma \ref{lem:ABC_dens_expans}, and we have 
\begin{align}
 & \sup_{A\in\mathscr{B}^{p}}\left|\tPi_{\varepsilon}\{a_{k_{\varepsilon}n}(\theta-\theta_{0})\in A\mid s_{{\rm obs}}\}\right.\nonumber \\
 & \left.-\int_{\mathbb{R}^{d}}\int_{t(B_{\delta})}\int_{A}|a_{k_{\varepsilon}n}^{-1}a_{1n}|^{p}N\{a_{k_{\varepsilon}n}^{-1}a_{1n}t;Q_{\theta_{0}}(W_{{\rm obs}}''+a_{n,diag}\varepsilon_{n}v),V_{\theta_{0}}^{-1}\}g_{n}(t',v)^{(norm)}\,dtdt'dv\right|\nonumber \\
 & =o_{p}(1),\label{eq:Thm1_leading1}
\end{align}
where $W_{{\rm obs}}''$, $Q_{\theta_{0}}$ and $V_{\theta_{0}}$
are defined in Lemma \ref{lem:missp_CDF_expans}. Let $t''(t,t')=a_{k_{\varepsilon}n}^{-1}a_{1n}(t-t'),$
$t''(A,B_{\delta})$ be the set $\{t''(t,t'):t\in A,t'\in t(B_{\delta})\}$
and $v''(v,t)=W_{\text{obs}}''+a_{n,diag}\varepsilon_{n}v-a_{k_{\varepsilon}n}^{-1}a_{1n}Q_{\theta_{0}}'t$,
where $Q_{\theta_{0}}'=\left(Ds_{1}(\theta_{0})^{T},\{A_{2:\tau,1}(\theta_{0})A_{11}^{-1}(\theta_{0})Ds_{1}(\theta_{0})\}^{T}\right)^{T}$.
For the numerator of the integral in (\ref{eq:Thm1_leading1}), since
$Q_{\theta_{0}}Q_{\theta_{0}}'=I_{p}$, with the transformation $t''=t''(t,t')$
it can be rewritten as the following, 
\begin{align*}
 & \int_{\mathbb{R}^{d}}\int_{t(B_{\delta})}\int_{A}|a_{k_{\varepsilon}n}^{-1}a_{1n}|^{p}N\{a_{k_{\varepsilon}n}^{-1}a_{1n}t;Q_{\theta_{0}}(W_{{\rm obs}}''+a_{n,diag}\varepsilon_{n}v),V_{\theta_{0}}^{-1}\}g_{n}(t',v)\,dtdt'dv\\
= & \int_{\mathbb{R}^{d}}\int_{t(B_{\delta})}\int_{A}|a_{k_{\varepsilon}n}^{-1}a_{1n}|^{p}N\{a_{k_{\varepsilon}n}^{-1}a_{1n}(t-t');Q_{\theta_{0}}v''(v,t'),V_{\theta_{0}}^{-1}\}g_{n}(t',v)\,dtdt'dv\\
= & \int_{\mathbb{R}^{d}}\int_{t''(A,B_{\delta})}\int_{A}N\{t'';Q_{\theta_{0}}v''(v,t'),V_{\theta_{0}}^{-1}\}g_{n}(t-a_{k_{\varepsilon}n}a_{1n}^{-1}t'',v)\,dtdt''dv.
\end{align*}
With techniques in the proof of \citet[Proposition 1]{li2018convergence},
in the above integral by letting $a_{k_{\varepsilon}n}a_{1n}^{-1}\rightarrow0$
and then integrating over $t''$, its leading term can be obtained
and is stated below, 
\begin{align*}
 & \sup_{A\in\mathscr{B}^{p}}\left|\int_{\mathbb{R}^{d}}\int_{t''(A,B_{\delta})}\int_{A}N\{t'';Q_{\theta_{0}}v''(v,t'),V_{\theta_{0}}^{-1}\}g_{n}(t-a_{k_{\varepsilon}n}a_{1n}^{-1}t'',v)\,dtdt''dv-\int_{A}\int_{\mathbb{R}^{d}}g_{n}(t,v)\,dvdt\right|\\
 & =o_{p}(1).
\end{align*}
For the denominator of the integral in (\ref{eq:Thm1_leading1}),
its leading term can be obtained similarly. Then we have 
\[
\sup_{A\in\mathscr{B}^{p}}\mid\tPi_{\varepsilon}\{a_{k_{\varepsilon}n}(\theta-\theta_{0})\in A\mid s_{\text{obs}}\}-\int_{A}\psi_{\varepsilon}(t)\,dt\mid=o_{p}(1),
\]
where $\psi_{\varepsilon}(t):=\int_{\mathbb{R}^{d}}g_{n}(t,v)\,dv/\int_{\mathbb{R}^{d}\times\mathbb{R}^{p}}g_{n}(t,v)\,dvdt$.
Therefore \textbf{(a)} holds.

For \textbf{(b)}, it holds by Lemma \ref{lem:normal_expect_expans}.

For \textbf{(c)}, consider the bias of $\tilde{\theta}_{\varepsilon}$
which is the correspondence of $\theta_{\varepsilon}$ using the normal
approximation of $f(s\mid\theta)$. Since 
\begin{align}
\tilde{\theta}_{\varepsilon}-\theta_{0} & =a_{n,\varepsilon}^{-1}\int t\frac{g_{n}(t,v)}{\int g_{n}(t,v)\,dtdv}\,dtdv\text{ and }\nonumber \\
g_{n}(t,v) & =N\{\underbrace{Ds(\theta_{0})}_{A}t;\underbrace{a_{n,\varepsilon}\varepsilon_{n}}_{b_{n}}v+\underbrace{a_{n,diag}^{-1}a_{n,\varepsilon}A_{\theta_{0},1/2}}_{K_{n}}\underbrace{W_{\text{obs}}}_{c},\underbrace{a_{n,diag}^{-1}a_{n,\varepsilon}^{2}A(\theta_{0})a_{n,diag}^{-1}}_{K_{n}K_{n}^{T}}\}K(v),\label{eq:mean_expression}
\end{align}
we will use the following decomposition to write $g_{n}(t,v)$ into
two valid normal densities in $p$-dimension and $d-p$-dimension
respectively. Consider a matrix $A\in\mathbb{R}^{d\times p}$, a sequence
$b_{n}\rightarrow[0,\infty)$ and an invertible $K_{n}\in\mathbb{R}^{d\times d}$.
For a matrix $X$, let $Q_{X}=X^{T}X$, $\beta_{X}=Q_{X}^{-1}X^{T}$
and $P_{X}=X\beta_{X}$. We have the following decomposition, 
\begin{align*}
 & N(At;b_{n}v+K_{n}c,K_{n}K_{n}^{T})K(v)=\frac{1}{(2\pi)^{d/2}|K_{n}|}N(K_{n}^{-1}At;K_{n}^{-1}b_{n}v+c,I_{d})K(v)\\
 & =N\{t;\beta_{K_{n}^{-1}A}(K_{n}^{-1}b_{n}v+c),Q_{K_{n}^{-1}A}^{-1}\}r(v;A,b_{n},K_{n},c),\\
\text{where } & r(v;A,b_{n},K_{n},c)=\frac{1}{(2\pi)^{(d-p)/2}|K_{n}|}\exp\{-\frac{1}{2}(K_{n}^{-1}b_{n}v+c)^{T}(I_{d}-P_{K_{n}^{-1}A})(K_{n}^{-1}b_{n}v+c)\}K(v).
\end{align*}
Then 
\begin{align*}
 & \int t\frac{N(At;b_{n}v+K_{n}c,K_{n}K_{n}^{T})K(v)}{\int N(At;b_{n}v+K_{n}c,K_{n}K_{n}^{T})K(v)\,dtdv}\,dtdv=\beta_{K_{n}^{-1}A}c+R(A,b_{n},K_{n},c),\\
\text{ where } & R(A,b_{n},K_{n},c)=\beta_{K_{n}^{-1}A}K_{n}^{-1}b_{n}\int v\frac{r(v;A,b_{n},K_{n},c)}{\int r(v;A,b_{n},K_{n},c)\,dv}\,dv.
\end{align*}
Using the decomposition above, we have 
\begin{align}
\tilde{\theta}_{\varepsilon}-\theta_{0} & =a_{n,\varepsilon}^{-1}\left\{ \beta_{K_{n}^{-1}A}c+R(A,b_{n},K_{n},c)\right\} ,\label{eq:mean_decomp}
\end{align}
with expressions of $A$, $b_{n}$, $K_{n}$ and $c$ given in (\ref{eq:mean_expression}). 

Using terms in $g_{n}(t,v)$, write $A_{\theta_{0},-1/2}=(A_{1;\theta_{0},-1/2},A_{2:\tau;\theta_{0},-1/2})$
where $A_{1;\theta_{0},-1/2}\in\mathbb{R}^{d\times d_{1}}$, and we
have 
\begin{align*}
\beta_{K_{n}^{-1}A}= & a_{n,\varepsilon}\{Ds(\theta_{0})^{T}a_{n,diag}A(\theta_{0})^{-1}a_{n,diag}Ds(\theta_{0})\}^{-1}Ds(\theta_{0})^{T}a_{n,diag}A_{\theta_{0},-1/2}^{T}\\
= & a_{n,\varepsilon}a_{1n}^{-1}\left\{ (Ds_{1}(\theta_{0})^{T},Ds_{2:\tau}(\theta_{0})^{T}a_{2:\tau,n,diag}a_{1n}^{-1})A(\theta_{0})^{-1}\left(\begin{array}{c}
Ds_{1}(\theta_{0})\\
a_{2:\tau,n,diag}a_{1n}^{-1}Ds_{2:\tau}(\theta_{0})
\end{array}\right)\right\} ^{-1}\\
 & \cdot(Ds_{1}(\theta_{0})^{T},Ds_{2:\tau}(\theta_{0})^{T}a_{2:\tau,n,diag}a_{1n}^{-1})\left(\begin{array}{c}
A_{1;\theta_{0},-1/2}^{T}\\
A_{2:\tau;\theta_{0},-1/2}^{T}
\end{array}\right)\\
\approx & a_{n,\varepsilon}a_{1n}^{-1}\beta_{A_{1;\theta_{0},-1/2}Ds_{1}(\theta_{0})},
\end{align*}
and similarly as $\beta_{K_{n}^{-1}A}$, 
\begin{align*}
P_{K_{n}^{-1}A}= & A_{\theta_{0},-1/2}a_{n,diag}Ds(\theta_{0})\{Ds(\theta_{0})^{T}a_{n,diag}A(\theta_{0})^{-1}a_{n,diag}Ds(\theta_{0})\}^{-1}Ds(\theta_{0})^{T}a_{n,diag}A_{\theta_{0},-1/2}^{T}\\
\approx & P_{A_{1;\theta_{0},-1/2}Ds_{1}(\theta_{0})}.
\end{align*}
Then in $R(A,b_{n},K_{n},c)$, 
\begin{align}
r(v;A,b_{n},K_{n},c) & \propto\exp\left\{ -\frac{1}{2}\|(A_{\theta_{0},-1/2}a_{n,diag}\varepsilon_{n}v+W_{obs})^{T}(I_{d}-P_{K_{n}^{-1}A})^{1/2}\|\right\} K(v)\nonumber \\
 & \approx\exp\left\{ -\frac{1}{2}\|(A_{\theta_{0},-1/2}a_{n,diag}\varepsilon_{n}v+W_{obs})^{T}(I_{d}-P_{A_{1;\theta_{0},-1/2}Ds_{1}(\theta_{0})})^{1/2}\|\right\} K(v).\label{eq:r_leading1}
\end{align}
Since $P_{A_{1;\theta_{0},-1/2}Ds_{1}(\theta_{0})}$ is the projection
matrix in $d$-dimension onto a $p$-dimension subspace, let 
\begin{align*}
I_{d}-P_{A_{1;\theta_{0},-1/2}Ds_{1}(\theta_{0})}=U\left(\begin{array}{cc}
0 & 0\\
0 & I_{d-p}
\end{array}\right)U^{T} & \text{ and }P_{A_{1;\theta_{0},-1/2}Ds_{1}(\theta_{0})}=U\left(\begin{array}{cc}
I_{p} & 0\\
0 & 0
\end{array}\right)U^{T},
\end{align*}
where $U:=(u_{1:p},u_{(p+1):d})$ and $u_{i}$ are orthonormal eigenvectors
in $d$-dimension space. Then the right hand side of (\ref{eq:r_leading1})
is 
\begin{align}
 & \exp\left\{ -\frac{1}{2}\|(A_{\theta_{0},-1/2}a_{n,diag}\varepsilon_{n}v+W_{obs})^{T}u_{(p+1):d}\|^{2}\right\} K(v)\nonumber \\
= & \exp\left\{ -\frac{1}{2}\|\sum_{k=1}^{3}v^{(k)^{T}}a_{n,diag}^{(k)}\varepsilon_{n}u_{(p+1):d}^{*(k)}+W_{obs}^{T}u_{(p+1):d}\|^{2}\right\} K(v),\label{eq:r_leading2}
\end{align}
where $U^{*}=(u_{1:p}^{*},u_{(p+1):d}^{*}):=A(\theta_{0})^{-1/2}U$,
$u_{(p+1):d}^{*T}=(u_{(p+1):d}^{*(1)^{T}},u_{(p+1):d}^{*(2)^{T}},u_{(p+1):d}^{*(3)^{T}})$,
$u_{(p+1):d}^{*(k)}\in\mathbb{R}^{d^{(k)}\times(d-p)}$ and $v^{T}=(v^{(1)^{T}},v^{(2)^{T}},v^{(3)^{T}})$.
The norm in (\ref{eq:r_leading2}) is for $(d-p)$-dimension vector.
We would like to write it as the norm of a $d$-dimension vector,
so that it is a density of $v$. 

Consider $\sum_{k=1}^{3}p_{k}>r$ where $r$, $p_{1}$, $p_{2}$ and
$p_{3}$ are positive integers, matrices $C_{k}\in\mathbb{R}^{r\times p_{k}}$,
vectors $v_{k}\in\mathbb{R}^{p_{k}\times1}$ and $W\in\mathbb{R}^{r\times1}$,
we have 
\begin{align*}
\|\sum_{k=1}^{3}C_{k}v_{k}+W\|^{2} & =\sum_{i,j=1}^{3}v_{i}^{T}C_{i}^{T}C_{j}v_{j}+2\left(\sum_{k=1}^{3}C_{k}v_{k}\right)^{T}W+W^{T}W\\
 & =\left\{ \left(\begin{array}{c}
v_{1}\\
v_{2}\\
v_{3}
\end{array}\right)+\left(\begin{array}{c}
A_{1}W\\
A_{2}W\\
A_{3}W
\end{array}\right)\right\} ^{T}\left(C_{i}^{T}C_{j}\right)_{i,j=1,2,3}\left\{ \left(\begin{array}{c}
v_{1}\\
v_{2}\\
v_{3}
\end{array}\right)+\left(\begin{array}{c}
A_{1}W\\
A_{2}W\\
A_{3}W
\end{array}\right)\right\} +H(W),
\end{align*}
where $(C_{1},C_{2},C_{3})(A_{1}^{T},A_{2}^{T},A_{3}^{T})^{T}=I_{d}$
and $H(W)$ does not depend on $v_{1}$, $v_{2}$ and $v_{3}$. 

Then since $u_{(p+1):d}^{*T}A(\theta_{0})^{-1/2}u_{(p+1):d}=I_{d-p}$,
(\ref{eq:r_leading2}) can be written as 
\begin{align}
 & \exp\left\{ -\frac{1}{2}\|\left(u_{(p+1):d}^{*(i)^{T}}u_{(p+1):d}^{*(j)}\right)_{i,j=1,2,3}^{1/2}\left(a_{n,diag}\varepsilon_{n}v+A(\theta_{0})^{1/2}u_{(p+1):d}u_{(p+1):d}^{T}W_{obs}\right)\|^{2}\right\} K(v).\label{eq:r_leading3}
\end{align}
Since (\ref{eq:r_leading3}) is in the form of $\exp\{-\|C(K_{n}v+c)\|^{2}/2\}K(v)$
where $K_{n}=diag(K_{n}^{(1)},K_{n}^{(2)},K_{n}^{(3)})$ and $K_{n}^{(i)}$
are diagonal matrices with diagonal elements going to $\infty$, $c\in[0,\infty)$
and $0$ respectively for $i=1,2,3$, we can use the following upper
bound similar to that used in the proof of Lemma \ref{lem:normal_expect_bounded}
for its case $\textcircled{\small2}$: since $\|K_{n}v+c\|\geq\|K_{n}^{(1)}v^{(1)}+c\|^{2}$
and $\|v\|^{2}\geq\|v^{(3)}\|^{2}+\|v^{(2)}\|^{2}$, 
\begin{align*}
 & \int\|K_{n}v\|^{2}|K_{n}^{(1)}|\exp\{-\frac{1}{2}\|C(K_{n}v+c)\|^{2}\}K(v)\,dv\\
\leq & \int\sum_{k=1}^{3}\|K_{n}^{(k)}v^{(k)}\|^{2}|K_{n}^{(1)}|\exp\{-\frac{1}{2}\|K_{n}^{(1)}v^{(1)}+c\|^{2}\}K^{*}(v^{(2)},v^{(3)})\,dv=O_{p}(1),\\
\text{ and} & \int|K_{n}^{(1)}|\exp\{-\frac{1}{2}\|C(K_{n}v+c)\|^{2}\}K(v)\,dv=\Theta_{p}(1),
\end{align*}
where $K^{*}(v^{(2)},v^{(3)})=K(\|v^{(2)}\|^{2}+\|v^{(3)}\|^{2})$.
Hence we have 
\begin{align*}
\int v\frac{r(v;A,b_{n},K_{n},c)}{\int r(v;A,b_{n},K_{n},c)\,dv}\,dv & =O_{p}(a_{n,diag}^{-1}\varepsilon_{n}^{-1}),\\
\text{ and }R(A,b_{n},K_{n},c) & =a_{n,\varepsilon}a_{1n}^{-1}\beta_{A_{1;\theta_{0},-1/2}Ds_{1}(\theta_{0})}(a_{n,diag}^{-1}a_{n,\varepsilon}A_{\theta_{0},1/2})^{-1}a_{n,\varepsilon}\varepsilon_{n}O_{p}(a_{n,diag}^{-1}\varepsilon_{n}^{-1})\\
 & =O_{p}(a_{n,\varepsilon}a_{1n}^{-1}).
\end{align*}
Therefore by (\ref{eq:mean_decomp}), $a_{1n}(\widetilde{\theta}_{\varepsilon}-\theta_{0})=O_{p}(1)$
and by Lemma \ref{lem:ABC_probs_tails}, $\theta_{\varepsilon}-\theta_{0}=O_{p}(a_{1n}^{-1})+O_{p}(\alpha_{n}^{-1}a_{n,\varepsilon}^{-1})$.
\end{proof}

\subsection{Lemmas}
\begin{lem}
\label{lem:missp_CDF_expans}When $\varepsilon_{n}=O(a_{\tau n}^{-1})$,
for any fixed $v\in\mathbb{R}^{d}$ and small enough $\delta$, as
$n\rightarrow\infty$,
\begin{align*}
 & \sup_{A\in\mathscr{B}^{p}}\left|\tPi\{a_{n,\varepsilon}(\theta-\theta_{0})\in A\mid s_{{\rm obs}}+\varepsilon_{n}v\}-\int_{A}|a_{n,\varepsilon}^{-1}a_{1n}|^{p}N\{a_{n,\varepsilon}^{-1}a_{1n}t;Q_{\theta_{0}}(W_{{\rm obs}}''+a_{n,diag}\varepsilon_{n}v),V_{\theta_{0}}^{-1}\}\,dt\right|=o_{p}(1),
\end{align*}
where 
\begin{align*}
 & A(\theta)\text{ is divided as }\left(\begin{array}{cc}
A_{11} & A_{1,2:\tau}\\
A_{2:\tau,1} & A_{2:\tau,2:\tau}
\end{array}\right)(\theta),\ A_{11}(\theta)\in\mathbb{R}^{d_{s,1}\times d_{s,1}},W_{\text{obs}}'':=\left(\begin{array}{c}
A_{11}(\theta_{0})^{1/2}W_{\text{obs},1}\\
A_{2:\tau|1}(\theta_{0})^{1/2}W_{\text{obs},2:\tau}
\end{array}\right),\\
 & Q_{\theta_{0}}:=V_{\theta_{0}}^{-1}\left(Ds_{1}(\theta_{0})^{T}A_{11}(\theta_{0})^{-1},\{A_{2:\tau,1}(\theta_{0})A_{11}^{-1}(\theta_{0})Ds_{1}(\theta_{0})\}^{T}A_{2:\tau|1}(\theta_{0})^{-1}\right),\\
 & V_{\theta_{0}}:=Ds_{1}(\theta_{0})^{T}A_{11}(\theta_{0})^{-1}Ds_{1}(\theta_{0})+\{A_{2:\tau,1}(\theta_{0})A_{11}^{-1}(\theta_{0})Ds_{1}(\theta_{0})\}^{T}A_{2:\tau|1}(\theta_{0})^{-1}\{A_{2:\tau,1}(\theta_{0})A_{11}^{-1}(\theta_{0})Ds_{1}(\theta_{0})\},\\
 & A_{2:\tau|1}(\theta_{0}):=A_{2:\tau,2:\tau}(\theta)-A_{2:\tau,1}(\theta)A_{11}(\theta)^{-1}A_{1,2:\tau}(\theta),
\end{align*}
and $W_{\text{obs}}=(W_{\text{obs},1}^{T},W_{\text{obs},2:\tau}^{T})^{T}$. 
\end{lem}
\medskip{}

\begin{lem}
\label{lem:ABC_dens_expans}Assume Condition \ref{cond:prior}, \ref{sum_conv} and \ref{kernel_prop}. If $\varepsilon_{n}=o(a_{\tau n}^{-1/2})$,
then $\int_{\mathbb{R}^{d}}\int_{t(B_{\delta})}|\tpi_{\varepsilon,tv}(t,v)-\pi(\theta_{0})g_{n}(t,v)|\,dtdv\rightarrow0$
in probability and $\int_{\mathbb{R}^{d}}\int_{t(B_{\delta})}c_{n,\varepsilon}g_{n}(t,v)\,dtdv=\Theta_{p}(1)$,
as $n\rightarrow\infty$, where $c_{n,\varepsilon}$ is defined in
Lemma \ref{lem:normal_expect_expans}. 
\end{lem}
\medskip{}

\begin{lem}
\label{lem:ABC_CDF_normal}Assume Condition \ref{cond:prior}, \ref{sum_conv}--\ref{sum_approx_tail}.
Then as $n\rightarrow\infty$,
\[
\sup_{A\in\mathscr{B}^{p}}\left|\Pi_{\varepsilon}\{a_{n,\varepsilon}(\theta-\theta_{0})\in A\mid s_{{\rm obs}}\}-\tPi_{\varepsilon}\{a_{n,\varepsilon}(\theta-\theta_{0})\in A\mid s_{{\rm obs}}\}\right|=o_{p}(1).
\]
\end{lem}
\medskip{}

\begin{lem}
\label{lem:ABC_CDF_tails}Assume Conditions \ref{cond:prior}, \ref{sum_conv}--\ref{kernel_prop}.
If $\varepsilon_{n}=o(a_{\tau n}^{-1/2})$, then\begin{enumerate}[(i)]

\item for any $\delta<\delta_{0}$, $\Pi_{\varepsilon}(\theta\in B_{\delta}^{c}\mid s_{{\rm obs}})$ and $\tPi_{\varepsilon}(\theta\in B_{\delta}^{c}\mid s_{{\rm obs}})$ are $o_{p}(1)$;

\item there exists a $\delta<\delta_0$ such that $\sup_{A\in\mathscr{B}^{p}}\left|\Pi_{\varepsilon}(\theta\in A\cap B_{\delta}\mid s_{{\rm obs}})-\tPi_{\varepsilon}(\theta\in A\cap B_{\delta}\mid s_{{\rm obs}})\right|=o_{p}(1)$.

\end{enumerate}
\end{lem}
\medskip{}

\begin{lem}
\label{lem:ABC_probs_tails}Assume Conditions \ref{cond:prior}, \ref{sum_conv}--\ref{kernel_prop}.
Then as $n\rightarrow\infty$,

\begin{enumerate}[(i)]

\item if Condition \ref{sum_approx_tail} also holds then, for any $\delta<\delta_{0}$, $\pi_{B_{\delta}^{c}}(1)$ and $\tpi_{B_{\delta}^{c}}(1)$ are $o_{p}(1)$, and $O_{p}(e^{-a_{n,\varepsilon}^{\alpha_{\delta}}c_{\delta}})$ for some positive constants $c_{\delta}$ and $\alpha_{\delta}$ depending on $\delta$;

\item $\pi_{B_{\delta}}(1)=\tpi_{B_{\delta}}(1)\{1+O_{p}(\alpha_{n}^{-1})\}$ and $\mbox{{\rm sup}}_{A\subset B_{\delta}}\left|\pi_{A}(1)-\tpi_{A}(1)\right|/\tpi_{B_{\delta}}(1)=O_{p}(\alpha_{n}^{-1})$;

\item if $\varepsilon_{n}=o(a_{\tau n}^{-1/2})$, $\tpi_{B_{\delta}}(1)$ and $\pi_{B_{\delta}}(1)$ are $\Theta_{p}(a_{n,\varepsilon}^{d-p})$, and thus $\tpi_{\mathcal{P}}(1)$ and $\pi_{\mathcal{P}}(1)$ are $\Theta_{p}(a_{n,\varepsilon}^{d-p})$

\item if $\varepsilon_{n}=o(a_{\tau n}^{-1/2})$ and Condition \ref{sum_approx_tail} holds, $\theta_{\varepsilon}=\ttheta_{\varepsilon}+O_{p}(\alpha_{n}^{-1}a_{n,\varepsilon}^{-1})$. 

\end{enumerate}
\end{lem}
\medskip{}

\begin{lem}
\label{lem:ABC_expect_tails}Assume Conditions \ref{sum_conv}--\ref{sum_approx_tail}.
Then for any $\delta<\delta_0,$ $\pi_{B_{\delta}^{c}}(h)=O_{p}(e^{-a_{n,\varepsilon}^{\alpha_{\delta}}c_{\delta}})$ for some positive constants $c_{\delta}$ and $\alpha_{\delta}$ depending on $\delta$.
\end{lem}
\medskip{}

\begin{lem}
\label{lem:ABC_expect_normal}Assume conditions of Lemma \ref{lem:normal_expect_expans}
and Conditions \ref{sum_approx} and \ref{kernel_prop}. Then if $\varepsilon_{n}=o(a_{\tau n}^{-1/2})$,
there exists a $\delta<\delta_{0}$ such that \begin{equation} \frac{\pi_{B_{\delta}}(h)}{\pi_{B_{\delta}}(1)}=h(\theta_{0})+a_{n,\varepsilon}^{-1}Dh(\theta_{0})^{T}\left\{\frac{\tpi_{B_{\delta}}(t)}{\tpi_{B_{\delta}}(1)}+O_{p}(\alpha_{n}^{-1})\right\}+ \frac{1}{2}a_{n,\varepsilon}^{-2}\left[\frac{\tpi_{B_{\delta}}\{t^{T}Hh(\theta_{t})t\}}{\tpi_{B_{\delta}}(1)}+O_{p}(\alpha_{n}^{-1})\right]. \end{equation}
\end{lem}
\medskip{}

\begin{lem}
\label{lem:normal_expect_bounded}Assume Condition \ref{kernel_prop}. For $t\in\mathbb{R}^{p}$ and $v\in\mathbb{R}^{d}$, let $\{A_{n}(t)\}$ be a series of $d\times p$ matrix functions, $\{C_{n}(t)\}$ be a series of $d\times d$ matrix functions, $Q$ be a positive definite matrix and $g_{1}(v)$ and $g_{2}(v)$ be probability densities in $\mathbb{R}^{d}$.Let
$c_{n}$ be a random vector bounded in probability. Let $d$, $\{d^{(i)}\}_{i=1}^{3}$,
$\{d_{j}^{(i)}\}_{j=1}^{\tau_{i}}$ be non-negative integers satisfying
$d=d^{(1)}+d^{(2)}+d^{(3)}$ and $d^{(i)}=\sum_{j=1}^{\tau_{i}}d_{j}^{(i)}$.
For $j=1,\dots,\tau_{1}$, let $\{k_{nj}^{(1)}\}_{n=1,\dots,\infty}$
be a series going to $\infty$ for each $j$, and these series satisfy
that $k_{nj}^{(1)}\rightarrow\infty$ and $k_{nj}^{(1)}/k_{n,j+1}^{(1)}\rightarrow\infty$,
i.e.~the rate of convergence decreases with $j$. For each $i=2,3$
and $j=1,\dots,\tau^{(i)}$, let $\{k_{nj}^{(i)}\}_{n=1,\dots,\infty}$
be a series converging to a finite positive constant. Let $K_{n}^{(i)}=diag(k_{n1}^{(i)}I_{d_{1}^{(i)}},\dots,k_{n\tau_{i}}^{(i)}I_{d_{\tau_{i}}^{(i)}})$
for $i=1,2,3$ and $K_{n}=diag(K_{n}^{(1)},K_{n}^{(2)},K_{n}^{(3)})$.
Let $b_{n}$ be a non-negative series satisfying $k_{nj}^{(1)}b_{n} \rightarrow\infty$ for any $j=1,\dots,\tau_{1}$, and $\lim_{n\rightarrow\infty}b_{n}\in[0,\infty)$.
Let $R$ be a random matrix and $\{B_{n}\}_{n=1,\dots,\infty}$ be
a series of random matrices satisfying one of the following cases:

\begin{table}[H]
\begin{tabular}{|c|c|c|c|c|c|c|}
\hline 
 & $d^{(1)}$ & $d^{(2)}$ & $d^{(3)}$ & $\lim b_{n}$ & $R\in\mathbb{R}^{d^{(3)}\times d^{(2)}}$ & $\text{lim}B_{n}\in\mathbb{R}^{d^{(3)}\times d^{(3)}}$\tabularnewline
\hline 
\hline 
$\textcircled{\small1}$ & $0$ & $>0$ & $>0$ & $[0,\infty)$ & UL bounded or $0$ & UL bounded or $0$\tabularnewline
\hline 
$\textcircled{\small2}$ & $>0$ & $>0$ & $\geq0$ & $[0,\infty)$ & UL bounded & $0$\tabularnewline
\hline 
$\textcircled{\small3}$ & $>0$ & $=0$ & $=0$ & $(0,\infty)$ & N/A & N/A\tabularnewline
\hline 
\end{tabular}
\end{table}
\noindent where UL bounded stands for upper bounded and lower bounded away from
zeo. 

\noindent Let $Q_{n}=diag\left(I_{d^{(1)}},\left(\begin{array}{cc}
I_{d^{(2)}} & 0\\
R & B_{n}
\end{array}\right)\right)$. According to the partitions $d=d^{(1)}+d^{(2)}+d^{(3)}$ and $d^{(1)}=\sum_{j=1}^{\tau_{i}}d_{j}^{(1)}$,
denote rows of $A_{n}(t)$ by $A_{n}(t)=(A_{n}^{(1)}(t)^{T},A_{n}^{(2)}(t)^{T},A_{n}^{(3)}(t)^{T})^{T}$
and $A_{n}^{(1)}(t)=(A_{n}^{(11)}(t)^{T},\dots,A_{n}^{(1\tau_{1})}(t)^{T})^{T}$,
where $A_{n}^{(1j)}(t)\in\mathbb{R}^{d_{j}^{(1)}\times p}$, and denote
dimension of $v$ and $c_{n}$ similarly, respectively. If 

\begin{enumerate}[(i)]

\item $g_{1}(v)$ and $g_{2}(v)$ are bounded in $\mathbb{R}^{d}$;

\item $g_{1}(v)$ and $g_{2}(v)$ depend on $v$ only through $\|v\|$ and are decreasing functions of $\|v\|$;

\item there exists an integer $l$ such that $\int\prod_{k=1}^{l+p}v_{i_{k}}g_{j}(v)\,dv<\infty$, $j=1,2$, for any coordinates $(v_{i_{1}},\dots,v_{i_{l}})$ of $v$;

\item there exists a positive constant $m$ such that for any $t\in\mathbb{R}^{p}$ and $n$, $\lambda_{{\rm min}}\{A_{n}(t)\}$ and $\lambda_{{\rm min}}\{C_{n}(t)\}$ are greater than $m$;

\item the column space of $A_n^{(1i)}$ is $\mathbb{R}^p$, 

\end{enumerate}\noindent then for any $P_{l}(t,v)\in\mathbb{P}_{l,(t,v)}$, 
\begin{align}
 & \int_{\mathbb{R}^{p}}\int_{\mathbb{R}^{d}}P_{l}(t,v)g_{1}[C_{n}(t)K_{n}\{A_{n}(t)t-b_{n}Q_{n}v-K_{n}^{-1}c_{n}\}]g_{2}(Qv)\,dvdt=O_{p}(\varphi_{n}),\label{eq:normal_expect1}\\
 & \int_{\mathbb{R}^{p}}\int_{\mathbb{R}^{d}}g_{1}[C_{n}(t)K_{n}\{A_{n}(t)t-b_{n}Q_{n}v-K_{n}^{-1}c_{n}\}]g_{2}(Qv)\,dvdt=\Theta_{p}(\varphi_{n}).\label{eq:normal_expect2}
\end{align}
where $\varphi_{n}=\mathbbm{1}_{\{\text{case }\textcircled{\small1}\}}+b_{n}^{p-d^{(1)}}|K_{n}^{(1)}|^{-1}\mathbbm{1}_{\{\text{case }\textcircled{\small2}\}}+|K_{n}^{(1)}|^{-1}\mathbbm{1}_{\{\text{case }\textcircled{\small3}\}}$.
\end{lem}
\medskip{}

\begin{lem}
\label{lem:normal_expect_expans}Assume Conditions \ref{cond:prior},
\ref{sum_conv} and \ref{kernel_prop} hold. If $\varepsilon_{n}=o(a_{\tau n}^{-1/2})$,
there exists a $\delta<\delta_{0}$ such that 

\begin{enumerate}[(i)]

\item $\widetilde{\pi}_{B_{\delta}}(1)=a_{n,\varepsilon}^{d-p}\pi(\theta_{0})\int_{t(B_{\delta})}\int_{\mathbb{R}^{d}}g_{n}(t,v)\,dvdt\Big\{1+O_{p}(a_{n,\varepsilon}^{-1})\Big\}$;

\item  $\int_{t(B_{\delta})}\int_{\mathbb{R}^{d}}g_{n}(t,v)\,dtdv=\Theta_{p}(c_{n,\varepsilon}),$
where $c_{n,\varepsilon}=a_{n,\varepsilon}^{-d^{(3)}}|a_{n,diag}^{(3)}|\left[(a_{n,\varepsilon}\varepsilon_{n})^{p-d^{(1)}}\mathbbm{1}_{\{1<k_{\varepsilon}\leq\tau\}}+\mathbbm{1}_{\{k_{\varepsilon}=1\text{ or }k_{\varepsilon}>\tau\}}\right]$;

\item  $\lim_{n\rightarrow\infty}\int_{\|t\|>M}\int_{\mathbb{R}^{d}}g_{n}(t,v)\,dtdv/\int_{\mathbb{R}^{p}}\int_{\mathbb{R}^{d}}g_{n}(t,v)\,dtdv=o_{p}(1)$,
as $M\rightarrow\infty$;

\item  $\widetilde{\pi}_{B_{\delta}}(t)/\widetilde{\pi}_{B_{\delta}}(1)=\int_{t(B_{\delta})}\int_{\mathbb{R}^{d}}tg_{n}(t,v)\,dtdv/\int_{t(B_{\delta})}\int_{\mathbb{R}^{d}}g_{n}(t,v)\,dtdv+O_{p}(a_{n,\varepsilon}^{-1});$

\item  $\ensuremath{\widetilde{\pi}_{B_{\delta}}\{P_{2}(t)\}/\widetilde{\pi}_{B_{\delta}}(1)=O_{p}(1)}\text{ for any }P_{2}(t)\in\mathbb{P}_{2,t}.$

\end{enumerate}
\end{lem}


%


\section{Proofs of lemmas in Supplement \ref{append:proof_sec4}}

Transformations below will be used. In the case of $\varepsilon_{n}=O(a_{\tau n}^{-1})$,
there are situations where some terms involving $s^{(3)}$ are dominated
by those involving $s^{(2)}$, and it is difficult to consider their
joint densities. We will consider the following transformation of
$s$ which linearly combines $s^{(2)}$ and $s^{(3)}$. Since $d^{(2)}+d^{(3)}>0$
in this case, by Condition \ref{sum_conv} and QR decomposition \citep{golub2013matrix},
we can decompose $A_{\theta_{0},-1/2}^{(3:2)}=(A_{\theta_{0},-1/2}^{(3)},A_{\theta_{0},-1/2}^{(2)})$
into $Q_{3:2}\left(\begin{array}{cc}
R_{33} & R_{32}\\
0 & R_{22}
\end{array}\right)$ where $Q_{3:2}$ is a $d\times(d^{(2)}+d^{(3)})$ matrix with orthogonal
columns and $R_{22}$ and $R_{33}$ are upper triangular matrices
with non-zero diagonal elements. Then let $s^{(3)'}=s^{(3)'}(s^{(2)},s^{(3)}):=a_{n,diag}^{(3)}a_{k_{\varepsilon}n}^{-1}s^{(3)}+R_{33}^{-1}R_{32}s^{(2)}$
which implies the transformation $v^{(3)'}=v^{(3)'}(v^{(2)},v^{(3)}):=a_{n,diag}^{(3)}a_{k_{\varepsilon}n}^{-1}v^{(3)}+R_{33}^{-1}R_{32}v^{(2)}$,
and $s'=s'(s):=(s^{(1:2)^{T}},s^{(3)'^{T}})^{T}$ which implies the
transformation $v'=v'(v):=(v^{(1:2)^{T}},v^{(3)'^{T}})^{T}$. 

Regarding the transformations $s'$ and $v'$, we have the following
facts. Under $s'$, the convergence rate of $s$ is changed from $a_{n,diag}$
to $a'_{n,diag}:=diag(a_{n,diag}^{(1)},a_{k_{\varepsilon}n}I_{d^{(2)}+d^{(3)}})$
which is the convergence rate of $s'$. The relationships between
the pair $s$ and $s'$ and the pair $v$ and $v'$ can be written
as 
\begin{align*}
 & v'=Q_{v1}v\text{ and }s'=Q_{v1}s,\\
\text{where } & Q_{v1}:=\left(\begin{array}{ccc}
I_{d^{(1)}} & 0 & 0\\
0 & I_{d^{(2)}} & 0\\
0 & R_{33}^{-1}R_{32} & a_{n,diag}^{(3)}a_{k_{\varepsilon}n}^{-1}
\end{array}\right)\mathbbm\{1\}_{\{d^{(3)}>0\}}+I_{d^{(1)}+d^{(2)}}\mathbbm\{1\}_{\{d^{(3)}=0\}},\\
 & a_{n,diag}v=Q_{v2}^{-1}a'_{n,diag}v'\text{ and }a_{n,diag}s=Q_{v2}^{-1}a'_{n,diag}s',\\
\text{where } & Q_{v2}:=\left(\begin{array}{ccc}
I_{d^{(1)}} & 0 & 0\\
0 & I_{d^{(2)}} & 0\\
0 & R_{33}^{-1}R_{32} & I_{d^{(3)}}
\end{array}\right),
\end{align*}
using the facts that $Q_{v1}=Q_{v2}diag(I_{d^{(1)}+d^{(2)}},a_{n,diag}^{(3)}a_{k_{\varepsilon}n}^{-1})$
and $Q_{v1}a_{n,diag}^{-1}=a{}_{n,diag}^{'-1}Q_{v2}$. 

\begin{proof}[Proof of Lemma \ref{lem:missp_CDF_expans}]

This extends \citet[Lemma A1]{li2018convergence}. Since $\varepsilon_{n}=O(a_{\tau n}^{-1})$,
$a_{n,\varepsilon}=a_{k_{\varepsilon}n}$. Recall that $a_{1n}$ is
the fastest convergence rate among all the rates of the summary statistics.
Let $v''=v''(v)=((a_{n,diag}^{(1)}a_{kn}^{-1}v^{(1)})^{T},v^{(2)^{T}},v^{(3)^{T}})^{T}$,
then $\varepsilon_{n}v=diag(a_{kn}\varepsilon_{n}a_{n,diag}^{(1)^{-1}},\varepsilon_{n}I_{d^{(2)}+d^{(3)}})v''=a_{n,diag}'a_{k_{\varepsilon}n}\varepsilon_{n}v''$
and 
\begin{equation}
\tPi\{a_{n,\varepsilon}(\theta-\theta_{0})\in A\mid s_{{\rm obs}}+\varepsilon_{n}v\}=\tPi\{a_{1n}(\theta-\theta_{0})\in a_{1n}a_{k_{\varepsilon}n}^{-1}A\mid s_{{\rm obs}}+a_{n,diag}'^{-1}a_{k_{\varepsilon}n}\varepsilon_{n}v''\}.\label{eq:Lem1_term1}
\end{equation}
So first we study $\tPi\{a_{1n}(\theta-\theta_{0})\in A\mid s_{{\rm obs}}+B_{n}v\}$
for a fixed $v$, where $B_{n}$ is a diagonal matrix and the elements
of $a_{n,diag}B_{n}$ converge to $0$ or $1$. Following the technique
in the proof of \citet[Lemma A1]{li2018convergence}, we verify the two conditions
of \citet{kleijn2012bernstein} that, if there exist $\Delta_{n,\theta_{0}}$
and a positive definite matrix $V_{\theta_{0}}$ such that, 

\begin{description} 
\item [{(KV1)}] for any compact set $K\subset t(B_{\delta})$,  
\begin{align*} 
\sup_{t\in K}\left|\log\frac{\widetilde{f}_{n}(s_{{\rm obs}}+B_{n}v\mid\theta_{0}+a_{1n}^{-1}t)}{\widetilde{f}_{n}(s_{{\rm obs}}+B_nv\mid\theta_{0})}-t^{T}V_{\theta_{0}}\Delta_{n,\theta_{0}}+\frac{1}{2}t^{T}V_{\theta_{0}}t\right| & \rightarrow0, 
\end{align*}  
in probability as $n\rightarrow\infty$, and 

\item [{(KV2)}] $E\{\tPi(a_{1n}\|\theta-\theta_{0}\|>M_{n}\mid s_{{\rm obs}}+\varepsilon_{n}v)\}\rightarrow0$  as $n\rightarrow\infty$ for any sequence of constants $M_{n}\rightarrow\infty$, 
\end{description} then 
\begin{equation}
\sup_{A\in\mathscr{B}^{p}}\left|\tPi\{a_{1n}(\theta-\theta_{0})\in A\mid s_{{\rm obs}}+B_{n}v\}-\int_{A}N(t;\Delta_{n,\theta_{0}},V_{\theta_{0}}^{-1})\,dt\right|\rightarrow0,\label{eq:Lem1_term2}
\end{equation}
 in probability as $n\rightarrow\infty$.

For (KV1), by the definition of $\widetilde{f}_{n}(s|\theta)$,
\[
\log\frac{\widetilde{f}_{n}(s_{\text{obs}}+B_{n}v|\theta_{0}+a_{1n}^{-1}t)}{\widetilde{f}_{n}(s_{\text{obs}}+B_{n}v|\theta_{0})}=\log\frac{N\{s_{\text{obs}}+B_{n}v;s(\theta_{0}+a_{1n}^{-1}t),a_{n,diag}^{-2}A(\theta_{0}+a_{1n}^{-1}t)\}}{N\{s_{\text{obs}}+B_{n}v;s(\theta_{0}),a_{n,diag}^{-2}A(\theta_{0})\}}.
\]

Applying Taylor expansion on $s(\theta_{0}+xt)$ and $A(\theta_{0}+xt)$
around $x=0$ by a trivial extension of the corresponding part in the
proof of \citet[Lemma A1]{li2018convergence}, we have the following expansion,
\begin{align}
 & \log\frac{N\{s_{\text{obs}}+B_{n}v;s(\theta_{0}+a_{1n}^{-1}t),a_{n,diag}^{-2}A(\theta_{0}+a_{1n}^{-1}t)\}}{N\{s_{\text{obs}}+B_{n}v;s(\theta_{0}),a_{n,diag}^{-2}A(\theta_{0})\}}\nonumber \\
= & t^{T}Ds(\theta_{0})^{T}a_{1n}^{-1}a_{n,diag}A(\theta_{0})^{-1}\{A(\theta_{0})^{1/2}W_{\text{obs}}+a_{n,diag}B_{n}v\}\nonumber \\
 & -\frac{1}{2}t^{T}Ds(\theta_{0})^{T}a_{1n}^{-1}a_{n,diag}A(\theta_{0})^{-1}a_{1n}^{-1}a_{n,diag}Ds(\theta_{0})t\nonumber \\
 & -\frac{a_{1n}^{-1}}{2}\|\{\sum_{i=1}^{p}D_{\theta_{i}}A^{-1}(\theta_{0}+e_{n}^{(2)}t)t_{i}\}^{1/2}[A(\theta_{0})^{1/2}W_{\text{obs}}+a_{n,diag}\{B_{n}v-a_{1n}^{-1}Ds(\theta_{0}+e_{n}^{(1)}t)t\}]\|^{2}\nonumber \\
 & +\frac{a_{1n}^{-1}}{2}\{D\log|A(\theta_{0}+e_{n}^{(3)}t)|\}^{T}t,\label{eq:Lem1_eqn1}
\end{align}
where for $j=1,2,3$, $e_{n}^{(j)}$ is a function of $t$ satisfying
$|e_{n}^{(j)}|\le a_{n}^{-1}$ which is from the remainder of the
Taylor expansions. However the expansion (\ref{eq:Lem1_eqn1}) can
not be used to verify (KV1) because as $n\rightarrow\infty$, part
of the diagonal elements of $a_{1n}^{-1}a_{n,diag}$ goes to zero,
causing the would-be $V_{\theta_{0}}$ not be positive definite. 

Instead of using the expansion (\ref{eq:Lem1_eqn1}), we seperate
those involving the part of $s_{\text{obs}}$ with convergence rate
$a_{1n}$ using conditional densities. First let $B_{n}=diag(B_{n,1},B_{n,2:\tau})$,
$a_{n,diag}=diag(a_{1n}I_{d_{1}},a_{2:\tau,n,diag})$, and for any
related vector, use subscripts $1$ and $2:\tau$ to denote the division
of it into the two dimensions $d_{s,1}$ and $d-d_{s,1}$, e.g.~$v=(v_{1}^{T},v_{2:\tau}^{T})^{T}$
with $dim(v_{1})=d_{s,1}$. Then decompose the normal densities in
both the numerator and denominator above by conditioning on the first
$d_{s,1}$ dimensions and we have 
\begin{align*}
 & N\{s_{\text{obs}}+B_{n}v;s(\theta),a_{n,diag}^{-2}A(\theta)\}\\
= & N\left\{ \left(\begin{array}{c}
s_{\text{obs},1}+B_{n,1}v_{1}\\
s_{\text{obs},2:\tau}+B_{n,2:\tau}v_{2:\tau}
\end{array}\right);\left(\begin{array}{c}
s_{1}(\theta)\\
s_{2:\tau}(\theta)
\end{array}\right)\right.\\
 & \left.,\left(\begin{array}{cc}
a_{1n}^{-1}I_{d_{1}} & 0\\
0 & a_{2:\tau,n,diag}^{-1}
\end{array}\right)\left(\begin{array}{cc}
A_{11} & A_{1,2:\tau}\\
A_{2:\tau,1} & A_{2:\tau,2:\tau}
\end{array}\right)(\theta_{0}+a_{1n}^{-1}t)\left(\begin{array}{cc}
a_{1n}^{-1}I_{d_{1}} & 0\\
0 & a_{2:\tau,n,diag}^{-1}
\end{array}\right)\right\} \\
= & \widetilde{f}_{n,2:\tau}(s_{\text{obs},2:\tau}+B_{n,2:\tau}v_{2:\tau}\mid s_{\text{obs},1}+B_{n,1}v_{1},\theta)\widetilde{f}_{n,1}(s_{\text{obs},1}+B_{n,1}v_{1}\mid\theta),
\end{align*}
where 
\begin{align*}
 & \widetilde{f}_{n,2:\tau}(s_{2:\tau}\mid s_{1},\theta)=N\left\{ s_{2:\tau};\mu_{2:\tau|1}(\theta),a_{2:\tau,n,diag}^{-1}A_{2:\tau|1}(\theta)a_{2:\tau,n,diag}^{-1}\right\} \\
\text{and } & \widetilde{f}_{n,1}(s_{1}\mid\theta)=N\left\{ s_{1};s_{1}(\theta),a_{1n}^{-2}A_{11}(\theta)\right\} ,
\end{align*}
where $\mu_{2:\tau|1}(\theta,s_{1})=s_{2:\tau}(\theta)+a_{2:\tau,n,diag}^{-1}A_{2:\tau,1}(\theta)A_{11}^{-1}(\theta)a_{1n}\{s_{1}-s_{1}(\theta)\}$. 

The expansion of the log-ratio between the two conditional probabilities
can be obtained by following the expansion (\ref{eq:Lem1_eqn1}),
and we have 
\begin{align*}
 & \log\frac{\widetilde{f}_{n,2:\tau}(s_{\text{obs},2:\tau}+B_{n,2:\tau}v_{2:\tau}\mid s_{\text{obs},1}+B_{n,1}v_{1},\theta_{0}+a_{1n}^{-1}t)}{\widetilde{f}_{n,2:\tau}(s_{\text{obs},2:\tau}+B_{n,2:\tau}v_{2:\tau}\mid s_{\text{obs},1}+B_{n,1}v_{1},\theta_{0})}\\
= & t^{T}\{A_{2:\tau,1}(\theta_{0})A_{11}^{-1}(\theta_{0})Ds_{1}(\theta_{0})\}^{T}A_{2:\tau|1}(\theta_{0})^{-1}\{A_{2:\tau|1}(\theta_{0})^{1/2}W_{\text{obs},2:\tau}+a_{2:\tau,n,diag}B_{n,2:\tau}v_{2:\tau}\}\\
 & -\frac{1}{2}t^{T}\{A_{2:\tau,1}(\theta_{0})A_{11}^{-1}(\theta_{0})Ds_{1}(\theta_{0})\}^{T}A_{2:\tau|1}(\theta_{0})^{-1}\{A_{2:\tau,1}(\theta_{0})A_{11}^{-1}(\theta_{0})Ds_{1}(\theta_{0})\}t\\
 & -\frac{a_{1n}^{-1}}{2}\left\Vert \left\{ \sum_{i=1}^{p}D_{\theta_{i}}A_{2:\tau|1}^{-1}(\theta_{0}+e_{n}^{(2)}t)t_{i}\right\} ^{1/2}\right.\\
 & \left.\cdot\left.\left[A_{2:\tau|1}(\theta)^{1/2}W_{\text{obs},2:\tau}+a_{2:\tau,n,diag}B_{n,2:\tau}v_{2:\tau}-\left\{ A_{2:\tau,1}(\theta)A_{11}^{-1}(\theta)Ds_{1}(\theta)\right\} t\right]\right|_{\theta=\theta_{0}+e_{n}^{(1)}t}\right\Vert ^{2}\\
 & +\frac{a_{1n}^{-1}}{2}\left\{ D\log\left|A_{2:\tau|1}(\theta_{0}+e_{n}^{(3)}t)\right|\right\} ^{T}t.
\end{align*}
Specifically, the expansion above is otained by that in the right
hand side of (\ref{eq:Lem1_eqn1}), replacing $s_{\text{obs}}+B_{n}v$
with $s_{\text{obs},2:\tau}+B_{n,2:\tau}v_{2:\tau}$, $a_{n,diag}^{-1}$
with $a_{2:\tau,n,diag}^{-1}$, $A(\theta)$ with $A_{2:\tau|1}(\theta)$
and $s(\theta)$ with $\mu_{2:\tau|1}(\theta,s_{\text{obs},1}+B_{n,1}v_{1})$.
Furthermore, replace $A(\theta_{0})^{1/2}W_{\text{obs}}$ with $A_{2:\tau|1}(\theta)^{1/2}W_{\text{obs},2:\tau}$,
and $t^{T}Ds(\theta_{0})^{T}a_{1n}^{-1}a_{n,diag}$ with $t^{T}D_{\theta}\mu_{2:\tau|1}(\theta_{0},s_{\text{obs},1}+B_{n,1}v_{1})a_{1n}^{-1}a_{2:\tau,n,diag}$
which is dominated by

\noindent$t^{T}\{A_{2:\tau,1}(\theta)A_{11}^{-1}(\theta)Ds_{1}(\theta_{0})\}^{T}$
since $s_{\text{obs},1}+B_{n,1}v_{1}-s_{1}(\theta_{0})\rightarrow0$
in probability and $a_{1n}^{-1}a_{2:\tau,n,diag}\rightarrow\infty$.
Similarly, we have 
\begin{align*}
 & \log\frac{\widetilde{f}_{n,1}(s_{\text{obs},1}+B_{n,1}v_{1}\mid\theta_{0}+a_{1n}^{-1}t)}{\widetilde{f}_{n,1}(s_{\text{obs},1}+B_{n,1}v_{1}\mid\theta_{0})}\\
= & t^{T}Ds_{1}(\theta_{0})^{T}A_{11}(\theta_{0})^{-1}\{A_{11}(\theta_{0})^{1/2}W_{\text{obs}}+a_{1n}B_{n,1}v_{1}\}-\frac{1}{2}t^{T}Ds_{1}(\theta_{0})^{T}A_{11}(\theta_{0})^{-1}Ds_{1}(\theta_{0})t\\
 & -\frac{a_{1n}^{-1}}{2}\|\{\sum_{i=1}^{p}D_{\theta_{i}}A_{11}^{-1}(\theta_{0}+e_{n}^{(2)}t)t_{i}\}^{1/2}[A_{11}(\theta_{0})^{1/2}W_{\text{obs},1}+a_{1n}B_{n,1}v_{1}-Ds_{1}(\theta_{0}+e_{n}^{(1)}t)t\}]\|^{2}\\
 & +\frac{a_{1n}^{-1}}{2}\{D\log|A_{11}(\theta_{0}+e_{n}^{(3)}t)|\}^{T}t,
\end{align*}
by replacing $t^{T}Ds(\theta_{0})^{T}a_{1n}^{-1}a_{n,diag}$ in the
right hand side of (\ref{eq:Lem1_eqn1}) with $t^{T}Ds(\theta_{0})^{T}$. 

Then combing the two expansions above, for a fixed $v$ and a compact
set $K$, we have 
\begin{align*}
\sup_{t\in K}\left|\log\frac{\widetilde{f}_{n}(s_{\text{obs}}+B_{n}v|\theta_{0}+a_{1n}^{-1}t)}{\widetilde{f}_{n}(s_{\text{obs}}+B_{n}v|\theta_{0})}-t^{T}V_{\theta_{0}}\Delta_{n,\theta_{0}}+\frac{1}{2}t^{T}V_{\theta_{0}}t\right| & \rightarrow0,
\end{align*}
in probability as $n\rightarrow\infty$, where 
\begin{align*}
 & V_{\theta_{0}}:=\{A_{2:\tau,1}(\theta_{0})A_{11}^{-1}(\theta_{0})Ds_{1}(\theta_{0})\}^{T}A_{2:\tau|1}(\theta_{0})^{-1}\{A_{2:\tau,1}(\theta_{0})A_{11}^{-1}(\theta_{0})Ds_{1}(\theta_{0})\}+Ds_{1}(\theta_{0})^{T}A_{11}(\theta_{0})^{-1}Ds_{1}(\theta_{0}),\\
\text{and } & \Delta_{n,\theta_{0}}:=V_{\theta_{0}}^{-1}\left(Ds_{1}(\theta_{0})^{T}A_{11}(\theta_{0})^{-1},\{A_{2:\tau,1}(\theta_{0})A_{11}^{-1}(\theta_{0})Ds_{1}(\theta_{0})\}^{T}A_{2:\tau|1}(\theta_{0})^{-1}\right)(W_{\text{obs}}''+a_{n,diag}B_{n}v).
\end{align*}
Therefore (KV1) holds. 

For (KV2), lines in the proof of \citet[Lemma A1]{li2018convergence} can be followed,
and the only part that needs extension is to show that for any constant
$M>0$ and $\|\bar{v}\|\le M$, 
\begin{equation}
\int_{\|t\|>M_{n}}\tilde{\pi}(t,\bar{v}\mid v)\:dt=o(1)\text{ whe }M_{n}\rightarrow\infty,\label{eq:Lem1_eqn2}
\end{equation}
where $\bar{v}=a_{n,diag}\{s-s(\theta_{0})\}$ and $\tilde{\pi}(t,\bar{v}\mid v)=\pi_{\delta}(\theta_{0}+a_{1n}^{-1}t)\|a_{n,diag}\|^{-1}\tilde{f}_{n}\{s(\theta_{0})+a_{n,diag}^{-1}\bar{v}+B_{n}v\mid\theta_{0}+a_{1n}^{-1}t\}$.
First when $\|\bar{v}\|<M$ and $M_{n}$ is large enough, since the
elements of $a_{n,diag}B_{n}$ converge to $0$ or $1$ and $v$ is
fixed, 
\[
\{t:\|t\|>M_{n}\}\subset\left\{ t:\inf_{\theta\in B_{\delta}}\|a_{1n}^{-1}a_{n,diag}Ds(\theta)t\|/2\geq\|\bar{v}+a_{n,diag}B_{n}v\|\right\} .
\]
Choose $\delta$ small enough such that $Ds(\theta)$ and $A(\theta)^{1/2}$
are bounded and away from $0$ which is possible by Condition \ref{sum_conv}(iii).
Then when $\|t\|\geq M_{n}$, by Taylor expansion, 
\begin{align*}
 & \|a_{n,diag}\|^{-1}\tilde{f}_{n}\{s(\theta_{0})+a_{n,diag}^{-1}\bar{v}+B_{n}v\mid\theta_{0}+a_{1n}^{-1}t\}\\
= & \|a_{n,diag}\|^{-1}N\{s(\theta_{0})+a_{n,diag}^{-1}\bar{v}+B_{n}v;s(\theta_{0}+a_{1n}^{-1}t),a_{n,diag}^{-1}A(\theta_{0}+a_{1n}^{-1}t)a_{n,diag}^{-1}\}\\
= & |A(\theta_{0}+a_{1n}^{-1}t)|^{1/2}\exp[-\|A(\theta_{0}+a_{1n}^{-1}t)^{-1/2}\{a_{1n}^{-1}a_{n,diag}Ds(\theta_{0}+e_{n}^{(1)}t)t-\bar{v}-a_{n,diag}B_{n}v\}\|^{2}/2]\\
\leq & N\{\|a_{1n}^{-1}a_{n,diag}Ds(\theta_{0}+e_{n}^{(1)}t)t\|^{2}/2;0,A(\theta_{0}+a_{1n}^{-1}t)\}\\
\leq & cN(\lambda_{\text{max}}^{-1}\lambda_{\text{min}}\|t\|/2;0,1),
\end{align*}
where $\lambda_{\text{min}}$ and $\lambda_{\text{max}}$ are the
bounds of $Ds(\theta)$ and $A(\theta)^{1/2}$ for $\theta\in B_{\delta}$
and $c$ is some constant. Then by following the lines in the proof
of \citet[Lemma A1]{li2018convergence}, (\ref{eq:Lem1_eqn2}) holds and we have (KV2)
holds. Therefore (\ref{eq:Lem1_term2}) holds. 

Finally the lemma holds by applying (\ref{eq:Lem1_term2}) on (\ref{eq:Lem1_term1})
and noting that in $\Delta_{n,\theta_{0}}$, $a_{n,diag}a_{n,diag}'^{-1}a_{k_{\varepsilon}n}\varepsilon_{n}v=\varepsilon_{n}a_{n,diag}v$.

\end{proof}

\begin{proof}[Proof of Lemma \ref{lem:ABC_dens_expans}]

This extends \citet[Lemma A2]{li2018convergence}. By the definition of $\tpi_{\varepsilon,tv}(t,v)$,
we have

\noindent $\tpi_{B_{\delta}}(1)=a_{n,\varepsilon}^{d-p}\int_{\mathbb{R}^{d}}\int_{t(B_{\delta})}\tpi_{\varepsilon,tv}(t,v)\,dtdv$.
Then by Lemma \ref{lem:normal_expect_expans}, the lemma holds. 

\end{proof}

\begin{proof}[Proof of Lemma \ref{lem:ABC_CDF_normal}]

This extends \citet[Lemma 5 in supplement]{li2018convergence}. The lines of the proof therein
can be followed exactly. 

\end{proof}

\begin{proof}[Proof of Lemma \ref{lem:ABC_CDF_tails}]

The lemma extends \citet[Lemma A3]{li2018convergence}. It holds immediately by Lemma
\ref{lem:ABC_probs_tails}.

\end{proof}

\begin{proof}[Proof of Lemma \ref{lem:ABC_probs_tails}]

The lemma extends \citet[Lemma 2 in supplement]{li2018convergence}. The lines of the proof therein
can be followed exactly. 

\end{proof}

\begin{proof}[Proof of Lemma \ref{lem:ABC_expect_tails}]

The lemma extends \citet[Lemma 3]{li2018asymptotic}. The lines of the proof therein
can be followed exactly except replacing $a_{n}$ therein by $a_{\tau n}$. 

\end{proof}

\begin{proof}[Proof of Lemma \ref{lem:ABC_expect_normal}]

The lemma extends \citet[Lemma 6]{li2018asymptotic}. The lines of the proof therein
can be followed exactly with some trivial extensions. 

\end{proof}

\begin{proof}[Proof of Lemma \ref{lem:normal_expect_bounded}]

This lemma extends \citet[Lemma 7]{li2018asymptotic}. For simplicity, here $\int$
denotes the integration over the whole Euclidean space. According
to (ii), $g_{1}(v)$ can be written as $\bar{g}_{1}(\|v\|)$. Since
we just want to obtain orders of the integrals, without loss of generality,
for all $i$, assume $k_{ni}^{(2)}=k_{ni}^{(3)}=1$. For completeness,
with quantities defined in the lemma, we state the results of \citet[Lemma 7]{li2018asymptotic} in the following, 
\begin{align}
 & \int_{\mathbb{R}^{p}}\int_{\mathbb{R}^{d}}P_{l}(t,v)k_{n}^{d}g_{1}[k_{n}C_{n}(t)\{A_{n}(t)t-b_{n}v-k_{n}^{-1}c_{n}\}]g_{2}(Qv)\,dvdt=O_{p}(1),\label{eq:old_normal_expect1}\\
 & \int_{\mathbb{R}^{p}}\int_{\mathbb{R}^{d}}k_{n}^{d}g_{1}[k_{n}C_{n}(t)\{A_{n}(t)t-b_{n}v-k_{n}^{-1}c_{n}\}]g_{2}(Qv)\,dvdt=\Theta_{p}(1),\label{eq:old_normal_expect2}
\end{align}
where $(k_{\infty},b_{\infty}):=\lim_{n\rightarrow\infty}(k_{n},b_{n})$
satisfies one of the two cases: A. $(k_{\infty},b_{\infty})=(\infty,1);$
B. $(k_{\infty},b_{\infty})=(c_{k},c_{b})$ where $c_{k}\in(0,\infty)$
and $c_{b}\geq0$. The extensions to (\ref{eq:normal_expect1}) and
(\ref{eq:normal_expect2}) are that the scalar $k_{n}$ is replaced
by the diagonal matrix $K_{n}$ and the variable $b_{n}v$ is replaced
by the transformation $b_{n}Q_{n}v$. Now we consider each of the
three cases $\textcircled{\small1}$, $\textcircled{\small2}$ and
$\textcircled{\small3}$.

For case $\textcircled{\small3}$, we have the diagonal matrix $K_{n}=K_{n}^{(1)}$
the diagonal elements of which all go to $\infty$, $Q_{n}=I_{d^{(1)}}$
and $\varphi_{n}=b_{n}^{d-p}\|K_{n}^{(1)}\|$. Then (\ref{eq:normal_expect1})
and (\ref{eq:normal_expect2}) hold by a trivial extension of the
proof of (\ref{eq:old_normal_expect1}) and (\ref{eq:old_normal_expect2})
for its case A. More specifically, let $v^{*}=K_{n}^{(1)}\{A_{n}(t)t-b_{n}v\}-c_{n}$,
(\ref{eq:normal_expect1}) can be rewritten as 
\begin{align*}
 & \int_{\mathbb{R}^{p}}\int_{\mathbb{R}^{d}}P_{l}(\|t\|,\|v\|)g_{1}[C_{n}(t)K_{n}^{(1)}\{A_{n}(t)t-b_{n}v-K_{n}^{-1}c_{n}\}]g_{2}(Qv)\,dvdt\\
= & \|K_{n}^{(1)}\|^{-1}\int_{\mathbb{R}^{p}}\int_{\mathbb{R}^{d}}P_{l}(\|t\|,\|v(v^{*})\|)g_{1}\{C_{n}(t)v^{*}\}g_{2}[Q\{A_{n}(t)t-K_{n}^{(1)^{-1}}(v^{*}+c_{n})\}\}]\,dv^{*}dt,
\end{align*}
then the arguments in the proof of (\ref{eq:old_normal_expect1})
for its case A in \citet[Lemma 7]{li2018asymptotic} can be applied. 

For case $\textcircled{\small1}$, we have $K_{n}=diag(K_{n}^{(2)},K_{n}^{(3)})=I_{d^{(2)}+d^{(3)}}.$
First consider $\lim B_{n}$ is UL bounded. Then $Q_{n}v$ is a bi-jection
to $v$, and (\ref{eq:normal_expect1}) and (\ref{eq:normal_expect2})
hold by a trivial extension of the proof of (\ref{eq:old_normal_expect1})
and (\ref{eq:old_normal_expect2}) for its case B. Next consider $\lim B_{n}=0$.
To simplify notations, assume $C_{n}(t)=I_{d}$. By the definition
of $Q_{n}$, rewrite the arguments of $g_{1}(\cdot)$ in the integral
as follows,
\begin{align*}
 & \|A_{n}(t)t-b_{n}Q_{n}v-c_{n}\|^{2}\\
= & \sum_{i=1}^{2}\|A_{n}^{(i)}(t)t-b_{n}v^{(i)}+c_{n}^{(i)}\|^{2}+\|A_{n}^{(3)}(t)t-b_{n}(Rv^{(2)}+B_{n}v^{(3)})+c_{n}^{(3)}\|^{2}\\
\geq & \|A_{n}^{(1:2)}(t)t-b_{n}v^{(1:2)}+c_{n}^{(1:2)}\|^{2}.
\end{align*}
Then the integral in (\ref{eq:normal_expect1}) is bounded as the
following, 
\begin{align*}
 & \int_{\mathbb{R}^{p}}\int_{\mathbb{R}^{d}}P_{l}(\|t\|,\|v\|)g_{1}\{A_{n}(t)t-b_{n}Q_{n}v-c_{n}\}g_{2}(Qv)\,dvdt\\
\leq & \int_{\mathbb{R}^{p}}\int_{\mathbb{R}^{d}}P_{l}(\|t\|,\|v\|)\bar{g}_{1}\{\|A_{n}^{(1:2)}(t)t-b_{n}v^{(1:2)}+c_{n}^{(1:2)}\|\}g_{2}(Qv)\,dvdt.
\end{align*}
Therefore (\ref{eq:normal_expect1}) and (\ref{eq:normal_expect2})
hold by integrating out $v^{(3)}$ first and then applying case B
of (\ref{eq:old_normal_expect1}) and (\ref{eq:old_normal_expect2}).

For case $\textcircled{\small2}$, we have $d^{(1)}>0$ and $d^{(2)}+d^{(3)}>0$.
For this case, the extension is non-trivial, because part of $K_{n}$
go to $\infty$ and the other part go to $1$, and when $d^{(3)}>0$,
$Q_{n}v$ reduces $v$ to its sub-dimensions $v^{(1:2)}$. To deal
with the heteogeneity, below we decompose the norm in $g_{1}(\cdot)$
to seperate the dominating dimensions of $t$ and $v$ from the dominated
ones, then the target integral can be bounded in the way that $g_{1}(\cdot)$
only contains the dominating dimensions of $t$ and $v$, and $g_{2}(\cdot)$
contains the rest of the dimensions. 

Similar as case $\textcircled{\small1}$, assume $C_{n}(t)=I_{d}$.
First rewrite the arguments of $g_{1}(\cdot)$ in the integral as
follows, 
\begin{align}
 & \|\left(\begin{array}{ccc}
K_{n}^{(1)} & 0 & 0\\
0 & I_{d^{(2)}} & 0\\
0 & 0 & I_{d^{(3)}}
\end{array}\right)\left\{ A_{n}(t)t-b_{n}\left(\begin{array}{ccc}
I_{d^{(1)}} & 0 & 0\\
0 & I_{d^{(2)}} & 0\\
0 & R & B_{n}
\end{array}\right)v\right\} +c_{n}\|^{2}\nonumber \\
= & \|A_{n}^{(11)}(t)t-b_{n}v^{(11)}+k_{n1}^{(1)^{-1}}c_{n}^{(11)}\|^{2}k_{n1}^{(1)^{2}}+\sum_{i=2}^{\tau_{1}}\|A_{n}^{(1i)}(t)t-b_{n}v^{(1i)}+k_{ni}^{(1)^{-1}}c_{n}^{(1i)}\|^{2}k_{ni}^{(1)^{2}}\nonumber \\
 & +\|A_{n}^{(2)}(t)t-b_{n}v^{(2)}+c_{n}^{(2)}\|^{2}+\|A_{n}^{(3)}(t)t-b_{n}Rv^{(2)}-b_{n}B_{n}v_{n}^{(3)}+c_{n}^{(3)}\|^{2}.\label{eq:g1_arguments1}
\end{align}
Let $P_{A}:=(A^{T}A)^{-1}A^{T}$ and $H_{A}:=A(A^{T}A)^{-1}A^{T}$
for a matrix $A$. Using condition (v) and the following orthogonality
decomposition, 
\[
\|At-v\|^{2}=(t-P_{A}v)^{T}A^{T}A(t-P_{A}v)+v^{T}(I-H_{A})v,
\]
the right hand side of (\ref{eq:g1_arguments1}) can be decomposed
as follows,
\begin{align}
 & k_{n1}^{(1)^{2}}\|\{A_{n}^{(11)}(t)^{T}A_{n}^{(11)}(t)\}^{1/2}\{t-P_{A_{n}^{(11)}}(b_{n}v^{(11)}+k_{n1}^{(1)^{-1}}c_{n}^{(11)})\}\|^{2}\nonumber \\
 & +\sum_{i=2}^{\tau_{1}}k_{ni}^{(1)^{2}}\|\{A_{n}^{(1i)}(t)^{T}A_{n}^{(1i)}(t)\}^{1/2}\{t-P_{A_{n}^{(1i)}}(b_{n}v^{(1i)}+k_{ni}^{(1)^{-1}}c_{n}^{(1i)})\}\|^{2}\nonumber \\
 & +\sum_{i=1}^{\tau_{1}}k_{ni}^{(1)^{2}}\|(I_{d_{i}^{(1)}}-H_{A_{n}^{(1i)}})(b_{n}v^{(1i)}+k_{ni}^{(1)^{-1}}c_{n}^{(1i)})\|^{2}\nonumber \\
 & +\|A_{n}^{(2)}(t)t-b_{n}v^{(2)}+c_{n}^{(2)}\|^{2}+\|A_{n}^{(3)}(t)t-b_{n}Rv^{(2)}-b_{n}B_{n}v_{n}^{(3)}+c_{n}^{(3)}\|^{2}\nonumber \\
= & \|\{A_{n}^{(11)}(t)^{T}A_{n}^{(11)}(t)\}^{1/2}(t^{*}-P_{A_{n}^{(11)}}c_{n}^{(11)})\|^{2}+\sum_{i=2}^{\tau_{1}}\|\{A_{n}^{(1i)}(t)^{T}A_{n}^{(1i)}(t)\}^{1/2}(v_{1:p}^{(1i)^{*}}-P_{A_{n}^{(1i)}}c_{n}^{(1i)})\|^{2}\nonumber \\
 & +\sum_{i=1}^{\tau_{1}}\|v_{(p+1):d_{1}^{(i)}}^{(1i)^{*}}+Q_{A_{n}^{(1i)}}^{T}c_{n}^{(1i)}\|^{2}\nonumber \\
 & +\|A_{n}^{(2)}(t)(k_{n1}^{(1)^{-1}}t^{*}-b_{n}v_{1:p}^{(11)})-b_{n}v^{(2)}+c_{n}^{(2)}\|^{2}\nonumber \\
 & +\|A_{n}^{(3)}(t)(k_{n1}^{(1)^{-1}}t^{*}-b_{n}v_{1:p}^{(11)})-b_{n}Rv^{(2)}-b_{n}B_{n}v_{n}^{(3)}+c_{n}^{(3)}\|^{2},\label{eq:g1_arguments2}
\end{align}
where we use the transformations $t^{*}=k_{n1}^{(1)}(t-b_{n}P_{A_{n}^{(11)}}v^{(11)})$,
$v_{1:p}^{(1i)^{*}}=k_{ni}^{(1)}(t-b_{n}P_{A_{n}^{(1i)}}v^{(1i)})$
and $v_{(p+1):d_{1}^{(i)}}^{(1i)^{*}}=k_{ni}^{(1)}b_{n}Q_{A_{n}^{(1i)}}^{T}v^{(1i)}$,
and $Q_{A_{n}^{(1i)}}\in\mathbb{R}^{d_{i}^{(1)}\times(d_{i}^{(1)}-p)}$
is defined by the QR decomposition that $I_{d_{i}^{(1)}}-H_{A_{n}^{(1i)}}=Q_{A_{n}^{(1i)}}Q_{A_{n}^{(1i)}}^{T}$,
i.e.~the columns of $Q_{A_{n}^{(1i)}}$ are orthogonal eigenvectors
and $Q_{A_{n}^{(1i)}}^{T}Q_{A_{n}^{(1i)}}=I_{d_{i}^{(1)}-p}$. For
the transformations above, to obtain its Jacobian, rewrite it as follows,
\begin{align*}
 & \left(\begin{array}{c}
v_{1:p}^{(11)^{*}}\\
v_{(p+1):d_{1}^{(1)}}^{(11)^{*}}\\
t^{*}\\
v_{1:p}^{(12)^{*}}\\
v_{(p+1):d_{2}^{(1)}}^{(12)^{*}}\\
\vdots\\
v_{1:p}^{(1\tau_{1})^{*}}\\
v_{(p+1):d_{\tau_{1}}^{(1)}}^{(1\tau_{1})^{*}}
\end{array}\right)=\left(\begin{array}{c}
-P_{A_{n}^{(11)}}v^{(11)}\\
k_{n1}^{(1)}b_{n}Q_{A_{n}^{(11)}}^{T}v^{(11)}\\
k_{n1}^{(1)}(t-b_{n}P_{A_{n}^{(11)}}v^{(11)})\\
k_{n2}^{(1)}(t-b_{n}P_{A_{n}^{(12)}}v^{(12)})\\
k_{n2}^{(1)}b_{n}Q_{A_{n}^{(12)}}^{T}v^{(12)}\\
\vdots\\
k_{n\tau_{1}}^{(1)}(t-b_{n}P_{A_{n}^{(1\tau_{1})}}v^{(1\tau_{1})})\\
k_{n\tau_{1}}^{(1)}b_{n}Q_{A_{n}^{(1\tau_{1})}}^{T}v^{(1\tau_{1})}
\end{array}\right)\\
 & =\left(\begin{array}{ccccc}
-P_{A_{n}^{(11)}} & 0 & 0 & \dots & 0\\
k_{n1}^{(1)}b_{n}Q_{A_{n}^{(11)}}^{T} & 0 & 0 & \dots & 0\\
-k_{n1}^{(1)}b_{n}P_{A_{n}^{(11)}} & k_{n1}^{(1)}I_{p} & 0 & \dots & 0\\
0 & k_{n2}^{(1)}I_{p} & -k_{n2}^{(1)}b_{n}P_{A_{n}^{(12)}} & \dots & 0\\
0 & 0 & k_{n2}^{(1)}b_{n}Q_{A_{n}^{(12)}}^{T} & \dots & 0\\
\vdots & \vdots & \vdots & \ddots & \vdots\\
0 & k_{n\tau_{1}}^{(1)}I_{p} & 0 & \dots & -k_{n\tau_{1}}^{(1)}b_{n}P_{A_{n}^{(1\tau_{1})}}\\
0 & 0 & 0 & \dots & k_{n\tau_{1}}^{(1)}b_{n}Q_{A_{n}^{(1\tau_{1})}}^{T}
\end{array}\right)\left(\begin{array}{c}
v^{(11)}\\
t\\
v^{(12)}\\
\vdots\\
v^{(1\tau_{1})}
\end{array}\right)\\
\Rightarrow & \left(\begin{array}{c}
v^{(11)^{*}}\\
t^{*}\\
v^{(12)^{*}}\\
\vdots\\
v^{(1\tau_{1})^{*}}
\end{array}\right)=\left(\begin{array}{ccccc}
\left(\begin{array}{cc}
I_{p} & 0\\
0 & k_{n1}^{(1)}b_{n}
\end{array}\right)W_{A_{n}^{(11)}} & 0 & 0 & \dots & 0\\
-k_{n1}^{(1)}b_{n}P_{A_{n}^{(11)}} & k_{n1}^{(1)}I_{p} & 0 & \dots & 0\\
0 & \left(\begin{array}{c}
k_{n2}^{(1)}I_{p}\\
0
\end{array}\right) & k_{n2}^{(1)}b_{n}W_{A_{n}^{(12)}} & \dots & 0\\
\vdots & \vdots & \vdots & \ddots & \vdots\\
0 & \left(\begin{array}{c}
k_{n\tau_{1}}^{(1)}I_{p}\\
0
\end{array}\right) & 0 & \dots & k_{n\tau_{1}}^{(1)}b_{n}W_{A_{n}^{(1\tau_{1})}}
\end{array}\right)\left(\begin{array}{c}
v^{(11)}\\
t\\
v^{(12)}\\
\vdots\\
v^{(1\tau_{1})}
\end{array}\right),
\end{align*}
 where $W_{A_{n}^{(1i)}}=\left(\begin{array}{c}
-P_{A_{n}^{(1i)}}\\
Q_{A_{n}^{(1i)}}^{T}
\end{array}\right)$. Then the Jacobian is given below,
\begin{align*}
D\left(\begin{array}{c}
v^{(11)^{*}}\\
t^{*}\\
v^{(12)^{*}}\\
\vdots\\
v^{(1\tau_{1})^{*}}
\end{array}\right) & /D\left(\begin{array}{c}
v^{(11)}\\
t\\
v^{(12)}\\
\vdots\\
v^{(1\tau_{1})}
\end{array}\right)=\left|\begin{array}{cc}
\left(\begin{array}{cc}
I_{p} & 0\\
0 & k_{n1}^{(1)}b_{n}
\end{array}\right)W_{A_{n}^{(11)}} & 0\\
-k_{n1}^{(1)}b_{n}P_{A_{n}^{(11)}} & k_{n1}^{(1)}I_{p}
\end{array}\right|\left|diag(k_{n2}^{(1)}b_{n}W_{A_{n}^{(12)}},\dots,k_{n\tau_{1}}^{(1)}b_{n}W_{A_{n}^{(1\tau_{1})}})\right|\\
 & =\left|\left(\begin{array}{cc}
I_{p} & 0\\
0 & k_{n1}^{(1)}b_{n}
\end{array}\right)W_{A_{n}^{(11)}}\right|k_{n1}^{(1)^{p}}\prod_{i=2}^{\tau_{1}}(k_{ni}^{(1)}b_{n})^{d_{i}^{(1)}}|W_{A_{n}^{(i\tau_{1})}}|\\
 & =b_{n}^{d^{(1)}-p}|K_{n}^{(1)}|c_{W},\text{ where the scalar }c_{W}=\prod_{i=1}^{\tau_{1}}|W_{A_{n}^{(i\tau_{1})}}|\in(0,\infty).
\end{align*}
Since $k_{n1}^{(1)}\rightarrow\infty$, in (\ref{eq:g1_arguments2})
among the terms containing $t^{*}$, the first term dominates the
others. Then we ignore the dominated terms by using the bound below,
{\small
\begin{align}
\eqref{eq:g1_arguments2} & \geq\left\Vert \left(\begin{array}{ccccc}
\{A_{n}^{(11)}(t)^{T}A_{n}^{(11)}(t)\}^{1/2}\\
 & \{A_{n}^{(12)}(t)^{T}A_{n}^{(12)}(t)\}^{1/2}\\
 &  & \ddots\\
 &  &  & \{A_{n}^{(1\tau_{1})}(t)^{T}A_{n}^{(1\tau_{1})}(t)\}^{1/2}\\
 &  &  &  & I_{d^{(1)}-\tau_{1}p}
\end{array}\right)\right.\nonumber \\
 & \left.\left\{ \left(\begin{array}{c}
t^{*}\\
v_{1:p}^{(12)^{*}}\\
\vdots\\
v_{1:p}^{(1\tau_{1})^{*}}\\
v_{(p+1):d_{1}^{(1)}}^{(11)^{*}}\\
\vdots\\
v_{(p+1):d_{\tau_{1}}^{(1)}}^{(1\tau_{1})^{*}}
\end{array}\right)+\left(\begin{array}{c}
P_{A_{n}^{(11)}}c_{n}^{(11)}\\
\vdots\\
P_{A_{n}^{(1\tau_{1})}}c_{n}^{(1\tau_{1})}\\
Q_{A_{n}^{(11)}}^{T}c_{n}^{(11)}\\
\vdots\\
Q_{A_{n}^{(1\tau_{1})}}^{T}c_{n}^{(1\tau_{1})}
\end{array}\right)\right\} \right\Vert ^{2}\label{eq:g_arguments_bound1}
\end{align}
}
In other words, we only keep $t^{*}$ and $v^{(1)^{*}}\backslash v_{1:p}^{(11)^{*}}$
in $g_{1}(\cdot)$ and leave the rest to $g_{2}(\cdot)$. So for $g_{2}(Qv)$,
we use the bound below, 
\begin{align}
 & \|Qv\|^{2}=\sum_{i=1}^{\tau_{1}}\|Q^{(1i)}v^{(1i)}\|^{2}+\|Q^{(2:3)}v^{(2:3)}\|^{2}\geq\|Q^{(11)}W_{A_{n}^{(11)},1:p}^{-1}v_{1:p}^{(11)^{*}}\|^{2}+\|Q^{(2:3)}v^{(2:3)}\|^{2}.\label{eq:g_arguments_bound2}
\end{align}
Then the integral in (\ref{eq:normal_expect1}) is bounded as the
following,
\begin{align}
 & \int_{\mathbb{R}^{p}}\int_{\mathbb{R}^{d}}P_{l}(\|t\|,\|v\|)b_{n}^{d^{(1)}-p}|K_{n}^{(1)}|g_{1}[K_{n}\{A_{n}(t)t-b_{n}Q_{n}v\}-c_{n}\}]g_{2}(Qv)\,dvdt\nonumber \\
\leq & \int_{\mathbb{R}^{p}}\int_{\mathbb{R}^{d}}P_{l}^{*}(\|t\|,\|v^{(1)^{*}}\|,\|v^{(2)}\|)g_{1}^{*}(t^{*},v_{(p+1):d_{1}^{(1)}}^{(11)^{*}},v^{(12)^{*}},\dots,v^{(1\tau)^{*}})g_{2}^{*}(v_{1:p}^{(11)},v^{(2)},v^{(3)})\,dv^{(1)^{*}}dt^{*}dv^{(2)}dv^{(3)},\label{eq:normal_expect_bound}
\end{align}
where $P_{l}^{*}(\cdot)$, $g_{1}^{*}(\cdot)$ and $g_{2}^{*}(\cdot)$
are obtained by the transformations of variables and the bounds in
(\ref{eq:g_arguments_bound1}) and (\ref{eq:g_arguments_bound2}).
It can be seen that the right hand side of (\ref{eq:normal_expect_bound})
is a trivial extension of the integral in (\ref{eq:old_normal_expect1})
where $k_{n}=1$, $C_{n}(t)=I_{d}$ and $b_{n}=0$. Note that the
above derivations are not affected by letting either $d^{(3)}>0$
or $d^{(3)}=0$. Therefore by \citet[Lemma 7]{li2018asymptotic}, for case $\textcircled{\small2}$
(\ref{eq:normal_expect1}) holds and (\ref{eq:normal_expect2}) holds
similarly. 

\end{proof}

\begin{proof}[Proof of Lemma \ref{lem:normal_expect_expans}]

This lemma extends \citet[Lemma 5]{li2018asymptotic}. First consider $\tpi_{B_{\delta}}(1)$.
With the transformation $t=t(\theta)$, 
\begin{align}
\tpi_{B_{\delta}}(1) & =a_{n,\varepsilon}^{-p}\int_{t(B_{\delta})}\int_{\mathbb{R}^{d}}\pi(\theta_{0}+a_{n,\varepsilon}^{-1}t)\ftil_{n}(s_{{\rm obs}}+\varepsilon_{n}v\mid\theta_{0}+a_{n,\varepsilon}^{-1}t)K(v)\,dvdt.\label{eq:norm_expand1}
\end{align}
We can obtain an expansion of $\tpi_{B_{\delta}}(1)$ by expanding
$\ftil_{n}(s_{{\rm obs}}+\varepsilon_{n}v\mid\theta_{0}+a_{n,\varepsilon}^{-1}t)K(v)$
as follows. Applying a Taylor expansion to $s(\theta_{0}+a_{n,\varepsilon}^{-1}t)$
and $A(\theta_{0}a_{n,\varepsilon}^{-1}t)^{-1/2}$ gives 
\begin{align}
 & \ftil_{n}(s_{{\rm obs}}+\varepsilon_{n}v\mid\theta_{0}+a_{n,\varepsilon}^{-1}t)K(v)\nonumber \\
 & =|a_{n,diag}|N[a_{n,diag}\{s_{\text{obs}}+\varepsilon_{n}v-s(\theta_{0}+a_{n,\varepsilon}^{-1}t)\};0,A(\theta_{0}+a_{n,\varepsilon}^{-1}t)]K(v)\nonumber \\
 & =\frac{|a_{n,diag}|}{|A(\theta_{0}+a_{n,\varepsilon}^{-1}t)|^{1/2}}N\left(\Big\{ A(\theta_{0})^{-1/2}+a_{n,\varepsilon}^{-1}r_{A}(t,\epsilon_{2})\Big\}\right.\nonumber \\
 & \cdot\Big[\left.A(\theta_{0})^{1/2}W_{{\rm obs}}+a_{n,diag}\varepsilon_{n}v-a_{n,diag}a_{n,\varepsilon}^{-1}\{Ds(\theta_{0})+a_{n,\varepsilon}^{-1}r_{s}(t,\epsilon_{1})\}t\Big];0,I_{d}\right)K(v),\label{eq:norm_expand2}
\end{align}

\noindent where $r_{s}(t,\epsilon_{1})$ is the $d\times p$ matrix whose $i$th
row is $t^{T}Hs_{i}\{\theta_{0}+\epsilon_{1}(t)\}$, $r_{A}(t,\epsilon_{2})$
is the $d\times d$ matrix $\sum_{k=1}^{p}\frac{d}{d\theta_{k}}A\{\theta_{0}+\epsilon_{2}(t)\}^{-1/2}t_{k}$,
and $\epsilon_{1}(t)$ and $\epsilon_{2}(t)$ are from the remainder
terms of the Taylor expansions and satisfy $\|\epsilon_{1}(t)\|\leq\delta$
and $\|\epsilon_{2}(t)\|\leq\delta$. If the further expansion in
the proof of \citet[Lemma 5]{li2018asymptotic} was conducted, it is would give a
formal leading term proportional to $g_{n}(t,v)$. However, it is
not straightforward to integrate $g_{n}(t,v)$ because some variance
terms of the normal density in $g_{n}(t,v)$ diverge to $\infty$
and some converge to $0$. The trick is to discuss the further expansion
for several cases, depending on the rate of $\varepsilon_{n}$. For
different rate of $\varepsilon_{n}$, $g_{n}(t,v)$ is subjected to
different variable transformation, then the target integral can be
written into the form of integrals in Lemma \ref{lem:normal_expect_bounded}
of its different cases, and the results of orders therein can be applied. 

Firstly, consider $\varepsilon_{n}=O(a_{\tau n}^{-1})$, then $a_{n,\varepsilon}=a_{k_{\varepsilon}n}$
with $1\leq k_{\varepsilon}\leq\tau$, $d^{(2)}+d^{(3)}>0$ and $a_{k_{\varepsilon}n}\varepsilon_{n}\rightarrow0$
or $1$. Recall the transformation $v'$ defined in the 'Notations'
subsection and the following relationships between $v$ and $v'$
and between $s$ and $s'$, $a_{n,diag}v=Q_{v2}^{-1}a'_{n,diag}v'$
and $a_{n,diag}s=Q_{v2}^{-1}a'_{n,diag}s'$. Let 
\[
g_{n}'(t,v'):=N\{Ds'(\theta_{0})t;a_{n,\varepsilon}\varepsilon_{n}v'+A_{\theta_{0},1/2}^{'}W_{\text{obs}},A_{\theta_{0},1/2}^{'}A_{\theta_{0},1/2}^{'T}\}K\{v(v')\},
\]
\noindent where $A_{\theta_{0},1/2}^{'}:=a_{k_{\varepsilon}n}a_{n,diag}^{'-1}Q_{v2}A(\theta_{0})^{1/2}$,
and $g_{n}'(t,v';C_{1},C_{2})$ be $g_{n}'(t,v')$ with $A(\theta_{0})^{-1/2}$
in the normal variance and $Ds'(\theta_{0})$ replaced by $A(\theta_{0})^{-1/2}+C_{2}$
and $Ds'(\theta_{0})+C_{1}$ respectively. It is easy to show that
$g(t,v)dv=g_{n}'(t,v')dv'$. The benefit of using the transformation
$v'$ is that, in the variance matrix of the normal density in $g_{n}'(t,v')$,
$a_{k_{\varepsilon}n}a_{n,diag}^{'-1}=diag(a_{n,diag}^{(1)^{-1}}a_{k_{\varepsilon}n},I_{d^{(2)}+d^{(3)}})$
the diagonal terms of which converge to either $0$ or $1$. Under
the transformation, (\ref{eq:normal_expect2}) is rewritten as 
\begin{align*}
 & \ftil_{n}(s_{{\rm obs}}+\varepsilon_{n}v\mid\theta_{0}+a_{n,\varepsilon}^{-1}t)K(v)\\
 & =\frac{|a_{n,diag}|}{|A(\theta_{0}+a_{n,\varepsilon}^{-1}t)|^{1/2}}N\left(\Big\{ A(\theta_{0})^{-1/2}+a_{n,\varepsilon}^{-1}r_{A}(t,\epsilon_{2})\Big\}\right.\\
 & \left.\cdot\Big[A(\theta_{0})^{1/2}W_{{\rm obs}}+a_{n,\varepsilon}^{-1}Q_{v2}^{-1}a'_{n,diag}\{a_{n,\varepsilon}\varepsilon_{n}v'-(Ds'(\theta_{0})+a_{n,\varepsilon}^{-1}r_{s'}(t,\epsilon_{1}))t\}\Big];0,I_{d}\right)K(v).
\end{align*}
By another Taylor expansion, the above can be further expanded into
the following 
\begin{align}
 & \ftil_{n}(s_{{\rm obs}}+\varepsilon_{n}v\mid\theta_{0}+a_{n,\varepsilon}^{-1}t)K(v)\,dv\nonumber \\
= & a_{n,\varepsilon}^{d}\frac{|A(\theta_{0})|^{1/2}}{|A(\theta_{0}+a_{n,\varepsilon}^{-1}t)|^{1/2}}[g_{n}'(t,v')+a_{n,\varepsilon}^{-1}P_{3}(t,v')g_{n}'\{t,v';e_{n1}r_{s}(t,\epsilon_{1}),e_{n1}r_{A}(t,\epsilon_{1})\}]\,dv',\label{eq:norm_expand3}
\end{align}
where $P_{3}(t,v')$ is the function 
\begin{align*}
 & \frac{1}{2|A(\theta_{0})^{-1/2}+e_{n1}r_{A}(t,\epsilon_{2})|}\\
 & \times\left.\frac{d}{dx}\left\Vert \left\{ A(\theta_{0})^{-1/2}+xr_{A}(t,\epsilon_{2})\right\} \left[A(\theta_{0})^{1/2}W_{{\rm \text{obs}}}+a_{n,\varepsilon}^{-1}Q_{v2}^{-1}a'_{n,diag}\{a_{n,\varepsilon}\varepsilon_{n}v'-(Ds'(\theta_{0})+xr_{s}(t,\epsilon_{1}))t\}\right]\right\Vert ^{2}\right|_{x=e_{n1}},
\end{align*}
and $e_{n1}$ is from the remainder term of Taylor expansion and satisfies $|e_{n1}|\leq a_{n}^{-1}$ . Since $\|e_{n1}t\|\leq\delta$ and $r_{s}(t,\epsilon_{1})$ and $r_{A}(t,\epsilon_{2})$ belong $\mathbb{P}_{1,t}$, this $P_{3}(t,v)$ belongs to $\mathbb{P}_{3,(t,v)}$. Furthermore, since $r_{s}(t,\epsilon_{1})$ and $r_{A}(t,\epsilon_{2})$ have no constant term, for any small $\sigma$, $e_{n1}r_{s}(t,\epsilon_{1})$ and $e_{n1}r_{A}(t,\epsilon_{2})$ can be bounded by $\sigma I_{d}$ and $\sigma I_{p}$ uniformly in $n$ and $t$, if $\delta$ is small enough. 

Secondly, consider $\varepsilon_{n}=\omega(a_{\tau n}^{-1})$, then
$a_{n,\varepsilon}=\varepsilon_{n}^{-1}$ and $a_{n,diag}\varepsilon_{n}\rightarrow\infty$.
Let $v^{*}=v^{*}(v):=A(\theta_{0})^{1/2}W_{\text{obs}}+a_{n,diag}\varepsilon_{n}v-a_{n,diag}\varepsilon_{n}Ds(\theta_{0}+\epsilon_{1}(t))t$,
where $\epsilon_{1}(t)$ is the same as that defined in the case of
$\varepsilon_{n}=O(a_{\tau n}^{-1})$. Under the transformation, (\ref{eq:normal_expect2})
is rewritten as
\begin{align*}
 & \ftil_{n}(s_{{\rm obs}}+\varepsilon_{n}v\mid\theta_{0}+a_{n,\varepsilon}^{-1}t)K(v)\\
 & =\frac{|a_{n,diag}|}{|A(\theta_{0}+\varepsilon_{n}t)|^{1/2}}N\left[\Big\{ A(\theta_{0})^{-1/2}+\varepsilon_{n}r_{A}(t,\epsilon_{2})\Big\}\left\{ v^{*}-a_{n,diag}\varepsilon_{n}^{2}r_{s}(t,\epsilon_{1})t\right\} ;0,I_{d}\right]K(v).
\end{align*}
Let $g_{n}^{*}(t,v^{*};\tau_{1})$ be the function 
\[
g_{n}^{*}(t,v^{*};\tau_{1})=N[v^{*};0,\{A(\theta_{0})^{-1/2}+\tau_{1}\}^{-2}]K\{Ds(\theta_{0}+\epsilon_{1}(t))t+a_{n,diag}^{-1}\varepsilon_{n}^{-1}v^{*}-a_{n,diag}^{-1}\varepsilon_{n}^{-1}A(\theta_{0})^{-1/2}W_{\text{obs}}\},
\]
and $g_{n}^{*}(t,v^{*})=g_{n}^{*}(t,v^{*};0)$. It is easy to show
that $g_{n}(t,v)dv=g_{n}^{*}(t,v^{*})dv^{*}$. 
Denote a $k_{1}\times k_{2}$ matrix with element being $P_{l}(t)$
by $P_{l}^{(k_{1}\times k_{2})}(t)$. Then by applying a Taylor expansion
to the normal density in the expansion above, 
\begin{align}
 & \ftil_{n}(s_{{\rm obs}}+\varepsilon_{n}v\mid\theta_{0}+\varepsilon_{n}t)K(v)dv\nonumber \\
= & \frac{\varepsilon_{n}^{-d}|A(\theta_{0})|^{1/2}}{|A(\theta_{0}+\varepsilon_{n}t)|^{1/2}}\Big[g_{n}^{*}(t,v^{*})+\varepsilon_{n}v^{*T}P_{1}^{(d\times d)}(t)v^{*}g_{n}^{*}(t,v^{*})+\varepsilon_{n}^{2}P_{4}(t,v^{*})g_{n}^{*}\{t,v^{*};e_{n2}r_{A}(t,\epsilon_{4})\}\Big]dv^{*},\label{eq:norm_expand4}
\end{align}
where $P_{1}^{(d\times d)}(t)$ is the function $-A(\theta_{0})^{-1/2}r_{A}(t,\epsilon_{4})$,
$e_{n2}$ is from the remainder term of the Taylor expansion and satisfies
$|e_{n2}|\leq\varepsilon_{n}$, and $P_{4}(t,v^{*})$ is a linear
combination of $\{d\rho(w)/dw\}^{2}$ and $d^{2}\rho(w)/dw^{2}$ at
$w=e_{n2}$ with $\rho(w)$ being the function 
\[
\Big\|\Big\{ A(\theta_{0})^{-1/2}+wr_{A}(t,\epsilon_{4})\Big\} v^{*}\Big\|^{2}.
\]
Obviously elements of $P_{2}^{(d\times1)}(t)$ and $P_{1}^{(d\times d)}(t)$
belong to $\mathbb{P}_{2,t}$ and $\mathbb{P}_{1,t}$ respectively.
Since $\|e_{n2}t\|\leq\delta$, the function $P_{4}(t,v^{*})$ belongs
to $\mathbb{P}_{4,(t,v^{*})}$ and, similar to before, $e_{n2}r_{A}(t,\epsilon_{4})$
can be bounded by $\sigma I_{d}$ and $\sigma I_{p}$ uniformly in
$n$ and $t$ for any small $\sigma$, if $\delta$ is small enough.
Besides being an extension, the expansion (\ref{eq:norm_expand4})
above is different from \citet[equation 8]{li2018asymptotic} in that it gives sharper
remainder terms. Specifically, in $v^{*}$ above, $Ds(\theta_{0}+\epsilon_{1}(t))$
is used so orders of the remainder terms above are $O_{n}(\varepsilon_{n})+O_{n}(\varepsilon_{n}^{2})$,
which in contrast, $v^{*}$ in \citet[equation 8]{li2018asymptotic} uses $Ds(\theta_{0})$,
which leads to the slower remainder terms $O_{n}(a_{n}\varepsilon_{n}^{2})+O_{n}(a_{n}^{2}\varepsilon_{n}^{4})$. 

For $\pi(\theta_{0}+a_{n,\varepsilon}^{-1}\tau)$ in the integral
of $\widetilde{\pi}_{B_{\delta}}(1)$, its expansion is the same as
that in the proof of \citet[Lemma 5]{li2018asymptotic}. Combining (\ref{eq:norm_expand3}),
(\ref{eq:norm_expand4}) and the expansion of $\pi(\theta_{0}+a_{n,\varepsilon}^{-1}\tau)$,
it can be seen that the leading term of $\widetilde{\pi}_{B_{\delta}}(1)$
is $a_{n,\varepsilon}^{d-p}\pi(\theta_{0})\int_{t(B_{\delta})\times\mathbb{R}^{d}}g_{n}(t,v)\,dtdv$.
The remainder terms are given in the following,
\begin{align}
 & a_{n,\varepsilon}^{p-d}\tpi_{B_{\delta}}(1)-\pi(\theta_{0})\int_{t(B_{\delta})\times\mathbb{R}^{d}}g_{n}(t,v)\,dtdv\nonumber \\
= & a_{n,\varepsilon}^{-1}\int_{t(B_{\delta})\times\mathbb{R}^{d}}|A(\theta_{0})|^{1/2}D\frac{\pi(\theta_{0}+\epsilon_{5})}{|A(\theta_{0}+\epsilon_{5})|^{1/2}}tg_{n}(t,v)\,dtdv\nonumber \\
 & +a_{n,\varepsilon}^{-1}\int_{t(B_{\delta})\times\mathbb{R}^{d}}P_{3}(t,v')g_{n}'\{t,v';e_{n1}r_{s}(t,\epsilon_{1}),e_{n1}r_{A}(t,\epsilon_{2})\}\,dv'dt\mathbbm{1}_{\{\varepsilon_{n}=O(a_{\tau n}^{-1})\}}\nonumber \\
 & +\varepsilon_{n}\int_{t(B_{\delta})\times\mathbb{R}^{d}}v^{*T}P_{1}^{(d\times d)}(t)v^{*}g_{n}^{*}(t,v^{*})\,dv^{*}dt\mathbbm{1}_{\{\varepsilon_{n}=\omega(a_{\tau n}^{-1})\}}\nonumber \\
 & +\varepsilon_{n}^{2}\int_{t(B_{\delta})\times\mathbb{R}^{d}}P_{4}(t,v^{*})g_{n}^{*}\{t,v^{*};e_{n2}r_{A}(t,\epsilon_{4})\}\,dv^{*}dt\mathbbm{1}_{\{\varepsilon_{n}=\omega(a_{\tau n}^{-1})\}},\label{eq:norm_expand5}
\end{align}
where $P_{3}(t,v')$, $P_{1}^{(d\times d)}(t)$ and $P_{4}(t,v^{*})$
are products of $\pi(\theta_{0}+a_{n,\varepsilon}^{-1}t)/|A(\theta_{0}+a_{n,\varepsilon}^{-1}t)|^{1/2}$
and corresponding terms in expansions (\ref{eq:norm_expand3}), (\ref{eq:norm_expand4}).
The four remainder terms above are extensions of the remainder terms
in \citet[equation 10]{li2018asymptotic}, and similar to the arguments therein, by
applying Lemma \ref{lem:normal_expect_bounded} orders of the integrals
in the remainder terms above can be obtained, which is done below. 

In (\ref{eq:norm_expand5}), recall that $g_{n}'(t,v')dv'=g_{n}^{*}(t,v^{*})dv^{*}=g_{n}(t,v)dv$.
Both $g_{n}'(t,v')$ and $g_{n}^{*}(t,v^{*})$ can be written in the
form of the integrands in Lemma \ref{lem:normal_expect_bounded} as
following: When $\varepsilon_{n}=O(a_{\tau n}^{-1}),$
\begin{align}
 & |A(\theta_{0})|^{-1/2}\frac{a_{n,\varepsilon}^{d}}{|a_{n,diag}|}g_{n}'(t,v')dv'\nonumber \\
 & =|A(\theta_{0})|^{-1/2}\frac{a_{n,\varepsilon}^{d}}{|a_{n,diag}|}g_{n}(t,v)dv\nonumber \\
 & =N\left[\underbrace{A(\theta_{0})^{-1/2}Q_{v2}^{-1}}_{C_{n}(t)}\underbrace{a_{n,\varepsilon}^{-1}a_{n,diag}'}_{K_{n}}\left\{ \underbrace{-Ds'(\theta_{0})}_{A_{n}(t)}t+\underbrace{a_{n,\varepsilon}\varepsilon_{n}}_{b_{n}}\underbrace{Q_{v1}}_{Q_{n}}v\right\} +\underbrace{W_{\text{obs}}}_{c_{n}};0,I_{d}\right]K(v)dv,\label{eq:g_prime_integrand}
\end{align}
where for Lemma \ref{lem:normal_expect_bounded}, notations in underbraces,
$R=R_{33}^{-1}R_{32}$ and $B_{n}=a_{n,diag}^{(3)}a_{k_{\varepsilon},n}^{-1}$,
are used, and similar for $g_{n}'\{t,v';e_{n1}r_{s}(t,\epsilon_{1}),e_{n1}r_{A}(t,\epsilon_{2})\}$;
when $\varepsilon_{n}=\omega(a_{\tau n}^{-1})$,
\begin{align}
g_{n}^{*}(t,v^{*})dv^{*} & =g_{n}(t,v)dv\nonumber \\
 & =K\left\{ \underbrace{Ds(\theta_{0}+\epsilon_{1}(t)t)}_{A_{n}(t)}t+\underbrace{(a_{\tau n}\varepsilon_{n})^{-1}}_{b_{n}}\underbrace{a_{n,diag}^{-1}a_{\tau n}}_{Q_{n}}v^{*}-\underbrace{a_{n,diag}^{-1}\varepsilon_{n}^{-1}A(\theta_{0})^{-1/2}W_{\text{obs}}}_{c_{n}}\right\} N\{v^{*};0,A(\theta_{0})\}dv^{*}.\label{eq:g_star_integrand}
\end{align}
and similar for $g_{n}^{*}\{t,v^{*};e_{n2}r_{A}(t,\epsilon_{4})\}$.
Recall that $\|\epsilon_{1}(t)\|\leq\delta$, and $e_{n1}r_{s}(t,\epsilon_{1})$
and $e_{n1}r_{A}(t,\epsilon_{2})$ can be bounded by $\sigma I_{d}$
and $\sigma I_{p}$ uniformly in $n$ and $t$ for any $\sigma$ if
$\delta$ is small enough. It is easy to verify that condition (i)--(v)
of Lemma \ref{lem:normal_expect_bounded} are satisfied, where $g_{1}(\cdot)$
and $g_{2}(\cdot)$ are the standard normal density and $K(v)$ rescaled
to have identity covariance. Below the three cases of Lemma \ref{lem:normal_expect_bounded}
are applied to obtain the orders of the remainder terms for various
values of $\varepsilon_{n}$.

When $\varepsilon_{n}=o(a_{1n}^{-1})$, $a_{n,\varepsilon}=a_{1n}$,
$d^{(1)}=0$, $d^{(2)}=d_{s,1}$ and $k_{\varepsilon}=1$. For notations
of Lemma \ref{lem:normal_expect_bounded}, we have $K_{n}=I_{d}$,
$b_{n}=a_{1n}\varepsilon_{n}\rightarrow0$ and $B_{n}=a_{n,diag}^{(3)}a_{1n}^{-1}\rightarrow0$,
so it's case $\textcircled{\small1}$ and $\varphi_{n}=1$. 

When $\varepsilon_{n}=\Omega(a_{1n}^{-1})$ and $\varepsilon_{n}=O(a_{\tau n}^{-1})$,
$a_{n,\varepsilon}=a_{k_{\varepsilon}n}$ for $k_{\varepsilon}\in\{1,2,\dots,\tau\}$,
$\lim_{n}a_{n,\varepsilon}\varepsilon_{n}$ is either $0$ or a positie
constant. For notations of Lemma \ref{lem:normal_expect_bounded},
we have 
\begin{align*}
 & K_{n}=a_{n,\varepsilon}^{-1}a_{n,diag}'=\underbrace{I_{d^{(2)}+d^{(3)}}}_{d^{(1)}=0}\mathbbm{1}_{\{k_{\varepsilon}=1\}}+\underbrace{\left(\begin{array}{cc}
a_{n,diag}^{(1)}a_{k_{\varepsilon}n}^{-1} & 0\\
0 & I_{d^{(2)}+d^{(3)}}
\end{array}\right)}_{d^{(1)},d^{(2)},d^{(3)}>0}\mathbbm{1}_{\{1<k_{\varepsilon}<\tau\}}+\underbrace{\left(\begin{array}{cc}
a_{n,diag}^{(1)}a_{\tau n}^{-1} & 0\\
0 & I_{d^{(2)}}
\end{array}\right)}_{d^{(1)},d^{(2)}>0,d^{(3)}=0}\mathbbm{1}_{\{k_{\varepsilon}=\tau\}},\\
 & b_{n}=a_{k_{\varepsilon}n}\varepsilon_{n}\rightarrow0\mathbbm{1}_{\{\varepsilon_{n}=o(a_{k_{\varepsilon}n})\}}+c\mathbbm{1}_{\{\varepsilon_{n}=\Theta(a_{k_{\varepsilon}n})\}}\text{ where }c\geq0,\\
 & B_{n}=a_{n,diag}^{(3)}a_{k_{\varepsilon}n}^{-1}\rightarrow0\text{ when }k_{\varepsilon}<\tau,\text{ and }Q_{n}=I_{d}\text{ when }k_{\varepsilon}=\tau.
\end{align*}

When $k_{\varepsilon}=1$, $d^{(1)}=0$, so it is case $\textcircled{\small1}$
and $\varphi_{n}=1$. 

When $1<k_{\varepsilon}\leq\tau$, it is case $\textcircled{\small2}$.
By the expression of $K_{n}$, $\varphi_{n}=b_{n}^{p-d^{(1)}}|K_{n}^{(1)}|^{-1}=(a_{k_{\varepsilon}n}\varepsilon_{n})^{p}/(|a_{n,diag}^{(1)}|\varepsilon_{n}^{d^{(1)}})$.

When $\varepsilon_{n}=\omega(a_{\tau n}^{-1})$, $a_{n,\varepsilon}=\varepsilon_{n}^{-1}$.
From (\ref{eq:g_star_integrand}), we can obtain values of notations
of Lemma \ref{lem:normal_expect_bounded} as follows: $b_{n}=(a_{\tau n}\varepsilon_{n})^{-1}\rightarrow0$;
$K_{n}=I_{d}$ from which $d^{(1)}=0$, $d^{(2)}>0$ and $d^{(3)}>0$;
by a change of variable indices, $Q_{n}=a_{n,diag}^{-1}a_{\tau n}$
can have the form in Lemma \ref{lem:normal_expect_bounded} with $B_{n}\rightarrow0$
and $R=0$. So it is case $\textcircled{\small1}$ and $\varphi_{n}=1$. 

Combining results above, 
\begin{align*}
\int_{\mathbb{R}^{p}\times\mathbb{R}^{d}}g_{n}(t,v)\,dtdv= & \begin{cases}
|A(\theta_{0})|^{1/2}|a_{n,diag}|a_{n,\varepsilon}^{-d}\Theta_{p}(1)=\Theta_{p}(c_{n,\varepsilon}), & \text{when }\varepsilon_{n}=O(a_{1n}^{-1}),\\
|A(\theta_{0})|^{-1/2}\frac{a_{n,\varepsilon}^{d}}{|a_{n,diag}|}\Theta_{p}\left\{ \frac{(a_{k_{\varepsilon}n}\varepsilon_{n})^{p}}{|a_{n,diag}^{(1)}|\varepsilon_{n}^{d^{(1)}}}\right\} =\Theta_{p}(c_{n,\varepsilon}), & \text{when }\varepsilon_{n}=\omega(a_{1n}^{-1})\text{ and }O(a_{\tau n}^{-1}).\\
\Theta_{p}(1), & \text{when }\varepsilon_{n}=\omega(a_{\tau n}^{-1}),
\end{cases}
\end{align*}
Therefore (ii) in the lemma holds. For any constant $M>0$, by upper-bounding
$c_{n,\varepsilon}^{-1}\int_{\{\|t\|>M\}\times\mathbb{R}^{d}}g_{n}(t,v)\,dtdv$
in the same way as (\ref{eq:normal_expect_bound}) and decompose the
transformed integration area in the right hand side into suitable
set, it is straightforward to see that the right hand side is $o_{p}(1)$
as $M\rightarrow\infty$. Therefore (iii) in the lemma hold. Since
in (\ref{eq:norm_expand5}), the first two remainders are $O_{p}(a_{n,\varepsilon}c_{n,\varepsilon})$
and the last two reminders are $O_{p}(\varepsilon_{n})$ and $O_{p}(\varepsilon_{n}^{2})$,
(i) in the lemma holds. The other equations hold by arguments similar
to the corresponding part in the proof of \citet[Lemma 5]{li2018asymptotic}. 

\end{proof}

\printbibliography[heading=subbibintoc]

\end{document}

%% file: Symbols_pream.tex
\def\ftil{\widetilde{f}}

\def\bv{\pmb{v}}

\def\fL{\mathfrak{L}}

\def\tpi{\widetilde{\pi}}

\def\tPi{\widetilde{\Pi}}

\def\ttheta{\widetilde{\theta}}

\def\argmax{\mbox{argmax}}
\def\Var{\mbox{Var}}